\documentclass[a4paper]{article}

\usepackage[l2tabu, orthodox]{nag}
\usepackage[utf8]{inputenc}
\usepackage[pdfusetitle]{hyperref}
\usepackage{microtype}

\usepackage{mathtools}

\mathtoolsset{
  showonlyrefs=true 
}

\usepackage{amsmath,amssymb,amsthm}

\newtheorem{lemma}{Lemma}
\newtheorem{theorem}{Theorem}

\newtheorem{corollary}{Corollary}
\theoremstyle{definition}

\theoremstyle{remark}

\usepackage{kbordermatrix}
\usepackage{thm-restate}
\usepackage{cleveref}

\renewcommand{\kbldelim}{(}
\renewcommand{\kbrdelim}{)}

\newif\iflong
\newif\ifshort
\longtrue

\iflong
\else
\shorttrue
\fi

\usepackage{fullpage}

\usepackage{enumerate} 

\usepackage[noend]{algpseudocode}
\usepackage{algorithm}

\usepackage{xspace}

\newcommand{\N}{\ensuremath{\mathbb N}\xspace}

\newcommand{\F}{\ensuremath{\mathbb F}\xspace}

\newcommand{\cI}{\ensuremath{\mathcal I}\xspace}
\newcommand{\cB}{\ensuremath{\mathcal B}\xspace}
\newcommand{\cC}{\ensuremath{\mathcal C}\xspace}
\newcommand{\cP}{\ensuremath{\mathcal P}\xspace}

\newcommand{\cF}{\ensuremath{\mathcal F}\xspace}
\newcommand{\cT}{\ensuremath{\mathcal T}\xspace}
\newcommand{\cS}{\ensuremath{\mathcal S}\xspace}

\newcommand{\bD}{\ensuremath{\mathbf D}\xspace}

\newcommand{\Oh}{\mathcal{O}}

\newcommand{\MM}{\mathsf{MM}}

\DeclareMathOperator{\supp}{supp}
\DeclareMathOperator{\rank}{rank}

\usepackage[usenames,dvipsnames]{color}

\usepackage{tikz}
\usetikzlibrary{arrows.meta,calc} 
\usetikzlibrary{decorations.pathreplacing}

\tikzset{vertex/.style={draw, circle, fill=white, inner sep=2.4pt}}
\tikzset{bvertex/.style={draw, circle, fill=black, inner sep=2.4pt}}
\tikzset{>={Latex[length=3mm, width=2mm]}}

\usepackage{todonotes}

\DeclareMathOperator{\Pf}{Pf}

\newcommand{\flowervertex}[2]{
  \foreach \angle in {30,120,...,300}
      \draw[fill=white] (#1, #2) -- ++(\angle:0.3) arc (\angle:\angle+30:0.5) -- cycle;
      \node[bvertex] (_) at (#1, #2) {};
}

\title{FPT algorithms over linear delta-matroids with applications}

\author{Eduard Eiben\thanks{ Dept. of Computer Science, Royal
    Holloway, University of London, \texttt{Eduard.Eiben@rhul.ac.uk}}
  \and
  Tomohiro Koana\thanks{Utrecht University \& Research Institute for Mathematical Sciences, Kyoto University, \texttt{tomohiro.koana@gmail.com}}
  \and
  Magnus Wahlström\thanks{Dept. of Computer Science, Royal Holloway,
    University of London, \texttt{Magnus.Wahlstrom@rhul.ac.uk}}}

\date{}

\begin{document}
\maketitle

\begin{abstract}
  Matroids, particularly linear matroids, have been a powerful tool for
  applications in parameterized complexity, both for algorithms and
  kernelization.   In particular, they have been instrumental in speeding up or replacing dynamic programming.
  Delta-matroids are a generalization of matroids that further
  encapsulates structures such as non-maximum matchings in general
  graphs and various path-packing and topological structures.
  There is also a notion of linear delta-matroids
  (represented by skew-symmetric matrices) which carries significant
  expressive power and enables powerful algorithms. 
  We investigate parameterized complexity aspects of problems defined over linear delta-matroids,
  or with delta-matroid constraints.
  Our initial analysis of basic intersection and packing problems
  reveals a different complexity landscape compared to the more familiar matroid case. 
  In particular, there is a stark contrast in complexity between
  the \emph{cardinality} parameter $k$ and the \emph{rank} parameterization $r$.
  For example, finding an intersection of size $k$ of three linear delta-matroids
  is W[1]-hard when parameterized by $k$, while far more general problems
  (such as finding a set packing of size $k$ that is feasible in a
  given linear delta-matroid) are FPT when parameterized by
  the rank $r$ of the delta-matroid.
  In fact, we extend the recent \emph{determinantal sieving}
  procedure of Eiben, Koana and Wahlström (SODA 2024)
  into a process that sieves a given polynomial for a monomial whose support is
  feasible in a given linear delta-matroid, parameterized by $r$.
  This is a direct generalization of determinantal sieving. 
  
  Second, we investigate a curious class of problems that turns out
  to be FPT parameterized by $k$, even on delta-matroids of unbounded rank.
  We begin with \textsc{Delta-matroid Triangle Cover} -- find a
  feasible set of size $k$ that can be covered by a vertex-disjoint
  packing of triangles (i.e., sets of size 3) out of a given triangle collection.
  For example, this allows us to find a packing of $K_3$'s and $K_2$'s
  in a graph with a maximum number of edges, parameterized above the
  matching number of the graph. As applications, we settle
  questions on the FPT status of \textsc{Cluster Subgraph} and
  \textsc{Strong Triadic Closure} parameterized above the matching number. 
\end{abstract}

\section{Introduction}

\emph{Matroids} are abstract structures that capture a notion of \emph{independence} across various domains -- from linearly independent
sets in vector spaces to
the spanning forests in graphs. Formally, a matroid is a pair
$M=(V,\cI)$ where $\cI \subseteq 2^V$, referred to as the \emph{independent sets},
is a collection of sets satisfying the axioms (i) $\emptyset \in \cI$; (ii) if $A \subseteq B$ and $B \in \cI$
then $A \in \cI$ and (iii) the \emph{independence augmentation axiom},
\[
  \text{ if } A, B \in \cI \text{ and } |A|<|B| \text{ then there exists } x \in B \setminus A \text{ such that } A \cup \{x\} \in \cI.
\]
A \emph{basis} of a matroid is a maximal independent set. A matroid can equivalently be defined as a pair $M=(V,\cB)$
where $\cB \subseteq 2^V$ is the set of \emph{bases} of $M$, subject to (i) $\cB \neq \emptyset$ and (ii)
the \emph{basis exchange axiom},
\[
  \text{ if } A, B \in \cB  \text{ and } x \in A \setminus B \text{ then there exists } y \in B \setminus A \text{ such that } A \Delta \{x,y\} \in \cB,
\]
where $\Delta$ denotes symmetric difference. 
The study of matroids goes back nearly a century (to
Whitney, 1935~\cite{whitney1935abstract}), and matroids find applications in practically all
parts of computer science. One of the central classes of matroids is \emph{linear matroids},
where the ground set $V$ represents vectors in a vector space over a field $\F$ 
and the independent sets are precisely the sets of vectors that are
linearly independent. This is usually represented via a matrix over $\F$
whose columns are labelled by $V$. 

A major benefit of matroids is the powerful set of algorithmic meta-results they offer. 
Famously, matroids are precisely those structures in which the basic
\emph{greedy algorithm} is guaranteed to find a max-weight solution.
Another pair of well-studied problems are \textsc{Matroid Intersection},
or the problem of finding a maximum-cardinality set in the
intersection of two matroids, and \textsc{Matroid Parity},
where the input is a matroid $M=(V,\cI)$ and a partition $\Pi$
of $V$ into pairs, and the task is to find an independent set $I$
of maximum cardinality that is a union of pairs from $\Pi$
(or equivalently, a basis $B$ of $M$ in which a minimum number of
pairs is broken, where a pair $p \in \Pi$ is broken by $B$
if $|p \cap B|=1$). \textsc{Matroid Intersection} can be solved
in polynomial time over any matroid, while \textsc{Matroid Parity}
is NP-hard in general but can be efficiently solved on a class of
matroids including linear matroids~\cite{Lovasz79,LP86}. 
For more on matroids, see Oxley~\cite{OxleyBook2};
see also Schrijver~\cite{SchrijverBook} for more algorithmic
applications, 

Linear matroids have also seen powerful applications in parameterized
complexity, which is the starting point of our current investigations.
This started with Marx~\cite{Marx09-matroid}, who gave FPT
algorithms for problems over linear matroids, with applications, based
on an algorithmic version of the representative sets theorem
(a.k.a.\ the \emph{two families theorem} of Bollobás \cite{bollobas1965generalized}, 
generalized to linear matroids by Lovász~\cite{lovasz1977flats}).
This was further developed by Fomin et al.~\cite{FominLPS16JACM},
who showed a faster version and observed several applications to
speeding up FPT dynamic programming algorithms; this strategy has
since become a standard building block in the FPT toolbox~\cite{CyganFKLMPPS15PCbook,FominLPS16JACM,FominLPS17}.
Linear matroids and representative sets have also been instrumental
in results for kernelization, particularly for graph separation
problems~\cite{KratschW14TALG,KratschW20JACM,Kratsch18vc,Wahlstrom22TALG}.
In another direction, the \emph{extensor coding}~\cite{BrandDH18}
and recent \emph{determinantal sieving}~\cite{EKW23} methods
can be used to filter an arbitrary polynomial (e.g., a multivariate
generating function) for monomials whose support forms a basis of a
provided linear matroid. This unifies and generalizes earlier monomial
sieving methods~\cite{BjorklundHKK17narrow,BjorklundKK16} and currently
represents the fastest, most general FPT algorithm for a range of
problems~\cite{EKW23}.

\paragraph*{Delta-matroids.}
\emph{Delta-matroids} are generalizations of matroids, defined and
initially investigated by Bouchet in the 1980s~\cite{Bouchet87DMone}
although similar or equivalent structures were introduced
independently by other researchers; see Moffatt~\cite{Moffatt19deltamatroids}
for a survey. A delta-matroid is a pair $D=(V,\cF)$ where $\cF \subseteq 2^V$ 
is a collection of subsets referred to as \emph{feasible sets},
subject to the \emph{symmetric exchange axiom}: 
\[
\text{ for all } A, B \in \cF \text{ and } x \in A \Delta B \text{ there exists } y \in A \Delta B \text{ such that }
A \Delta \{x,y\} \in \cF.
\]
Both the set of bases of a matroid, and the independent sets of a matroid, 
form the feasible sets of a delta-matroid.
We refer to these as the \emph{basis delta-matroid}
and \emph{independence delta-matroid} of the matroid.
For another prominent example, let $G=(V,E)$ be a graph
and let $\cF$ contain those subsets $F \subseteq V$
such that $G[F]$ has a perfect matching. Then $D=(V,\cF)$
forms a delta-matroid, the \emph{matching delta-matroid} of $G$.
As a special case, given a partition $\Pi$ of $V$ into pairs,
we define the \emph{pairing delta-matroid} $D_\Pi$ of $\Pi$ as 
the matching delta-matroid of the graph whose edge set is $\Pi$;
i.e., $F \subseteq V$ is feasible in $D_\Pi$ if and only if no pair is
broken in $F$.
This is frequently used in combinatorial constructions in this paper.
Like with matroids, there is also a notion of \emph{linear delta-matroid}. 
A matrix $A$ is \emph{skew-symmetric} if $A^T=-A$.
Let $A$ be a skew-symmetric matrix with rows and columns labelled by $V$,
and let $\cF=\{F \subseteq V \mid A[F] \text{ is non-singular}\}$. 
Then $D=(V,\cF)$ defines a delta-matroid 
(although this fact is not trivial);
we denote this delta-matroid by $D=\bD(A)$.
More generally, let $S \subseteq V$ and define
$
  \cF \Delta S = \{F \Delta S \mid F \in \cF\}.
$
Then $D \Delta S=(V,\cF \Delta S)$ defines a delta-matroid,
the \emph{twist} of $D$ by $S$. A \emph{linear delta-matroid}
is a delta-matroid which can be represented as
$D=\bD(A) \Delta S$ for some skew-symmetric matrix $A$ as above.
We will refer to this representation as a \emph{twist representation}.
Examples of linear delta-matroids include 
basis delta-matroids of linear matroids,
\emph{twisted matroids} (being the twists of basis delta-matroids of linear matroids)
and matching delta-matroids, the latter represented via the famous
\emph{Tutte matrix} of $G$. Again, see Moffatt~\cite{Moffatt19deltamatroids}
for more background and applications. 
In addition, slightly more generally, one may consider
\emph{projected} linear representations of delta-matroids,
which includes independence delta-matroids of linear matroids.

Finally, as with matroids, delta-matroids enjoy a collection of algorithmic meta-results.
There is a signed generalization of the greedy algorithm that finds a max-weight
solution in a set system precisely when the set system defines a
delta-matroid, and while both the natural generalizations
\textsc{Delta-matroid Intersection} and \textsc{Delta-matroid Parity}
are NP-hard in general, both are tractable on linear delta-matroids.
More precisely, the problems are usually defined as follows. In
\textsc{Delta-matroid Intersection}, the task is to find \emph{any}
common feasible set $F$ in two delta-matroids $D_1$ and $D_2$, without
any requirement on $|F|$, and in \textsc{Delta-matroid Parity}
the task is to find a feasible set $F$ in a delta-matroid
such that a minimum number of pairs are broken in $F$.
Since delta-matroids are not normally closed under taking subsets, 
this definition of \textsc{Delta-matroid Intersection} is already
a non-trivial problem. These versions are tractable on linear delta-matroids
due to a result of Geelen et al.~\cite{GeelenIM03}, recently
improved to $\tilde O(n^\omega)$ time by Koana and Wahlström~\cite{KW24}.
The natural variant \textsc{Maximum Delta-matroid Intersection},
where $|F|$ is maximised, was raised by Kakimura and Takamatsu~\cite{KakimuraT14sidma},
and appears harder. It was only recently solved by a randomized algorithm~\cite{KW24}.
See the latter for a deeper overview on algorithms for linear delta-matroids.
In summary, although they are less known, delta-matroids (and
especially linear delta-matroids) generalize many of the useful
aspects of (linear) matroids.\iflong\footnote{As a side note, another frequent occurrence of
  matroids in combinatorial optimization is as a form of \emph{constraints},
  striking a good balance between expressive power and algorithmic
  tractability. See for example the work on approximation algorithms
  under matroid and knapsack constraints, and the \emph{matroid secretary} problem.
  Although not precisely the same, there are also a number of
  cases where delta-matroid constraints capture the boundary
  of polynomial-time tractable problems; cf.~the \textsc{Boolean edge CSP}  
  and \textsc{Boolean planar CSP} problems~\cite{KW24,KazdaKR19}
  and the discussion on \emph{general factor} in~\cite{KW24}.
}\fi



In this paper, we investigate parameterized complexity aspects of problems on linear and
projected linear delta-matroids, in the same vein as the work on linear matroids
surveyed above. The representative sets theorem over linear matroids 
has recently been extended to linear delta-matroids, with some initial applications
to kernelization~\cite{Wahlstrom24SODA}, but the applicability of
delta-matroid tools for FPT algorithms has so far not been pursued. 
Our purpose is to push this forward. We investigate
upper and lower bounds for parameterizations of natural NP-hard
problems over linear delta-matroids, on the way generalizing the
method of determinantal sieving. As a consequence, we settle open
questions on above-matching parameterizations of \textsc{Cluster Subgraph}
and \textsc{Strong Triadic Closure}~\cite{abs-2001-06867,GolovachHKLP20}.

\subsection{Delta-matroid intersection and packing problems}

\label{sec:ourresults}

In this paper, we investigate the parameterized complexity aspects of problems over linear delta-matroids.
We show two classes of results; we begin with delta-matroid versions
of basic intersection and packing problems.

More precisely, we introduce the \textsc{$q$-Delta-matroid Intersection}
and \textsc{Delta-matroid Set Packing} problems,
generalizing \textsc{Delta-matroid Intersection} and \textsc{Delta-matroid Parity}.
Among others, we consider the following problems.
\begin{itemize}
\item \textsc{$q$-Delta-matroid Intersection}: Given $q$ 
  delta-matroids $D_1=(V,\cF_1)$, \ldots, $D_q=(V,\cF_q)$
  and an integer $k \in \N$, is there a set of cardinality $k$ that is
  feasible in $D_i$ for every $i \in [q]$?
\item \textsc{$q$-Delta-matroid Parity}: Given a delta-matroid $D=(V,\cF)$, 
  a partition $\cP$ of $V$ into blocks of size $q$, and an integer $k
  \in \N$, is there a feasible set that is the union of $k$ blocks
  from $\cP$? 
\item \textsc{Delta-matroid Set Packing}: Given a delta-matroid
  $D=(V,\cF)$, an arbitrary partition $\cP$ of $V$, and an integer $k
  \in \N$, is there a feasible set $F$ in $D$, $|F|=k$, which is the union of
  blocks from $\cP$?
\end{itemize}
We use the following shorthand: \textsc{DDD Intersection}
refers to \textsc{3-Delta-matroid Intersection};
\textsc{DDM Intersection} refers to the special case of
\textsc{3-Delta-matroid Intersection} where $D_3$ is a matroid;
and \textsc{D$\Pi$M Intersection} refers to the case where $D_1$ 
is an arbitrary delta-matroid, $D_2$ is a pairing
delta-matroid for a
partition $\Pi$, and $D_3$ is a matroid. In other words,
\textsc{D$\Pi$M Intersection} is \textsc{Delta-matroid Parity}
with an additional matroid constraint. Throughout, we assume that all
matroids and delta-matroids are linear (or projected linear)
with representations given in the input over some common field $\F$. 

The classical problems \textsc{Linear Delta-matroid Intersection} and 
\textsc{Linear Delta-matroid Parity} correspond to the special case
where $q=2$ and where there is in addition no cardinality constraint $k$
on the solution, and was solved by Geelen et al.~\cite{GeelenIM03}.
For the variant where $q=2$ and the cardinality constraint is present,
only a randomized algorithm is known~\cite{KW24} (even if replace the
constraint $|F|=k$ by asking $|F|$ to be maximized). 

The matroid versions of these problems are all FPT for linear matroids
when parameterized by the solution cardinality $k$ for any $q$ (and more
generally FPT when parameterized by $k+q$)~\cite{Marx09-matroid,BrandDH18,EKW23}.
The most powerful method here, at least for representations over
fields of characteristic 2, is the recent method of
\emph{determinantal sieving}~\cite{EKW23}, which forms the inspiration 
for our algorithms; see below.
In particular, for matroids represented over a field of characteristic 2,
\textsc{$q$-Matroid Intersection} can be solved in $O^*(2^{(q-2)k})$ time,
\textsc{$q$-Matroid Parity} in $O^*(2^{qk})$ time, and
\textsc{Matroid Set Packing} in $O^*(2^k)$ time~\cite{EKW23}.
Results over other fields are somewhat slower, but still single-exponential.

For delta-matroids, we find a more intricate picture. The \emph{rank}
of a delta-matroid $D=(V,\cF)$ is the cardinality of a largest
feasible set in $D$. We consider the above problems both parameterized
by $k$ and by the rank $r$. For matroids, this distinction makes no
difference -- given a linear matroid $M=(V,\cI)$ of rank $r$ and a
parameter $k$, via a \emph{truncation} operation one can construct
a linear matroid $M'=(V,\cI')$ of rank $k$, which retains all
independent sets of cardinality at most $k$ from $M$. 
Furthermore, this can be done efficiently with a maintained linear
representation~\cite{Marx09-matroid,LokshtanovMPS18TALG}. 
Thus, over matroids, parameterizing by $r$ is no easier than
parameterizing by $k$. However, for delta-matroids, the distinction
turns out to be highly significant. 

Parameterized by $k$, we find a dividing line where some problems,
including \textsc{DDM Intersection}, are FPT given representations
over a common field $\F$, whereas others, including \textsc{DDD Intersection},
are W[1]-hard. This holds even for very simple delta-matroids; see
Section~\ref{sec:part1}.

\begin{restatable}{theorem}{dmhardness} \label{hard:ddd}
  \textsc{DDD Intersection} is W[1]-hard parameterized by
  $k$ even for linear delta-matroids represented over the same field. 
\end{restatable}

On the other hand, if we parameterize by the rank $r$, then we recover
all positive cases from the matroid case, and with matching running times.

\begin{restatable}{theorem}{dmponefpt} \label{tractable:rank}
  The following problems are FPT over linear delta-matroids and
  matroids provided as representations over a common field. 
  \begin{itemize}
  \item \textsc{DDD Intersection} parameterized by $r=\rank(D_3)$,
    with a running time of $O^*(2^r)$ over characteristic 2
    and $O^*(2^{O(r)})$ otherwise
  \item More generally, \textsc{$q$-Delta-matroid Intersection} parameterized by $r$,
    where $r$ is the maximum rank of $D_i$, $i \geq 3$, with a running
    time of $O^*(2^{(q-2)r})$ over characteristic 2
    and $O^*(2^{O(qr)})$ otherwise
  \item \textsc{$q$-Delta-matroid Parity} and \textsc{Delta-matroid Set Packing}
    parameterized by the rank $r$ of the delta-matroid,
    with a running time of $O^*(2^r)$ over characteristic 2
    and $O^*(2^{O(r)})$  in general
  \item \textsc{DDM Intersection} and \textsc{D$\Pi$M Intersection},
    parameterized by the cardinality $k$, 
    with a running time of $O^*(2^k)$ over characteristic 2
    and $O^*(2^{O(k)})$  in general
  \item More generally, \textsc{Delta-matroid Intersection} and
    \textsc{Delta-matroid Parity} with an additional $q-2$ matroid constraints,
    parameterized by $q$ and the cardinality $k$, 
    with a running time of $O^*(2^{(q-2)k})$ over characteristic 2
    and $O^*(2^{O(qk)})$ in general
  \end{itemize}
  In all cases, the algorithms are randomized, and use polynomial
  space over characteristic 2 but exponential space otherwise.
\end{restatable}

Finally, we note a problem of \textsc{Colorful Delta-matroid Matching}
that will be used later in the paper. In this problem, we are
given a delta-matroid $D=(V,\cF)$, a graph $G=(V,E)$ 
with edges colored in $k$ colors, and an integer $k \in \N$.
The task is to find a feasible set $F \in \cF$, $|F|=2k$
such that $G[F]$ has a matching with $k$ distinct edge colors. 
This reduces to \textsc{D$\Pi$M Intersection}.

\begin{restatable}{theorem}{dmmatching} \label{corollary:colorful-matching}
  \textsc{Colorful Delta-matroid Matching} on linear delta-matroids
  can be solved in $O^*(2^k)$ time if the representation is over a
  field of characteristic~2 and in $O^*(2^{O(k)})$ time otherwise.
\end{restatable}

\paragraph*{Sieving methods and a brief technical comment.}
Due to the amount of required preliminaries, we leave all technical
discussion to Section~\ref{sec:part1}, but we want to highlight 
the main technical contribution from this part of our paper. 
As mentioned, the main technical tool in the fastest FPT algorithms
for problems over linear matroids is \emph{determinantal sieving}.
Let $P(X)$ be a multivariate polynomial over a field $\F$, and assume
that $P(X)$ is given as black-box evaluation access with the ability
to also evaluate $P(X)$ over an extension field of $\F$ as needed.
The \emph{support} of a monomial $m$ in $P(X)$ is the set of variables
with non-zero degree in $m$. 
The main technical statement is then as follows.

\begin{theorem}[Basis sieving~\cite{EKW23}]
  \label{theorem:basis-sieve}
  Let $P(X)$ be a polynomial of degree $d$ over a set of variables $X$ over a field $\F$ of characteristic 2
  and let $M=(X,\cI)$ be a linear matroid of rank $k$ represented over $\F$. 
  There is a randomized algorithm that in time $O^*(d2^k)$ and
  polynomial space tests if the monomial expansion of $P(X)$ contains
  a term which is multilinear of degree $k$ and whose support is a basis of $M$. 
\end{theorem}

For a quick illustration, let $A_1, A_2 \in \F^{k \times V}$ be two
linear matroids of rank $k$ over a ground set $V$. Recall that
by the Cauchy-Binet formula,
$\det A_1A_2^T=\sum_{B \in \binom{V}{k}} \det A_1[B] \det A_2[B]$.
By scaling the columns of $A_1$ by a set of variables, we can define
a polynomial $P(X)$ that ``enumerates'' (or, is a generating function for)
common bases of $A_1$ and $A_2$. Given a third linear matroid
represented by a matrix $A_3 \in \F^{k \times V}$, 
Theorem~\ref{theorem:basis-sieve} combined with the Schwartz-Zippel lemma
then immediately gives an $O^*(2^k)$-time
algorithm for \textsc{3-Matroid Intersection} for linear matroids
represented over a common field.

We show a generalization of Theorem~\ref{theorem:basis-sieve} to
linear delta-matroids. This generalization comes without any
additional cost to the running time (up to the polynomial factor)
compared to the above. For technical simplicity we assume that $P(X)$
is homogeneous and that the delta-matroid is given with a so-called
\emph{sparse representation}. Both of these assumptions can be made
without loss of generality; see Section~\ref{sec:dm-sieve}.

\begin{restatable}{theorem}{dmsieve} \label{theorem:sieving-char-2}
  Let $D=(V,\cF)$, $V=\{v_1,\ldots,v_n\}$ be a (projected) linear delta-matroid of rank $r$ given as a sparse representation $\bD(A) / T$,
  and let $P(X)$, $X=\{x_1,\ldots,x_n\}$ be a homogeneous polynomial of degree $k$
  given as black-box access over a field of characteristic 2 with at least $3k$ elements.
  In time $O^*(2^r)$ we can sieve for those terms of $P(X)$ which
  are multilinear and whose support $\{x_i \mid i \in I\}$
  in $X$ is such that $\{v_i \mid i \in I\}$ is feasible in $D$. 
\end{restatable}

Theorems~\ref{theorem:basis-sieve} and~\ref{theorem:sieving-char-2}
also have variants over fields of characteristic other than 2, except that
except (1) the process uses exponential space, (2) the process works
over an arithmetic circuit, not via black-box access, and (3) with
a higher running time of $O^*(2^{O(r)})$;
see Section~\ref{sec:dm-sieve}.

Most of the algorithmic results in Theorem~\ref{tractable:rank} work
by applying Theorem~\ref{theorem:sieving-char-2} to a suitably
constructed polynomial $P(X)$ generalizing the Cauchy-Binet example
into linear delta-matroids. In fact, a good working summary of our
upper and lower bounds for packing and intersection problem is the
following: \emph{If a problem can be solved using Theorem~\ref{theorem:sieving-char-2} 
  then it is FPT; otherwise W[1]-hard.}
Thus, for example, \textsc{DDD Intersection} is hard with parameter $k$
since there is no guarantee that either of the three delta-matroids
$D_1$, $D_2$, $D_3$ has rank bounded by $k$, but it is FPT under the
parameter $r=\rank(D_3)$; and (e.g.) \textsc{DDM Intersection}
is FPT since the matroid layer can be truncated to rank $k$.
See Section~\ref{sec:part1} for details.

However, the second set of results in our paper is a counterpoint to
this, presenting a set of surprising algorithms that cannot be
explained purely algebraically. 

\subsection{Delta-matroid triangle cover with applications}
Inspired by applications in graph theory (see below) we study a
relaxation of \textsc{3-Delta-matroid Parity} which,
surprisingly, turns out to have an FPT algorithm.

Let $D=(V,\cF)$ be a linear delta-matroid and let
$\cT=\{T_1,\ldots,T_m\} \subseteq \binom{V}{3}$ be a set of
triples over $V$ (which we refer to as triangles).
We say that a set $S \subseteq V$ can be \emph{covered} by triangles
if there is a disjoint packing $\cT' \subseteq \cT$
such that $S \subseteq V(\cT')$. Note that we do not require
that $S$ is precisely partitioned by $\cT'$, i.e., we allow for
triangles $T_i \in \cT'$ such that $|T_i \cap S| \in \{1,2\}$.
Let $k \in \N$. The \textsc{Delta-matroid Triangle Cover} problem
asks to find a feasible set $F \in \cF$ with $|F|=k$ such that
$F$ can be covered by triangles. 

\begin{restatable}{theorem}{dmtc}
  \label{theorem:dmtc-fpt}
  \textsc{Delta-matroid Triangle Cover} can be solved in $O^*(k^{\Oh(k)})$ time
  over a linear delta-matroid.
\end{restatable}

Let us make a few observations. On the one hand, this problem generalizes
\textsc{Colorful Delta-matroid Matching}:
Let $(D=(V,\cF), G=(V,E), k)$ be an instance of the latter.
Introduce $k$ new vertices $z_1$, \ldots, $z_k$, and for
every edge $uv \in E$ of color $i$ create a triangle $\{u,v,z_i\}$.
Thus in particular \textsc{Delta-matroid Triangle Cover}
generalizes \textsc{Delta-matroid Parity}. 
On the other hand, the most obvious structural simplification
of \textsc{Delta-matroid Triangle Cover} would seem to be to assume 
that we are looking for a set $F \in \cF$ that is precisely
partitioned by triangles. However, this defines
\textsc{3-Delta-matroid Parity}, which is W[1]-hard.
Hence \textsc{Delta-matroid Triangle Cover} inhabits
an interesting spot of being ``just barely'' tractable.
Indeed, by the same reduction as above, the extended version
\textsc{Delta-matroid $K_4$ Cover}, where $F$ is to be covered
by sets of cardinality 4, generalizes
\textsc{3-Delta-matroid Parity} and is W[1]-hard.

Our algorithm for \textsc{Delta-matroid Triangle Cover} seamlessly combines combinatorial and algebraic techniques,
ultimately reducing the problem to \textsc{Colorful Delta-matroid Matching}.  
On the algebraic side, we introduce a new delta-matroid operation, termed \emph{$\ell$-projection},
which generalizes the classical matroid truncation operation.  
We believe that this operation is of independent interest, particularly for kernelization \cite{Wahlstrom24SODA}.
On the combinatorial side, we use arguments inspired by the sunflower lemma to effectively reduce triangles to pairs.  
For details, see Section~\ref{sec:triangle-cover}.

We give two graph-theoretic applications of this problem.
The first is to  \textsc{Cluster Subgraph}: given a graph $G$ and an integer $\ell$, the problem is to find a cluster subgraph (i.e., a graph whose connected components are cliques) with at least $\ell$ edges.
This problem is NP-hard \cite{ShamirST04}, but can be solved in $O^*(9^\ell)$ time \cite{GruttemeierK20}.
We consider a tighter \emph{above guarantee} parameterization of $\ell-\MM(G)$ where $\MM(G)$ is the matching number of $G$.
Note that a matching is a cluster subgraph. 
The parameterized complexity with respect to this parameterization is marked as an open question in a recent survey~\cite{abs-2001-06867}.

\begin{restatable}{theorem}{csam} \label{theorem:cs-am}
  \textsc{Cluster Subgraph} can be solved in $O^*(k^{O(k)})$ time for $k = \ell - \MM(G)$.
\end{restatable}

To see the connection to \textsc{Delta-matroid Triangle Cover},
consider the case that $G$ is $K_4$-free, and assume for simplicity
that $G$ has a perfect matching. Then the goal is to find a
packing of $K_2$'s and $K_3$'s with as many edges as possible, i.e.,
a collection of $K_3$'s whose deletion leaves a graph with a large
matching number.\footnote{
  We remark that the problem becomes polynomial-time solvable
  when the goal is to maximize the number of covered vertices \cite{HellK84}.
  However, the edge maximization problem is NP-hard as it generalizes the perfect triangle packing problem.
} We show that this reduces to solving
\textsc{Delta-matroid Triangle Cover} over the dual matching
delta-matroid of $G$. On the other hand, if $G$ contains a large
packing of $K_4$'s, then the instance is immediately positive. 
Thus, we may assume that $G$ has a set $K$ of $O(k)$ vertices
such that $G-K$ is $K_4$-free, and Theorem~\ref{theorem:cs-am}
follows by guessing its interaction with the solution. 

\begin{figure}
  \centering
  \begin{tikzpicture}
    \def\radius{1.2cm}
    \def\numvertices{7}
    \foreach \x in {1,...,\numvertices} {
      \pgfmathsetmacro{\angle}{360/\numvertices * (\x - 1)}
      \node[circle,fill,inner sep=2pt] (vertex\x) at (\angle:\radius) {};
    }
    \foreach \x in {1,...,\numvertices} {
      \pgfmathsetmacro{\angle}{360/\numvertices * (\x - 1)}
      \draw[very thick] (\angle:\radius) node[circle,fill,inner sep=2pt] {} -- (\angle + 360/\numvertices:\radius);
    }
    \foreach \x in {1,...,\numvertices} {
      \pgfmathtruncatemacro{\target}{mod(\x+1, \numvertices)+1}
      \pgfmathtruncatemacro{\source}{\x}
        \draw (vertex\source) to[bend left] (vertex\target);
    }
    \foreach \x in {1,...,\numvertices} {
      \pgfmathsetmacro{\angle}{360/\numvertices * (\x - 1)}
      \node at (\angle:1.5cm) {\x};
    }
  \end{tikzpicture}
  \caption{A maximum strong set (thick lines) which does not induce a cluster graph.}
  \label{fig:strong-cycle}
\end{figure}
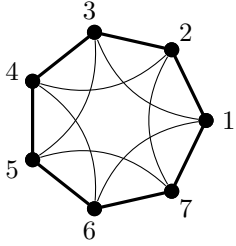

The second application is to \textsc{Strong Triadic Closure}, formalized by Sintos \cite{SintosT14}.
Let $G$ be a graph. We say that a 3-path $(u, v, w)$ is \emph{closed} if the edge $uw$ exists in $G$.
For a graph $G$, an edge set $S$ is said to be \emph{strong} if every 3-path in $S$ is closed. 
Note that the shortcut edge $uw$ is not required to exist in $S$, only in $G$.
Given $G=(V,E)$ and $\ell$, the \textsc{Strong Triadic Closure}
problem is to find a strong edge set $S \subseteq E$ with $|S| \geq \ell$. 
This is based on the notion of triadic closure \cite{granovetter1973strength}.
This problem is closely related to \textsc{Cluster Subgraph} as an edge set that induces a cluster graph is strong.
However, not every strong set induces a cluster subgraph (see Figure~\ref{fig:strong-cycle}).
It is known to be NP-hard~\cite{SintosT14} and FPT when parameterized by $\ell$~\cite{GolovachHKLP20,GruttemeierK20}
with the fastest algorithm running in $O^*(\ell^{O(\ell)})$ time~\cite{GruttemeierK20}.
Again, we consider the \emph{above matching} parameterization. 

Again, the $K_4$-free case turns out to be tractable: if $G$ is $K_4$-free, then the
solution $S$ must be a collection of strong paths and cycles, and
in fact it can be shown that there exists an optimum where $S$ consists
only of strong cycles and $K_2$'s. Similarly to above, we solve the
problem by looking for a feasible set $F$ in the dual matching
delta-matroid of $G$ such that $F$ can be covered by a collection of
strong cycles, possibly of length $\Omega(k)$. 
However, this is significantly more complicated -- as mentioned before, the \textsc{Delta-matroid $K_c$ Cover} is W[1]-hard for any $c \ge 4$.
Yet, we still manage to give a reduction to \textsc{Colorful Delta-matroid Matching}.
To that end, we identify several cases where a strong cycle can be replaced by triangles.
We then argue that when this triangle replacement fails, pairing constraints (which can be encoded in \textsc{Colorful Delta-matroid Matching}) are powerful enough to solve the problem.

\begin{restatable}{theorem}{tcamfpt} \label{theorem:tc-am-fpt}
  \textsc{Strong Triadic Closure} on $K_4$-free graphs can be solved in $O^*(k^{O(k)})$ time for $k = \ell - \MM(G)$.
\end{restatable}

Furthermore, as above, a $K_4$ in $G$ implies an immediate contribution to $\ell-\MM(G)$,
and we can find a small modulator $K$ such that $G-K$ is $K_4$-free.
Unfortunately, despite our efforts, achieving a full FPT algorithm remained elusive -- in fact, guided by the very precise 
structure of solutions where our algorithmic insights prove ineffective, we
show that the problem is W[1]-hard in general.
Intriguingly, the source of hardness contradicts our initial beliefs, stemming not from strong cycles but from an unexpectedly different aspect.

\begin{restatable}{theorem}{tcamhard}
  \textsc{Strong Triadic Closure} is W[1]-hard when parameterized by $k = \ell - \MM(G)$.
\end{restatable}

On the other hand, a counting argument shows that, although we do not get an FPT
algorithm, we can get a constant-factor FPT approximation simply by
packing a modulator $K$ and ignoring all edges $\delta(K)$ between
$K$ and $G-K$ in the solution, solving $G[K]$ (exhaustively)
and $G-K$ separately. That is, we give an algorithm that either
shows that $G$ has no strong edge set of size $\MM(G)+k$
or finds a strong edge set of size $\MM(G) + \Omega(k)$ in FPT time.

\begin{restatable}{theorem}{tcamapprox}
  There is $O^*(k^{O(k)})$-time algorithm for \textsc{Strong Triadic Closure} that approximates the above matching parameter $\ell - \MM(G)$ within a  factor 7.
\end{restatable}

In addition, we generalize a result of Golovach et al.~\cite{GolovachHKLP20}, who gave an FPT algorithm for maximum degree 4,  
by showing that \textsc{Strong Triadic Closure} can be solved in time $O^*(\Delta^{O(k)})$
for graphs of maximum degree $\Delta$:

\begin{restatable}{theorem}{tcamdelta}
  \textsc{Strong Triadic Closure} can be solved in $O^*(\Delta^{O(k)})$ time, where $\Delta$ is the maximum degree of $G$ and $k = \ell - \MM(G)$.
\end{restatable}

Lastly, we remark that the question whether \textsc{Strong Triadic Closure} is FPT has
been asked repeatedly in the literature~\cite[Open Problem 5.4]{abs-2001-06867},~\cite{GolovachHKLP20}.
The connection to \textsc{Cluster Subgraph} has also been observed~\cite{GruttemeierK20}. 
We settle the status of both of these problems. 

\subsection{Further related work}

Given a graph, the problem of finding a vertex-disjoint $\{K_2, K_3\}$-packing that covers the maximum number of vertices can be solved in polynomial time \cite{HellK84}.
In contrast, maximizing the number of edges in such a packing, as discussed above, is NP-hard since it generalizes the problem of partitioning a graph into triangles.
Moreover, this problem is APX-hard and admits a $3/2$-approximation \cite{ChataignerMWY09}.

There is a range of algebraic algorithms for FPT problems, beyond
what has been surveyed here. Much of this is surveyed by Eiben et
al.~\cite{EKW23}. For instance, the basic \textsc{$q$-Set Packing}
and \textsc{$q$-Dimensional Matching} problems parameterized by cardinality
are currently best solved by reduction to algebraic sieving arguments~\cite{BjorklundHKK17narrow},
and it is an intriguing question what the best possible running time
for these problems is.

Other intriguing connections for delta-matroids includes their
generalization to \emph{jump systems}~\cite{BouchetC95}, with 
connections to discrete convexity theory (cf.~\cite{murota2021note}).
In another direction, Bouchet generalized delta-matroids into
\emph{multi-matroids}, with their own notion of representations~\cite{BouchetMMI,BouchetMMIV}.

For further related work, we again refer to Moffatt~\cite{Moffatt19deltamatroids} and
Koana and Wahlström~\cite{KW24}.

\iflong
\paragraph*{Organization.} Section~\ref{sec:prelim} covers additional
background material on delta-matroids. Section~\ref{sec:part1}
contains the first set of results, regarding intersection and packing
problems as well as the delta-matroid sieving result.
Section~\ref{sec:triangle-cover} gives the FPT algorithm for
\textsc{Delta-matroid Triangle Cover}, and
Sections~\ref{sec:cs-above-matching} and~\ref{sec:tcam}
cover the \textsc{Cluster Subgraph} and \textsc{Strong Triadic Closure}
applications, respectively. We wrap up in Section~\ref{sec:conclusions}
with conclusions and open question.
\fi

\iflong
\paragraph*{Preliminaries.} 
A parameterized problem is a collection of pairs $(I, k)$, where $I$ is the instance and $k$ is called the parameter, and
is \emph{fixed-parameter tractable} (in FPT) if there is an \emph{FPT algorithm}.
Such an algorithm solves every instance $(I,k)$ in $f(k) \cdot |I|^{O(1)}$ time, for some computable function $f$.
This complexity is denoted by $O^*(f(k))$.
For a comprehensive exposition on parameterized complexity, refer to the book of Cygan et al.~\cite{CyganFKLMPPS15PCbook}.

We will use standard notation for graphs as outline in Diestel \cite{DiestelBook}.
By default, graphs are assumed to be undirected. 
For a graph $G$, we often denote the sets of its vertices and edges by $V(G)$ and~$E(G)$, respectively.
The vertex and edge sets of a graph $G$ are denoted by $V(G)$ and $E(G)$, respectively, with $n = |V(G)|$ and $m = |E(G)|$.
We may abbreviate an edge $\{ u, v \}$ in an undirected graph as $uv$.
Similarly, in directed graphs, an arc from $u$ to $v$ is represented by $(u, v)$ or $uv$.
A vertex $v\in V(G)$ has neighborhood $N_G(v)=\{u \mid uv \in E\}$, and degree
$|N_G(v)|$.
For any subset $V'\subseteq V$, $G[V']$ denotes the \emph{induced subgraph} of $G$ on $V'$, that is, $G[V']=(V',E')$ where $E' = \{ uv \in E \mid u \in V, v \in V \}$.
The notation $G - V'$ refers to the induced subgraph $G[V \setminus V']$.

Given a matrix $A$ and a set of rows $S$ and columns $T$, $A[S, T]$ denotes the submatrix containing rows $S$ and columns $T$.
If $S$ or $T$ contains all rows or columns, respectively, we use the notation $A[\cdot, T]$ and $A[S, \cdot]$, where $\cdot$ denotes all rows or columns.
The $n \times m$ zero matrix and the $n \times n$ identity matrix is denoted by $O_{n \times m}$ and $I_{n}$, with the subscripts omitted when the dimension are clear from context.

We assume that each field in this paper comprises $\exp(n^{O(1)})$ many elements.
Under this assumption, all arithmetic operations can be performed in polynomial time.
The Schwartz-Zippel lemma \cite{Schwartz80,Zippel79} states that a polynomial $P(X)$ of total degree $d$ or less becomes nonzero with a probability of at least $1 - d / |\F|$, when evaluated at uniformly random elements from $\F$, provided $P(X)$ is not identically zero.

\fi

\section{Background on delta-matroids} \label{sec:prelim}

\subsection{Preliminaries on delta-matroids and skew-symmetric matrices} 

For two sets $A, B$, we let $A \Delta B=(A \setminus B) \cup (B \setminus A)$ denote their symmetric difference. 

\paragraph*{Delta-matroids.}

A \emph{delta-matroid} is defined as a pair $D=(V,\cF)$ where $V$ is a ground set
and $\cF \subseteq 2^V$ is a collection of \emph{feasible sets}.
The feasible sets adhere to the following \emph{symmetric exchange axiom}:
\[
\forall A, B \in \cF,\, x \in A \Delta B\, \exists y \in A \Delta B :
A \Delta \{x,y\} \in \cF.
\]
For a delta-matroid $D=(V,\cF)$, let $V(D)=V$ and $\cF(D)=\cF$. 
A delta-matroid is called \emph{even} if all feasible sets share the same
parity. 
This implies $x \neq y$ in the application of the
symmetric exchange axiom, although this is not required generally.

Delta-matroids also feature operations that extend beyond traditional matroid theory. Specifically, the \emph{twisting} of a delta-matroid 
$D=(V,\cF)$ by $S \subseteq D$ is the delta-matroid
$D \Delta S = (V, \cF \Delta S)$, where
$
  \cF \Delta S = \{F \Delta S \mid F \in \cF\}.
$
The \emph{dual} delta-matroid of $D$ is $D \Delta V(D)$.
The \emph{deletion} of $S \subseteq V$
yields a delta-matroid
$
  D \setminus S = (V \setminus S, \{F \in \cF \mid F \subseteq V \setminus S\}),
$
and the \emph{contraction} by $S$ is defined by
$
  D/S = (D \Delta S) \setminus S = (V \setminus S, \{F \setminus S \mid F \in \cF, S \subseteq F\}).
$

\paragraph*{Skew-symmetric matrices.}

A square matrix $A$ is \emph{skew-symmetric} if $A = -A^T$.
Particularly, when $A$ is over a field of characteristic 2, it is assumed
 that it has zero diagonal, unless stated otherwise.
For a skew-symmetric matrix $A$, whose rows and columns are indexed by a set $V = [n]$, 
the \emph{support graph} of $A$ is the graph $G=(V, E)$ where $E$ is the set of edges $\{ uv \mid A[u, v] \ne 0 \}$.
A fundamental tool in the analysis of skew-symmetric matrices is the
\emph{Pfaffian}, which is defined for a skew-symmetric matrix $A$ as follows:
\begin{align*}
  \Pf A = \sum_{M} \sigma(M) \prod_{e \in M} A[u, v],
\end{align*}
where the sum is taken over all perfect matchings of $A$'s support graph, and 
$\sigma(M) \in \{1,-1\}$ denotes the sign of the permutation:
\[
  \begin{pmatrix}
    1 & 2 & \cdots n - 1 & n \\
    v_1 & v_1' & \cdots v_{n/2} & v_{n/2}'
  \end{pmatrix},
\]
where $M = \{ v_i v_i' \mid i \in [n/2] \}$ with $v_i < v_i'$ for all $i \in [n/2]$.
It is well-known that $\det A=(\Pf A)^2$, so $\Pf A \neq 0$
if and only if $A$ is non-singular.
However, directly working with the Pfaffian is often more effective. 
The Pfaffian can be computed in $O(n^3)$ steps by an elimination procedure.
Alternatively, a combinatorial algorithm for computing the Pfaffian is also known \cite{MahajanSV99}.
The Pfaffian generalizes the notion of determinants as follows.
\begin{lemma} \label{lemma:det-pf}
For an $n \times n$-matrix $M$, it holds that
\begin{align*}
 \det M = (-1)^{\binom{n}{2}} \Pf \begin{pmatrix}
  O & M \\
  -M^T & O \\
\end{pmatrix}.
\end{align*}
\end{lemma}

An essential operation on skew-symmetric matrices is \emph{pivoting}.
Given a skew-symmetric matrix $A \in \F^{n \times n}$ and $S \subseteq [n]$ where $A[S]$ is non-singular,
reorder the rows and columns of
$A$ such that 
\[
  A =
  \begin{pmatrix}
    B & C \\
    -C^T& D
  \end{pmatrix},
\]
where $A[S]=B$.
Then the \emph{pivoting} of $A$ by $S$ is defined as
\[
  A * S =
  \begin{pmatrix}
    B^{-1} & B^{-1}C \\
    C^T B^{-1}& D + C^T B^{-1} C
  \end{pmatrix},
\]
which is again a skew-symmetric matrix.

\begin{lemma}[\cite{tucker1960combinatorial}] \label{lemma:tucker}
It holds that for any $X \subseteq [n]$, 
\begin{align*}
  \det (A*S)[X] = \frac{\det A[X \Delta S]}{\det A[S]},
\end{align*}
thereby,
$(A*S)[X]$ is non-singular if and only if $A[X \Delta S]$ is
non-singular. 
\end{lemma}

Additionally, we mention a formula regarding the Pfaffian of the sum of two skew-symmetric matrices:
\begin{lemma}[{\cite[Lemma 7.3.20]{Murota99}}]
  \label{lemma:sum-pf}
  For two skew-symmetric matrices $A_1$ and $A_2$, both indexed by $V$,
  it holds that
  \[ \Pf \, (A_1 + A_2) = \sum_{U \subseteq V} \sigma_U \Pf A_1[U] \cdot \Pf A_2[V \setminus U], \]
  where $\Pf A_i[\emptyset] = 1$ for $i = 1, 2$ and $\sigma_U \in \{1,-1\}$ is a sign of the permutation
  \begin{align*}
    \begin{pmatrix}
      1 & 2 & \cdots & |U| & |U| + 1 & \cdots & |V| - 1 & |V| \\
      u_1 & u_2 & \cdots & u_{|U|} & v_1 & \cdots & v_{|V \setminus U|-1} & u_{|V \setminus U|}
    \end{pmatrix},
  \end{align*}
  where $u_i$ and $v_i$ are the $i$-th largest elements of $U$ and $V \setminus U$, respectively.
\end{lemma}

Finally, we note a generalization of the Cauchy-Binet formula to skew-symmetric matrices:

\begin{restatable}[Ishikawa-Wakayama formula \cite{IshikawaW95}]{lemma}{iwformula}
  \label{lemma:cauchy-binet-ss}
  For a skew-symmetric $2n \times 2n$-matrix $A$ and a $2k \times 2n$-matrix $B$ with $k \le n$, we have
  \begin{align*}
    \Pf B A B^T = \sum_{U \in \binom{[2n]}{2k}} \det B[\cdot, U] \Pf A[U].
  \end{align*}
\end{restatable}

\subsection{Representations and operations on linear delta-matroids}
\label{ssec:rep}

We review some recent results on representations of, and operations
over, linear delta-matroids due to Koana and Wahlström~\cite{KW24}.

The traditional notion of linear representation of a delta-matroid is as follows. 
For a skew-symmetric matrix $A \in \F^{V \times V}$ over a field $\F$, define
\[ \cF = \{ X \subseteq V \mid A[X] \text{ is nonsingular} \}. \]
Then, $(V, \cF)$, which is denoted by $\bD(A)$, is a delta-matroid. 
We say that a delta-matroid $D = (V, \cF)$ is \emph{representable} over $\F$ if there is a skew-symmetric matrix $A \in \F^{V \times V}$ and a twisting set $X \subseteq V$ such that $D = \bD(A) \Delta X$.
If $A[X]$ is nonsingular, or equivalently $\emptyset \in \cF$, we say that $D$ is \emph{directly representable} over~$\F$.
Note that a directly representable delta-matroid $D$ can be represented without a twisting set $X$, as $\bD(A) \Delta X = \bD(A * X)$ by Lemma~\ref{lemma:tucker}.
We will say that $D$ is \emph{directly represented} by $A$ if $D = \bD(A)$.
We also refer to a representation $D=\bD(A) \Delta S$ as a \emph{twist representation}.

Another linear representation, of equivalent expressive power,
recently proposed by Koana and Wahlström~\cite{KW24}
is a \emph{contraction representation}, where a delta-matroid $D=(V,\cF)$
is represented as $D=\bD(A)/T$ where $A$ is a skew-symmetric matrix over ground set $V \cup T$. 
These forms have equivalent power, in the sense that one can convert between
twist representations and contraction representations of a given
linear delta-matroid in matrix multiplication time~\cite{KW24}. 
However, contraction representations are easier to work with for
certain algorithmic purposes.
Given a twist representation $D = \bD(A) \Delta S$, we get a contraction representation $D = \bD(A^*) / T$, where $A^*$ is indexed by $V \cup T$ (with $T$ being a set of size~$|S|$), defined as follows:
\begin{align} \label{eq:twist-to-contraction}
  A^* = \kbordermatrix{
    & T & S & V \setminus S \\
    T & A[S] & -I & -A[S, V \setminus S] \\
    S & I & O & O \\
    V \setminus S & A[V \setminus S, S] & O & A[V \setminus S] \\
  }.
\end{align}

Let $D=(V,\cF)$ be a delta-matroid and $X \subseteq V$.
The \emph{projection} $D|X$ of $D$ to $V \setminus X$ 
is a delta-matroid $D'=(V \setminus X, \cF')$
where $\cF'=\{F \setminus X \mid F \in \cF\}$. 
A \emph{projected linear delta-matroid}
is a delta-matroid $D=(V,\cF)$ represented as a projection
$D=D'|X$ from some linear delta-matroid $D'$ on a ground set $V \cup X$.
A projected linear delta-matroid is not necessarily even,
hence they generalize linear delta-matroids.
Every projected linear delta-matroid can be represented
as an \emph{elementary projection} $D=D'|\{x\}$ for a single element $x$,
and such a representation can be constructed from a general
projected linear representation in randomized polynomial time.\footnote{Koana
  and Wahlström consider \emph{$\varepsilon$-approximate representations}
  for a tweakable error parameter $\varepsilon$. Since we are not
  optimizing the polynomial factors in the running times in this
  paper, we set $\varepsilon = 1/2^{\Omega(n)}$ and construct a representation
  in polynomial time that is inexact only with exponentially small failure probability.}

In particular, let $A \in \F^{r \times V}$ be a linear representation
of a matroid $M$ and construct the matrix
\[
  A_D= \kbordermatrix {
    & T & V \\
    T & 0 & A \\
    V & -A^T & 0
  }.
\]
Then, by Lemma~\ref{lemma:det-pf}, $D_B=\bD(A_D)/T$ represents the set of bases of $M$ (i.e.,
the feasible sets of $D_B$ are precisely the bases of $M$)
and $D_I=\bD(A_D)|T$ represents the independent sets of $M$. 

Let $D_1=(V,\cF_1)$ and $D_2=(V,\cF_2)$ be delta-matroids on a shared
ground set. The \emph{union} $D_1 \cup D_2$ is the delta-matroid $D=(V,\cF)$
with feasible set
\[
  \cF=\{F_1 \cup F_2 \mid F_1 \in \cF_1, F_2 \in \cF_2, F_1 \cap F_2 = \emptyset\}.
\]
Note the disjointness condition (which is distinct from the more
familiar union operation over matroids). Similarly,
the \emph{delta-sum} $D_1 \Delta D_2$ is the delta-matroid $D=(V,\cF)$
with feasible set
\[
  \cF=\{F_1 \Delta F_2 \mid F_1 \in \cF_1, F_2 \in \cF_2\}.
\]
Bouchet~\cite{Bouchet89dam} and Bouchet and Schwärzler~\cite{BouchetS98deltasum}
showed that $D_1 \cup D_2$ and $D_1 \Delta D_2$ are delta-matroids.
Koana and Wahlström~\cite{KW24} showed that the union and delta-sum
of linear delta-matroids, represented over a common field,
are also linear, and linear representations can be constructed
in randomized polynomial time.

\iflong
  \section{Basic delta-matroid problems}
\label{sec:part1}
In this section, we investigate the parameterized complexity of
natural generalizations of ``basic'' delta-matroid search problems
such as \textsc{Delta-matroid Intersection} and
\textsc{Delta-matroid Parity}.

We recall the definition of the \textsc{Delta-matroid Parity} problem. 
Let $D=(V,\cF)$ be a delta-matroid and $\Pi$ a partition of $V$ into pairs. 
For a set $F \subseteq V$, let $\delta_{\Pi}(F) = |\{ P \in \Pi : |F \cap P| = 1 \}|$
denote the number of pairs broken by $F$, and define
$\delta(D, \Pi) = \min_{F \in \cF} \delta_{\Pi}(F)$.
The goal of \textsc{Delta-matroid Parity} is to find a set $F \in \cF$ with
$\delta_\Pi(F)=\delta(D,\Pi)$. 
In our generalizations, we will focus on the harder case that $\delta(D,\Pi)=0$
where we are looking for a solution $F$ with $|F|=k$. This version can
be solved in randomized polynomial time given a (projected) linear
representation of $D$~\cite{KW24}, but a deterministic polynomial-time
algorithm is open.

In \textsc{Delta-matroid Intersection} the input is two delta-matroids
$D_1=(V,\cF_1)$ and $D_2=(V,\cF_2)$ and we are seeking a set $F$ that is
feasible both in $D_1$ and $D_2$. Again, we will focus on the case
where we have an additional requirement that $|F|=k$ for an integer
$k$ given in the input, and we will assume that $D_1$ and $D_2$ are
(projected) linear delta-matroids provided with representations over a
common field. Again, this version can be solved in randomized
polynomial time~\cite{KW24}.

We investigate the parameterized complexity of generalizations of
these problems. We focus on the following variants, as surveyed in Section~\ref{sec:ourresults}.
Restricted to linear or projected linear inputs, they are as follows. 
\begin{itemize}
\item \textsc{$q$-Delta-matroid Intersection}:
  Given $q$ (projected) linear delta-matroids $D_i=(V,\cF_i)$, $i \in [q]$
  represented over a common field $\F$, and an integer $k \in \N$,
  is there a set $F \subseteq V$ with $|F|=k$ such that $F \in \cF_i$
  for every $i \in [q]$?
\item \textsc{$q$-Delta-matroid Parity}:
  Given a (projected) linear delta-matroid $D=(V,\cF)$ provided with linear
  representation, a partition $\cP$ of $V$ into blocks of size $q$,
  and $k \in \N$, is there a union of $k$ blocks from $\cP$ that is
  feasible in $D$? 
\item \textsc{Delta-matroid Set Packing}:
  Given a (projected) linear delta-matroid $D=(V,\cF)$ provided with linear
  representation, a partition $\cP$ of $V$ into arbitrary blocks,
  and $k \in \N$, is there a set $F \in \cF$ with $|F|=k$ which is the 
  union of blocks from $\cP$?
\end{itemize}
We also recall the following special cases when $q=3$:
\begin{itemize}
\item \textsc{DDD Intersection} refers to \textsc{3-Delta-matroid Intersection}
\item \textsc{DDM Intersection} refers to the special case of
  \textsc{DDD Intersection} where $D_3$ is a linear matroid
\item \textsc{D$\Pi$M Intersection} is \textsc{Delta-Matroid Parity}
  with an additional parity constraint $\Pi$, i.e.,
  \textsc{DDM Intersection} when $D_2$ is a pairing delta-matroid
\end{itemize}
Over linear matroids, these problems have a status as follows.
\textsc{$q$-Matroid Intersection} for linear matroids can be solved in
randomized time $O^*(2^{(q-2)k})$ over a field of characteristic 2, and
randomized time and space $O^*(2^{qk})$ more generally~\cite{EKW23}.
The time for general fields can be improved to $O^*(4^k)$ for $q=3$~\cite{EKW23,BrandKS23}
and for $q=4$ over characteristic 0~\cite{BrandKS23}.
For the remaining problems, the fastest algorithm runs in $O^*(2^{qk})$
respectively $O^*(2^k)$~\cite{EKW23}.
All of these algorithms are randomized.
There are also somewhat slower, but single-exponential deterministic
algorithms over characteristic 0~\cite{BrandP21}
and over arbitrary fields~\cite{FominLPS16JACM}.

We consider the two basic parameters $k$ (cardinality of the solution)
and $r$ (the rank of the involved delta-matroid(s)). We also consider
the special case where some $D_i$'s are in fact basis
delta-matroids of linear matroids.

The algorithm for the matroid case of the above, using determinantal
sieving, works as follows. Let $V=\{v_1,\ldots,v_n\}$
be the ground set and $X=\{x_1,\ldots,x_n\}$ a set of variables. For
\textsc{$q$-Matroid Intersection}, on matroids $M_1, \ldots, M_q$,
we can use the Cauchy-Binet formula to efficiently evaluate a
polynomial $P(X)$ whose monomials enumerate the elements of the
intersection $M_1 \cap M_2$, where we can sieve for those monomials
which in addition are bases of $M_3, \ldots, M_q$ (see~\cite{EKW23}).
For the remaining problems, we replace $P(X)$ by the much simpler
polynomial whose monomials enumerate unions of blocks from $\cP$
and apply sieving for all the matroids $M_1, \ldots, M_q$.

Over linear delta-matroids, our findings are the following. To begin
with, under a pure cardinality parameter all the problems above are
either in P (for $q=2$) or W[1]-hard. This holds even in relatively
restricted cases, including twisted matroids (see Section~\ref{sec:dm-part1-hardness}).
On the other hand, under a rank parameter we recover all the above
results over delta-matroids, by just replacing the cardinality parameter $k$
by the rank parameter $r$. In fact, the strategy is the same as for
linear matroids. First, we generalize the determinantal sieving result
(which is phrased in terms of sieving for bases of a linear matroid)
to sieving for sets that are feasible in a linear delta-matroid of
rank $r$; this can be done with the same dependency on $r$ as
determinantal sieving has on $k$. Next, we construct a polynomial
which enumerates solutions to \textsc{Delta-matroid Intersection},
and sieve for the remaining constraints over this polynomial.
This polynomial is simply the Pfaffian of a skew-symmetric matrix
constructed by Koana and Wahlström~\cite{KW24} as a linear
representation of delta-sums of linear delta-matroids.

To give two further observations, we first note that the result only
requires that the delta-matroid that is plugged into the sieving
procedure has bounded rank; the initial two layers $D_1$ and $D_2$
that are used to construct $P(X)$ can be arbitrary. Second, if any
layer $D_i$ is a linear matroid, then naturally we can simply compute
its $k$-truncation and replace its rank by $k$ in the above-mentioned
results.

In particular, for the problems with $q=3$ the status is as follows.
\begin{itemize}
\item \textsc{DDD Intersection} is W[1]-hard parameterized by $k$,
  but FPT parameterized by $\rank(D_3)$
\item \textsc{DDM Intersection} and \textsc{D$\Pi$M Intersection}
  are FPT parameterized by $k$
\end{itemize}

An additional parameter that has been considered for delta-matroids is
the \emph{spread}~\cite{Moffatt19deltamatroids}. The spread of a
delta-matroid $D=(V,\cF)$ is $\max_{F \in \cF} |F|-\min_{F \in \cF} |F|$, 
and it is zero precisely when $D$ is the basis delta-matroid of a matroid.
However, if the spread of $D$ is $s$, then the rank is at most $k+s$,
hence the rank is the stronger parameter. On the other hand,
parameterizing purely by spread is unreasonable since
\textsc{3-Matroid Intersection} (with spread 0) is NP-hard. 
Thus the spread is not a useful parameter for the problems listed
above, but it could potentially be useful in other questions.
For example, is \textsc{Delta-matroid Intersection} FPT parameterized
by the spread, even if the delta-matroids are not linear?
Note that this does not allow the trick of reducing from
\textsc{Delta-matroid Parity} since the pairing delta-matroid has
the largest possible spread $|V|$. 

\paragraph*{Structure of the section.}
We demonstrate the enumerating polynomial in Section~\ref{sec:dm-gf};
we show the linear delta-matroid sieving procedure in Section~\ref{sec:dm-sieve};
we show FPT results in Section~\ref{sec:dm-part1-fpt} and
hardness results in Section~\ref{sec:dm-part1-hardness}.

\subsection{Generating functions for delta-matroid problems}
\label{sec:dm-gf}
In preparation for the positive results of this section, we present efficiently computable 
multivariate generating functions (i.e., \emph{enumerating polynomials}, in the terminology of Eiben et al.~\cite{EKW23})
for the \textsc{Delta-matroid Intersection} and \textsc{Delta-matroid Parity} problems over linear delta-matroids.
This follows from the algebraic algorithms of Koana and Wahlström~\cite{KW24}.

To reduce notational overhead, we establish some terms. Let $\cS \subseteq 2^V$ be a set system.
Let $X=\{x_v \mid v \in V\}$ be a set of variables. An \emph{enumerating polynomial} for $\cS$ over a field $\mathbb{R}$
is a polynomial
\[
  P_\cS(X) = \sum_{S \in \cS} c_S \prod_{v \in S} x_v,
\]
where $c_S \in \mathbb{F}$ is some non-zero constant.
We base the results on the following support lemma (derived from the construction of linear representation for $D_1 \Delta D_2$~\cite{KW24}).

\begin{lemma} \label{lm:delta-combo-polynomial}
  Let $D_1=(V,\cF_1)$ and $D_2=(V,\cF_2)$ be linear delta-matroids provided as representations over a common field $\F$.
  Let $X=\{x_{v,i} \mid v \in V, i \in [3]\}$ be a set of variables.
  There is a skew-symmetric matrix $A$ indexed by $V \cup T$ for a new set of elements $T$
  such that for every $S \subseteq V$ we have
  \[
    \Pf A[S \cup T] = \sum_{F_1 \in \cF_1, F_2 \in \cF_2 \colon F_1 \Delta F_2=S}
    f(F_1,F_2) \prod_{v \in F_1 \cap S} x_{v,1} \prod_{v \in F_2 \cap S} x_{v,2} \prod_{v \in F_1 \cap F_2} x_{v,3},
  \]
  where $f(F_1,F_2)$ is a non-zero constant for each $F_1 \in \cF_1$, $F_2 \in \cF_2$.   
\end{lemma}
\begin{proof}
  Let $D_1^*=\bD(A_1)/T_1$ and $D_2^*=\bD(A_2)/T_2$ be contraction representations of the duals
  of $D_1$ and $D_2$, where $T_1 \cap T_2=\emptyset$. This can be constructed in matrix multiplication
  time~\cite{KW24}. 
  Define a matrix $A$ indexed by $V \cup V_1 \cup T_1 \cup V_2 \cup T_2$
  where $V_1, V_2$ are copies of $V$ such that
  $A[V_1 \cup T_1]=A_1$, $A[V_2 \cup T_2]=A_2$,
  and for every $v \in V$ (with corresponding copies $v^1 \in V_1$ and $v^2 \in V_2$)
  let $A[v,v^1]=x_{v,1}$, $A[v,v^2]=x_{v,2}$ and $A[v^1,v^2]=x_{v,3}$,
  completed so that $A$ is skew-symmetric and all other positions are 0.
  We claim that $A$ is the matrix we need, with $T=V_1 \cup T_1 \cup V_2 \cup T_2$.
  Indeed, as in~\cite{KW24}, we view $A$ as the sum of the matching delta-matroid $A_H$
  whose entries correspond to variables $X$ and $A'$ which retains the copies of $A_1$ and $A_2$.
  Then by the Pfaffian sum formula in Lemma~\ref{lemma:sum-pf}, for $S \subseteq V$,
  \begin{align*}
    \Pf A[S \cup T] &= \sum_{S=S_1 \uplus S_2} \sum_{S' \subseteq V \setminus S}
    \sigma_{S, S'} \prod_{v \in S_1} x_{v,1} \prod_{v \in S_2} x_{v,2} \prod_{v \in S'} x_{v,3} \cdot
                      \Pf A_1[S_1' \cup T_1] \cdot
                      \Pf A_2[S_2' \cup T_2],
  \end{align*}
  where $\sigma_{S,S'} \in \{ 1, -1 \}$, $\uplus$ denotes disjoint union,
  $S_1'=V_1 \setminus \{v^1 \mid v \in S_1 \cup S'\}$
  and $S_2'=V_2 \setminus \{v^2 \mid v \in S_2 \cup S'\}$.
  A term $(S_1,S_2,S')$ of this sum contributes non-zero precisely when
  $S_1 \cup S' \in \cF_1$ and $S_2 \cup S' \in \cF_2$, since $A_1$ and $A_2$
  represent the duals of $D_1$ and $D_2$. Thus letting $F_1=S_1 \cup S'$
  and $F_2=S_2 \cup S'$ we have $S'=F_1 \cap F_2$, $S_1=F_1 \setminus F_2$
  and $S_2=F_2 \setminus F_1$ as expected. 
\end{proof}

It is easy to specialize this construction into providing the enumerating polynomials we need.

\begin{corollary} \label{cor:intersection-enum}
  There are enumerating polynomials for (projected) linear delta-matroid intersection  
  and parity.
\end{corollary}
\begin{proof}
  For the former, let $D_1=D_1'|Z_1$ and $D_2=D_2'|Z_2$ be projections (possibly with $Z_1, Z_2=\emptyset$)
  of linear delta-matroids $D_1'$, $D_2'$ and apply Lemma~\ref{lm:delta-combo-polynomial}
  to $D_1$ and the dual $D_2^*$. Let $(A,T)$ be the result. Set $x_{v,3}=0$ and $x_{v,2}=1$ for every $v \in V$
  and evaluate $\Pf A$ on the result. Let $X_i=\{x_{v,i} \mid v \in V\}$, $i=1, 2$. 
  This restricts the sum to be over pairs $(S_1,S_2)$ where $S_1 \cup S_2=V$
  and $S_1 \cap S_2=\emptyset$. Since the construction is performed over the dual of $D_2$,
  this corresponds to enumeration over elements $S_1 \in \cF(D_1) \cap \cF(D_2)$
  and $\Pf A$ as a polynomial over the remaining variables $X_1$ is the polynomial we seek.
  For the second, let $(D,\Pi)$ be an instance of \textsc{Delta-matroid Parity}.
  Let $D_\Pi$ be the pairing delta-matroid on pairs $\Pi$.
  Then perfect delta-matroid parity solutions correspond to members of $D \cap D_\Pi$. 
  For the projected variants, in both cases we can reduce the inputs to elementary projections
  (i.e., $Z_1=\{z_1\}$ and $Z_2=\{z_2\}$) and explicitly enumerate over contracting or
  deleting each of $z_1$ and $z_2$ before proceeding as above. 
\end{proof}

\subsection{Delta-matroid sieving}
\label{sec:dm-sieve}

In this section, we give a generalization of determinantal sieving \cite{EKW23}.
To that end, we first propose a ``sparse'' representation for delta-matroids.

\paragraph*{Sparse representation.}

For a delta-matroid of rank $r$, its linear representation may involve $\Omega(n^2)$ non-zero entries, even if $r = O(1)$.\footnote{Consider a delta-matroid, where all sets of size 0 or 2 are feasible. The most straightforward representation, in our view, is a matrix $xy^T - yx^T$, where $x$ and $y$ are random vectors of dimension $n$.}
We give a sparse representation that contains $O(rn)$ non-zero entries, in the context of twist representation as well as contraction representation.

First, let us consider twist representation.
Given a delta-matroid $D = \bD(A) \Delta S$ in twist representation, let $F_{\max}$ be a maximum feasible set. 
As $A[S \Delta F_{\max}]$ is non-singular, the pivoting $A^* = A * (F_{\max} \Delta S)$ is well-defined, and $D = \bD(A) \Delta (F_{\max} \Delta S) \Delta F_{\max} = \bD(A^*) \Delta F_{\max}$ by Lemma~\ref{lemma:tucker}.
We claim that $A^*[V \setminus F_{\max}] = O$.
Assume for contradiction that there is a pair $\{ v, v' \} \subseteq V \setminus F_{\max}$ such that $A^*[v, v'] \ne 0$.
This implies that $\{ v, v' \}$ is feasible in $\bD(A^*)$, and consequently, $F_{\max} \cup \{ v, v' \}$ is feasible in $D$.
However, this contradicts the maximality of $F_{\max}$.

For contraction representation, we
transform the sparse representation into a contraction representation via Equation~\eqref{eq:twist-to-contraction}.
We then obtain $D = \bD(A') / T$, where
\begin{align*}
  \kbordermatrix{
    & T & S & V \setminus S \\
    T & A[S] & -I & -A[S, V \setminus S] \\
    S & I & O & O \\
    V \setminus S & -A[V \setminus S, S] & O & O \\
  }
\end{align*}

\paragraph*{Delta-matroid sieving.}

As in determinantal sieving, polynomial interpolation and inclusion-exclusion will be crucial for delta-matroid sieving as well:

\begin{lemma}[Interpolation]
  \label{lemma:interpolation}
  Let $P(z)$ be a polynomial of degree $n - 1$ over a field $\F$.
  Suppose that $P(z_i) = p_i$ for distinct $z_1, \cdots, z_n \in \F$.
  By the Lagrange interpolation, 
  \[
    P(z)
    = \sum_{i \in [n]} p_i \prod_{j \in [n] \setminus \{ i \}} \frac{z - z_j}{z_i - z_j}.
  \]
  Thus, given $n$ evaluations $p_1, \dots, p_n$ of $P(z)$, the coefficient of $z^t$ in $P(z)$ for every $t \in [n]$ can be computed in polynomial time.
\end{lemma}

\begin{lemma}[Inclusion-exclusion \cite{Wahlstrom13STACS}]
  \label{lemma:inclusion-exclusion}
  Let $P(Y)$ be a polynomial over a set of variables $Y = \{ y_1, \cdots, y_n \}$ and a field of characteristic two.
  For $T \subseteq [n]$, $Q$ be a polynomial identical to $P$ except that the coefficients of monomials not divisible by $\prod_{i \in T} y_i$ is zero.
  Then, $Q = \sum_{I \subseteq T} P_{-I}$, where $P_{-I}(y_1, \cdots, y_n) = P(y_1', \cdots, y_n')$ for $y_i' = y_i$ if $i \notin I$ and $y_i' = 0$ otherwise. 
\end{lemma}

We will also utilize the fact that the Pfaffian of a sparse representation can be expressed as a sum of the product of Pfaffian and determinant:

\begin{lemma} \label{lemma:pf-det-decomposition}
  For a skew-symmetric matrix $A$ indexed by $V$ and a matrix $B$ whose rows and columns are indexed by $V'$ and $V$, respectively, with $|V| \le |V'|$,
  \begin{align*}
    \Pf C = \sum_{U \subseteq V,\, |U| = |V| - |V'|} \sigma_{U} \Pf A[U] \cdot \det B[V \setminus U, V'], 
    \text{ where }
    C = \begin{pmatrix}
      A & B \\ -B^T & O
    \end{pmatrix}
  \end{align*}
  and $\sigma_{U} \in \{ 1, -1 \}$.
\end{lemma}
\begin{proof}
  Letting
  \begin{align*}
    C = C_1 + C_2, \text{ where }
    C_1 = \begin{pmatrix} A & O \\ O & O \end{pmatrix}  \text{ and }
    C_2 = \begin{pmatrix} O & B \\ -B^T & O \end{pmatrix}
  \end{align*}
  Lemmas~\ref{lemma:sum-pf} yields
  \begin{align*}
    \Pf C = \sum_{U \subseteq V \cup V'} \sigma_U \Pf C_1[U] \cdot \Pf C_2[(V \cup V') \setminus U].
  \end{align*}
  Observe that $\Pf C_1[U] = 0$ for any subset $U$ where $|U| > |V|$, as the support graph of $C_1$ contains edges exclusively within $V$.
  Similarly, $\Pf C_2[(V \cup V') \setminus U] = 0$ when $|U| < |V|$, since all edges in the support graph of $C_2$ are incident with vertices in $V'$.
  These observations imply that non-zero contributions to the sum occur only when $U$ has size exactly $|V| - |V'|$.
  Consequently, $\Pf C_1[U] = \Pf A[U]$ and $\Pf C_2 [(V \cup V') \setminus U] = \det B[V \setminus U, V']$ by Lemma~\ref{lemma:det-pf}, which concludes the proof.
\end{proof}

Before presenting our sieving method, let us introduce notations for polynomials, which is in consistent with \cite{EKW23}.
Consider a polynomial $P(X)$ over a set of variables $X = \{ x_1, \cdots, x_n \}$.
A monomial is a product $m = x_1^{m_1} \dots x_n^{m_n}$, where $m_1, \dots, m_n$ are nonnegative integers.
A monomial $m$ is called multilinear if $m_i \le 1$ for each $i \in [n]$.
The \emph{support} of a monomial $m$, denoted by $\supp(m)$, is $\{ i \in [n] \mid m_i
> 0 \}$. 
We also use the notation $X^m$ for the monomial $m=x_1^{m_1} x_2^{m_2} \cdots x_n^{m_n}$,
where its coefficient is excluded.
The coefficient of $m$ in $P$ is denoted by $P(m)$.
Therefore, the polynomial $P$ can be written as $P(X)=\sum_m P(m) X^m$
where $m$ ranges over all monomials in $P(X)$.


\dmsieve*

\begin{proof}
  Our objective is to transform the polynomial $P$ into $Q$ such that (i) every non-multilinear monomial $m$ vanishes, (ii) every multilinear monomial $m$ is translated into $\Pf A[T \cup V_m] \cdot m$, where $V_m = \{ v_i \mid i \in \supp(m) \}$, and (iii) $Q$ can be evaluated in $O^*(2^r)$ time via black-box access to $P$.

  Introduce a set of auxiliary variables $Y=\{ y_t \mid t \in T \}$.
  Let $A_T$ be a $T \times T$ matrix obtained from $A[T]$ by multiplying the row and column $t$ by $y_t$.
  Consider the coefficient of $z^k$ in the characteristic polynomial $\det(A_T + zI)$.
  It is known to be the sum of all $r - k$ principle minors:
  \begin{align*}
    \sum_{U \in \binom{T}{r-k}} \det A_T[U]
    = \sum_{U \in \binom{T}{r-k}} \det A[U] \prod_{t \in U} y_t^2.
  \end{align*}
  Now define $P_T(X, Y)$ be its square root.\footnote{Square roots are well-defined over a field of characteristic 2 because $\alpha^2 = \beta^2$ implies $\alpha = \beta$.
  Furthermore, square roots can be efficiently computed: If $\F$ is a finite field with $2^m$ elements, the square root of $\alpha$ over $\F$ is $\alpha^{2^{m-1}}$ because $(\alpha^{2^{m-1}})^2 = \alpha \cdot \alpha^{2^m-1} = \alpha$ by Fermat's little theorem.}
  This equals
  \begin{align*}
    P_T(X, Y) = \sum_{U \in \binom{T}{r - k}} \Pf A[U] \prod_{t \in U} y_t,
  \end{align*}
  since $(\Pf A[U])^2 = \det A[U]$ and $\sum_i \alpha_i^2 = (\sum_i \alpha_i)^2$ over a field of characteristic 2.

  Define another polynomial $P_V(X, Y)$ by
  \[
    P_V(X,Y) = P\left(x_1 \sum_{t \in T} y_t A[t,v_1], \ldots, x_n \sum_{t \in T} y_t A[t, v_n]\right).
  \]
  Furthermore, define $Q(X,Y)$ as the result of extracting terms from $P'(X, Y) = P_T(X, Y) \cdot P_V(X, Y)$ that have at least degree one in each variable $y_i$ for $i \in [k]$,
  and set $Q(X)=Q(X,1)$.
  Note that $Q$ can be evaluated from $2^r$ evaluations of $P'$, using inclusion-exclusion in Lemma~\ref{lemma:inclusion-exclusion}.
  We aim to show that every monomial $m$ is transformed into $\Pf A[T \cup V_m] \cdot m$.
  To that end, let us examine $P_V(X, Y)$ first.
  For every monomial $m=x_1^{m_1} \cdots x_n^{m_n}$ in $P$,
  the coefficient of $m$ in $P_V$ is given by
  \[
    P_V(m) = P(m) \cdot \prod_{v_i \in V_m} \left( \sum_{t \in T} y_t A[t, v_i]\right)^{m_i},
  \]
  where $P(m)$ is the coefficient of $m$ in $P$.
  In particular, the coefficient of $\prod_{t \in W} y_t$ for $W \in \binom{T}{k}$ in this expression is:
  \[
  P(m) \cdot \sum_{\sigma} \left( \prod_{t \in W} A[t, \sigma(t)] \right),
  \]
  where the sum is over all mappings $\sigma \colon W \to V_m$ such that for each $i \in \supp(m)$, exactly $m_i$ elements of $W$ are that mapped to $v_i$.
  This is equivalent to the determinant of a matrix $A_m$ whose columns comprise $m_i$ copies of $A[T, v_i]$ for each $v_i \in V_m$.
  Consequently, this term vanishes if $m$ is not multilinear.
  For multilinear terms $m$, it simplifies to $P(m) \cdot \det A[T, V_m]$.
  Thus, the part of $P_V(m)$ that is multilinear in $Y$ is
  \begin{align*}
    P(m) \cdot \sum_{W \in \binom{T}{k}} \det A[W, V_m] \prod_{t \in W} y_t.
  \end{align*}
  Note that the non-multilinear part of $P_V(m)$ is irrelevant since it will be excluded from $Q$ in the inclusion-exclusion step.
  
  Multiplying this with $P_T(X, Y)$ yields
  \begin{align*}
    P(m) \cdot \sum_{U \in \binom{T}{r - k}, W \in \binom{T}{k}} \Pf A[U] \cdot \det A[T, V_m] \cdot \left(\prod_{t \in U} y_t \right) \left(\prod_{t \in W} y_t \right)
  \end{align*}
  The coefficient of $\prod_{t \in T} y_t$ in this expression, which equals $Q(m)$, is
  \begin{align*}
    P(m) \cdot \sum_{W \in \binom{T}{k}} \Pf A[T \setminus W] \cdot \det A[W, V_m] = P(m) \cdot \Pf A[T \cup V_m],
  \end{align*}
  where the equality follows from Lemma~\ref{lemma:pf-det-decomposition}.
  Finally, evaluating $Q(X)$ at random coordinates yields the desired result by the Schwartz-Zippel lemma.
  Note that $Q(X)$ remains homogeneous of degree $k$. 
\end{proof}

\paragraph*{Sieving over general fields.}

Following prior work~\cite{BrandDH18,EKW23}, we will use the exterior algebra to develop a sieving method applicable to general fields.
For a field $\F$, $\Lambda(\F^T)$ is a $2^{|T|}$-dimensional vector space, where each basis $e_I$ corresponds to a subset $I \subseteq T$.
Each element $a = \sum_{I \subseteq [k]} {a_I} e_I$ is called an \emph{extensor}.
We denote the vector subspace spanned by bases $e_I$ with cardinality $|I| = i$ for $i \in \{ 0, \dots, |T| \}$, by $\Lambda^i(\F^T)$.
Notably, $\Lambda^0(\F^k)$ is isomorphic to $\F$, and $\Lambda^1(\F^k)$ is isomorphic to the vector space $\F^k$.
Addition in $\Lambda(\F^k)$ is defined element-wise, while
multiplication, referred to as \emph{wedge product}, is defined as follows:
when $I$ and $J$ intersect, the wedge product yields zero;
otherwise, $e_I \wedge e_J = (-1)^{\sigma(I, J)} e_{I \cup J}$, where $\sigma(I, J) = \pm 1$ is the sign of the permutation mapping the concatenation of $I$ and $J$ into the increasing sequence of $I \cup J$.
For vectors $v, v' \in \F^k$, the wedge product exhibits self-annihilation $v \wedge v = 0$ and anti-commutativity $v \wedge v' = -v' \wedge v$. 
For a $T \times T$-matrix $A$, the wedge product of the column vectors yields $\det A \cdot e_T$.
We refer interested readers to prior works~\cite{BrandDH18,EKW23} for a more accessible introduction to exterior algebra used in this work.


Now we discuss the complexity of operations in exterior algebra.
The addition of two extensors can be done by $2^{|T|}$ field operations.
The wedge product $a \wedge a'$ of two extensors $a \in \Lambda(\F^k)$ and $b \in \Lambda^i(\F^k)$ involves $2^k \binom{k}{i}$ field operations, according to the definition (resulting in $O^*(2^k)$ time complexity for $i \in O(1)$).
More generally, an $O(2^{\omega k / 2})$-time algorithm for computing the wedge product is presented by W\l{}odarczyk~\cite{Wlodarczyk19}, where $\omega < 2.372$ is the matrix multiplication exponent.
For a detailed exposition, one may refer to Brand's thesis~\cite{Brand19thesis}.

To address issues arising from non-commutativity in exterior algebra, we use the \emph{lift mapping} $\bar{\phi} \colon \Lambda(\F^T) \to \Lambda(\F^{T \cup T'})$, where $T' = \{ t' \mid t \in T \}$ and 
\vspace{-3ex}
\begin{align*}
  \bar{\phi}(v) = v_1 \wedge v_2 \text{ for } v_1 =\,
  \kbordermatrix{
    & \\
    T & v \\ 
    T' & 0 \\
  } \, \text{ and } \,
  v_2 =\, \kbordermatrix{
    & \\
    T & 0 \\ 
    T' & v \\
  }
\end{align*}
This mapping has been useful in previous studies \cite{Brand19,BrandDH18,EKW23}.
The subalgebra generated by the image of $\bar{\phi}$ exhibits commutativity, i.e., $\bar{\phi}(v) \wedge \bar{\phi}(v') =  \bar{\phi}(v') \wedge \bar{\phi}(v)$.

For an element $t'$ in $T'$, let $t''$ denote the corresponding element $t$ in $T$. Similarly, for a subset $W'$ of~$T'$, let $W'' = \{ t'' \mid t' \in W' \}$.

We will assume that the polynomial is represented by an \emph{arithmetic circuit}.
It is a directed acyclic graph with one sink node (referred to as output gate) in which every source node is labelled with either a variable $x_i$ or an element of $\F$ (referred to as input gate) and all other nodes are either addition gates or multiplication gates.
It is further assumed that every sum and product gate has fan-in 2.

\begin{theorem} \label{theorem:sieving-general}
  Let $D=(V,\cF)$ be a (projected) linear delta-matroid of rank $r$, $V=\{v_1,\ldots,v_n\}$ given as a sparse representation $\bD(A) / T$,
  and let $P(X)$ be a homogeneous polynomial of degree $k$ given as an arithmetic circuit over
  sufficiently large field, $X=\{x_1,\ldots,x_n\}$.
  In time $O^*(2^{\omega r})$ we can sieve for those terms of $P(X)$ which
  are multilinear in $X$ and whose support $\{x_i \mid i \in I\}$
  in $X$ is such that $\{v_i \mid i \in I\}$ is feasible in $D$. 
\end{theorem}

\begin{proof}
We evaluate the circuit over the subalgebra of $\Lambda(\F^{T \cup T'})$ by substituting every variable $x_i$ with $x_i a_i$, where $a_i = \phi(A[T, v_i])$.
Let $r \in \Lambda^k(\F^{T \cup T'})$ denote the resulting extensor.
Note that with each variable $x_i$ substituted with an element from $\F$, the extensor $r$ can be computed in $O^*(2^{\omega k})$ time.

For each monomial $m$ in $P$, there is a term containing $e_{W \cup W'}$ for $W \in \binom{T}{k}$ and $W' \in \binom{T'}{k}$:
\begin{align*}
  \bigwedge_{v_i \in V_m} \left( \sum_{t \in W} A[t, v_i] e_t \wedge \sum_{t' \in W'} A[t'', v_i] e_{t'} \right)  = (-1)^{\binom{r - k}{2}} \cdot \det A[W, V_m] \cdot \det A[W'', V_m] \cdot e_{W \cup W'},
\end{align*}
where $V_m = \{ v_i \mid i \in \supp(m) \}$.
In particular, all non-multilinear terms vanish.
Consequently, the coefficient of $e_{W \cup W'}$ in $r$ is given by
\begin{align*}
  (-1)^{\binom{r-k}{2}} \sum_{m} P(m) \cdot \det A[W, V_m] \cdot \det A[W'', V_m],
\end{align*}
where the sum is taken over all multilinear monomials $m$ of $P$.

We compute the Pfaffian $A[U]$ for each $U \in \binom{T}{r-k}$ in $O^*(2^r)$ time.
For each partition $U \cup W = T$, and each partition $U' \cup W' = T'$ with $|W| = |W'| = k$, we multiply the coefficient of $e_{W \cup W'}$ by $(-1)^{\binom{r-k}{2}} \sigma_{U} \sigma_{U''} \Pf [U] \cdot \Pf A[U'']$, where $\sigma_U, \sigma_{U''} \in \{ 1, -1 \}$ are defined as in Lemma~\ref{lemma:pf-det-decomposition}.
We then sum up the products.
The resulting sum, denoted by $Q(X)$, contains the following term for each monomial in $P$:
\begin{align*}
  P(m) \sum_{U, U'' \in \binom{T}{r - k}} \sigma_U \sigma_{U''} \Pf A[U] \cdot \Pf A[U''] \cdot \det A[T \setminus U, V_m] \cdot \det A[T \setminus U'', V_m]. 
\end{align*}
This equals $P(m) \cdot \det A[T \cup V_m]$ since by Lemma~\ref{lemma:pf-det-decomposition}, 
\begin{align*}
  \det A[T \cup V_m] = (\Pf A[T \cup V_m])^2
  = \left( \sum_{U \in \binom{T}{r - k}} \sigma_U \Pf A[U] \cdot \det A[T \setminus U, V_m] \right)^2,
\end{align*}
which matches  the above expression.
It follows that $Q(X) = \sum_{m} P(m) \cdot \det A[T \cup V_m] \cdot m$.
Evaluating $Q(X)$ at random coordinates yields the desired result by the Schwartz-Zippel lemma.
\end{proof}

We remark that when the arithmetic circuit is \emph{skew}, i.e., every multiplication is connected from an input gate, the running time improves to $O^*(4^r)$.

\subsection{Delta-matroid FPT results}
\label{sec:dm-part1-fpt}

We note the FPT consequences of the sieving results. The following is
a direct generalization of the fastest known algorithms for linear
matroid intersection~\cite{EKW23}. 

\begin{lemma} \label{lm:ddd-intersection}
  Let $D_1=(V,\cF_1), \ldots, D_q=(V,\cF_q)$ be (projected) linear delta-matroids represented over a common field $\F$,
  let $k \in \N$ and $r=\max_{i \geq 3} \rank D_i$. Then a feasible set of cardinality $k$
  in the common intersection $D_1 \cap \ldots \cap D_q$ can be found in randomized time and space $O^*(2^{O(qr)})$.
  In particular, if $\F$ is of characteristic 2 then this can be reduced to time $O^*(2^{(q-2)r})$ and polynomial space. 
\end{lemma}
\begin{proof}
  We assume $q \geq 3$ as otherwise the problem is in P~\cite{GeelenIM03,KW24}.
  Let $D'=(V',\cF)$ be the direct sum $D_3 \uplus \ldots \uplus D_q$, with
  $V'=V_3 \uplus \ldots \uplus V_q$ where each $V_i$ is a separate copy of $V$. 
  For each $v \in V$ and $i \in [q]$, $i \geq 3$, let $v^i$ denote the copy of $v$ in $V_i$ (and let $v^1=v^2=v$). 
  Then $D'$ is a linear delta-matroid of rank $(q-2)r$.
  Reduce $D_1$, $D_2$ and $D'$ to elementary projections over elements $z_1$, $z_2$ and $z_3$
  and guess for each $z_i$, $i=1, 2, 3$ whether $z_i$ is a member of the solution or not.
  For each such guess, correspondingly delete or contract $z_i$ and proceed as follows
  with the resulting linear delta-matroids. 
  Let $X=\{x_v \mid v \in V\}$ and let $P(X)$ be the enumerating polynomial for the intersection $D_1 \cap D_2$ of Cor.~\ref{cor:intersection-enum}.
  For each $v \in V$ and $i=3, \ldots, q$ let $x_{v,i}$ be a new variable associated with the element $v^i$ in $D'$
  and let $X'=\bigcup_{v \in V} \{x_{v,3}, \ldots, x_{v,q}\}$ be the set of all these variables.
  Let $P'(X')$ be $P(X)$ evaluated with $x_v=\prod_{i=3}^q x_{v,i}$ for each $v \in V$.
  By Theorems \ref{theorem:sieving-char-2} and \ref{theorem:sieving-general}, we can sieve in $P'(X')$ for a term that is feasible in $D'$
  by a randomized algorithm
  in time and space $O^*(2^{\omega (q-2)r})$ in the general case, and time $O^*(2^{(q-2)r})$ and polynomial space
  if $\F$ is over char.~2, as promised. In particular, for the success probability, 
  we can choose to work over a sufficiently large extension field
  of the given field at no significant penalty to the running time.
\end{proof}

In other words, \textsc{$q$-Delta-matroid Intersection} on linear
delta-matroids is FPT parameterized by the rank
(or even, parameterized by the third largest rank of the participating delta-matroids).

In summary, we get the following algorithmic implications (repeated
from Section~\ref{sec:ourresults}). Variants for projected linear
delta-matroids also apply, with the same guessing preprocessing step
as in Lemma~\ref{lm:ddd-intersection}.

\dmponefpt*

\begin{proof}
  The first two items follow directly from Lemma~\ref{lm:ddd-intersection}.
  For the third, let $\cP=V_1 \cup \ldots \cup V_n$ be the partition, 
  and define the polynomial
  \[
    P(X,z) = \prod_{i=1}^n (1+\prod_{v \in V_i} zx_v),
  \]
  $X=\{x_v \mid v \in V\}$. Now we can sieve in the coefficient of $z^k$ in $P(X,z)$ for
  a multilinear monomial whose support is feasible in $D$. 
  For items four and five, we can replace a parity constraint $\Pi$ by
  a pairing delta-matroid $D_\Pi$ over $\Pi$, if applicable,
  truncate each matroid $M_i$, $i \geq 3$ to rank $k$, then apply Lemma~\ref{lm:ddd-intersection}.
  In particular, a pairing delta-matroid can be represented over any field. 
\end{proof}

The following variant will be used later in the paper.
Suppose that we are given a linear delta-matroid $D = (V, \mathcal{F})$ represented by $D = \bD(A) / T$ for $A \in \mathbb{F}^{(V \cup T) \times (V \cup T)}$ and a set of edges $E \subseteq \binom{V}{2}$.
A matching $M$ in the graph $G=(V, E)$ is a \emph{delta-matroid matching} if $\bigcup M$ is feasible in $D$.
Furthermore, assume that the edges $E$ are colored in $k$ colors. 
\textsc{Colorful Delta-matroid Matching}
asks for a delta-matroid matching $M$ of $k$ edges whose edges all have distinct colors.
We show the following.

\dmmatching*

\begin{proof}
  Let $(D=(V,\cF), G=(V,E), c \colon E \to [k])$ be an instance of 
  \textsc{Colorful Delta-matroid Matching} as defined above.
  Let $D=\bD(A)/T$ be a representation of $D$. 
  We may assume that $G$ has degree 1: For every $v \in V$ with degree greater than one in $G$,
  introduce $d$ copies $v_1, \dots, v_d$ of $v$ and duplicate rows and columns of $v$ in $A$ for each copy.
  Let $V'$ be the new, larger ground set. Note that delta-matroid matchings are preserved.
  Now the result follows by reduction to \textsc{D$\Pi$M Intersection}:
  Let $M=(V',\cI)$ be the partition matroid where for every edge $uv \in E$ of color $c$,
  one member (say $u$) is placed in a set $S_c$ from which at most one element may be chosen,
  and the other (i.e., $v$) is a free element. Then $M$ is linear over every field.
  Thus, the problem has been reduced to an instance of \textsc{D$\Pi$M Intersection}
  $(D,E,M,k)$ with partition $\Pi=E$. 
\end{proof}

\subsection{Hardness results}
\label{sec:dm-part1-hardness}

To complement the above, we show that \textsc{DDD Intersection} is
W[1]-hard parameterized by $k$, even for some quite restrictive cases. 

\dmhardness*

Specifically, consider a delta-matroid $D=(V,\cF)$ on a ground set
$V=\bigcup_{i=1}^m \{x_i, y_{i,1}, \ldots, y_{i,n_i}\}$
with a feasible set $\cF$ defined as
\[
  F \in \cF \Leftrightarrow \forall i (x_i \in F \Leftrightarrow |\{ j \in [n_i]: y_{i,j} \in F \}| = 1)
\]
for all $F \subseteq V$. 
Let us note some ways to describe $D$.

\begin{lemma} \label{lm:star-forest-matching}
  $D$ is simultaneously a matching delta-matroid (for a star forest),
  a twisted unit partition matroid, and a twisted co-graphic matroid.
  (In particular, $D$ can be represented over every field.)
\end{lemma}
\begin{proof}
  As a matching delta-matroid, it is the matching delta-matroid on
  the star forest where for every $i \in [m]$ there is star with root $x_i$
  and leaves $y_{i,j}$, $j \in [n_j]$. On the other hand, $D \Delta X$
  (for $X=\{x_1,\ldots,x_m\}$) denotes the unit partition matroid where
  for every $i$ an independent set $I$ contains at most one element
  of $\{x_i, y_{i,1}, \ldots, y_{i,n_i}\}$.
  This can also be described as the co-graphic matroid for the graph $G$
  which for every $i \in [m]$ consists of a cycle
  on the set $\{x_i,y_{i,1}, \ldots,y_{i,n_i}\}$. 
  In particular, $M$ is regular~\cite{OxleyBook2}.
  (Indeed, it is represented by a matrix $A \in \{0,1\}^{[m] \times V}$
  where each element $x_i$ and $y_{i,j}$ is represented by the $i$-th unit vector.)
\end{proof}

We refer to $D$ as a \emph{star forest matching delta-matroid}, since this is the most specific description.
We have the following.

\begin{theorem}[Theorem~\ref{hard:ddd}, refined] \label{thm:ddd-hard}
  \textsc{DDD Intersection} is W[1]-hard parameterized by cardinality,
  even for the intersection of a star forest matching delta-matroid
  and two pairing delta-matroids. 
\end{theorem}
\begin{proof}
  We show a reduction from \textsc{Multicolored Clique}.
  Let $G=(V,E)$ be a graph with a partition $V=V_1 \cup \ldots \cup V_k$,
  where the task is to decide whether $G$ contains a clique $C=\{v_1,\ldots,v_k\}$
  where $v_i \in V_i$ for every $i \in [k]$. For $i, j \in [k]$, $i \neq j$,
  let $E_{i,j} \subseteq E$ be the edges between $V_i$ and $V_j$. 
  Assume w.l.o.g.\ that $E=\bigcup_{i \neq j} E_{i,j}$. 
  Define a ground set
  \[
    X \cup Y :=   \{x_{i,j:u} \mid i, j \in [k], i \neq j, u \in V_i\}
    \cup
    \{y_{i,j:u,v} \mid i, j \in [k], i \neq j,
    uv \in E_{i,j}, u \in V_i, v \in V_j\}.
  \]
  Let $D_1$ and $D_2$ be pairing delta-matroids
  such that the intersection $D_1 \cap D_2$ 
  induces a partition of $X$ into blocks $B_{i,u}=\{x_{i,j:u} \mid j \in [k]-i\}$
  and $Y$ into pairs $P_e=\{y_{i,j:u,v},y_{j,i:v,u}\}$, $e \in E$.
  For $D_3$, let $M$ be the partition matroid over blocks
  \[
    C_{i,j,u}=\{x_{i,j:u}\} \cup \{y_{i,j:u,v} \mid uv \in E_{i,j}\}.
  \]
  For $i, j \in [k]$ and $u \in V_i$, 
  let $D_3=M \Delta \{x_{i,j:u} \mid i, j \in [k], i \neq j, u \in V_i\}$.
  Then $F \in \cF(D_3)$ has the effect of an implication
  \[
    x_{i,j:u} \in F \Leftrightarrow \exists v : y_{i,j:u,v} \in F,
  \]
  which finishes the construction with a parameter $k'=2k(k-1)$.
  In particular, $D_3$ is a star forest matching delta-matroid.

  For correctness, on the one hand, let $C=\{v_1,\ldots,v_k\}$ be a multicolored $k$-clique.
  Then
  \[
    F=\{x_{i,j:v_i}, y_{i,j:v_i,v_j}  \mid i \in [k], j \in [k] - i\}  \in D_1 \cap D_2 \cap D_3
  \]
  is a feasible set of cardinality $k'$. In the other direction,
  let $F \in D_1 \cap D_2 \cap D_3$, and assume that $x_{i,j:v_i} \in F$ 
  for some $i, j, v_i$. This must hold, since $y_{i,j:v_i,w} \in F$
  for any $i, j, v_i, w$ implies $x_{i,j:v_i} \in F$ by $D_3$. Then by
  $D_1 \cap D_2$, $x_{i,a:v_i} \in F$ for every $a \in [k]-i$,
  and by $D_3$ for every $a \in [k]-i$ we have $y_{i,a:v_i,v_a} \in F$ 
  for some $v_iv_a \in E_{i,a}$. By $D_1 \cap D_2$ we also have
  $y_{a,i:v_a,v_i} \in F$ and by $D_3$ we have $x_{a,i:v_a} \in F$.
  Finally by $D_1 \cap D_2$ we have $x_{a,b:v_a} \in F$ for every
  $a, b \in [k]$, $a \neq b$ and by $D_3$ we have $x_{a,b:v_a,v_{b,a}} \in F$
  for some $v_{b,a} \in V_b$ for all pairs $a, b \in [k]$, $a \neq b$.   
  This represents two elements in $F$ for every ordered pair $a, b \in [k]$,
  $a \neq b$, thus since $|F| \leq k'$ this represents an exhaustive
  account of every element in $F$. Thus $v_{b,a}=v_b$ for all $a, b \in [k]$, 
  $a \neq b$, and repeating the argument shows that 
  $C=\{v_i\} \cup \{v_{j,i} \mid j \in [k]-i\}$ forms a clique.
\end{proof}

We note some other hard problem variants. These are justified primarily with the positive results related to
\textsc{Delta-matroid triangle cover}.

\begin{corollary}
  The following problems are W[1]-hard parameterized by $k$.
  \begin{enumerate}
  \item Finding a feasible set of cardinality $k$ in the intersection of three twisted regular matroids
  \item Finding a feasible set of cardinality $k$ in the intersection of three matching delta-matroids.
  \item \textsc{3-Delta-matroid Parity}, i.e., given
    a directly represented linear delta-matroid $D=(V,\cF)$,
    a partition of $D$ into triples, and an integer $k$, 
    find a union of $k$ triples that is feasible in $D$.
    In particular, \textsc{Delta-matroid Triangle Packing} is hard. 
  \item Given a directly represented linear delta-matroid $D=(V,\cF)$,
    a graph $G=(V,E)$, and an integer $k$, find a feasible
    set $F$ of size $k$ in $D$ that can be covered by $K_4$'s in $G$.
  \end{enumerate}
\end{corollary}
\begin{proof}
  The first two follow from Theorem~\ref{thm:ddd-hard}
  via Lemma~\ref{lm:star-forest-matching}.
  The third follows immediately. If $(D_1,D_2,D_3,k)$ is an instance of
  \textsc{DDD Intersection} from Theorem~\ref{thm:ddd-hard},
  then let $D$ be the direct sum of $D_1$, $D_2$ and $D_3$
  on ground set $V=V_1 \cup V_2 \cup V_3$ (where $D$ on $V_i$ is a copy of $D_i$).
  Let the triples consist of the three copies of every element
  of the original ground set. Then the result follows. 
  
  The fourth follows similarly. Note that the input is partitioned,
  so that there are $k$ colour classes of triples and the solution
  to the triangle-packing problem needs to use one triple per colour
  class. We can implement this constraint by adding, for every
  colour class $i$, an artificial element $z_i$ that we add as the
  fourth element to every triple in class $i$. Now the only way to
  cover a feasible set $F$ is to use at most one 4-set per colour
  class, and a feasible set of size $|F|=3k$ that is covered by 4-sets
  must be a disjoint union of triangles as above.
  It is easy to see that the 4-sets in this construction can be
  represented precisely as $K_4$'s in a graph $G$. 
\end{proof}

Finally, we note that finding a feasible $k$-path is W[1]-hard.
This is again in contrast to the case of matroids, where the
corresponding problem is FPT, and matroid-related techniques have been
used in algorithms for the basic \textsc{$k$-Path}
problem~\cite{EKW23,FominLPS16JACM} (see also~\cite{FominGKSS23soda}).

\begin{theorem}
  Given a linear delta-matroid $D=(V,\cF)$, a graph $G=(V,E)$ and an integer $k \in \N$, 
  it is W[1]-hard to find a $k$-path in $G$ whose vertex set is feasible in $D$. 
\end{theorem}
\begin{proof}
  Let $G=(V,E)$ with $V=V_1 \cup \ldots \cup V_k$ be
  an input to \textsc{Multicolored Clique}.
  
  Create a vertex set $U \cup W \cup \{s,t\}$ where
  $U=\{u_{i,j,v} \mid i, j \in [k], i \neq j, v \in V_i\}$ and
  $W=\{w_{i,j,e} \mid i, j \in [k], i \neq j, e \in E \cap (V_i \times V_j)\}$.
  Add edges as follows:
  \begin{itemize}
  \item $su_{1,2,v}$ for every $v \in V_1$
  \item $u_{i,j,v}w_{i,j,e}$ for every $w_{i,j,e} \in W$ where $e=vv' \in E$
  \item $w_{i,j,e}u_{i,j',v}$ for every $w_{i,j,w} \in W$ where $e=vv' \in E$
    where $j'=j+1$ except (1) $j'=j+2$ if $j+1=i<k$ and (2) the edge is
    skipped if $j+1=i=k$  
  \item $w_{i,k,e}u_{i+1,1,v}$ for every $w_{i,k,e} \in W$ and $v \in V_{i+1}$
    where $i<k$
  \item $w_{k-1,k,e}t$ for every $w_{k-1,k,e} \in W$
  \end{itemize}
  Let $D=(U \cup W \cup \{s,t\}, \cF)$ be a delta-matroid with feasible sets meeting the following conditions.
  \begin{itemize}
  \item A pairing constraint over $W$ on all pairs $(w_{i,j,w},w_{j,i,e})$
    representing the same edge $e$
  \item A partition constraint over $U$ on blocks $B_{i,j}=\{u_{i,j,v} \mid v \in V_i\}$,
    $i, j \in [k]$, $i \neq j$
  \item $s$ and $t$ are co-loops, i.e., present in every feasible set
  \end{itemize}
  This can be easily constructed, e.g., as a contraction of a matching
  delta-matroid (where the blocks $B_{i,j}$ are formed from stars
  where the midpoint vertex has been contracted).
  Let $H=(U \cup W \cup \{s,t\}, E_H)$ be the graph just constructed and let $k'=2k(k-1)+2$.
  We claim that $H$ contains a $k'$-path whose vertex set is feasible in $D$
  if and only if $G$ has a multicolored $k$-clique.
  
  On the one hand, let $C=\{v_1,\ldots,v_k\}$ be a
  multicolored $k$-clique in $G$, with $v_i \in V_i$ for every $i \in [k]$.
  Let $F=\{u_{i,j,v_i}, w_{i,j,v_iv_j} \mid i, j \in [k], i \neq j\} \cup \{s,t\}$.
  Then $F$ induces a path in $H$, $|F|=k'$ and $F$ is feasible in $D$. 
  
  On the other hand, let $F \in \cF(D)$ with $|F|=k'$, and let $P$ be a
  path in $H$ such that $F=V(P)$. Then $s, t \in F$, and the shortest
  $st$-path in $H$ contains $k'$ vertices. Hence $F$ contains precisely
  one vertex $u_{i,j,v_{ij}}$ and $w_{i,j,e_{ij}}$ for every $i, j \in [k]$, $i \neq j$.
  Furthermore, for every $j, j' \in [k]-i$ we have $v_{ij}=v_{ij'}$
  since the $w$-vertices $w_{i,j,e_{ij}}$ in $F$ remember the identity of $v_i$.
  Finally, since $F$ is feasible in $D$, we have $e_{ij}=e_{ji}$
  for all $i, j \in [k]$, $i \neq j$.
  Thus, the $w$-vertices contain one selection of an edge $e_{ij} \in V_i \times V_j$
  for every $i, j \in [k]$, $i \neq j$ where $e_{ij}$ and $e_{ij'}$ agree at $V_i$, 
  and $e_{ij}$ and $e_{ji}$ are the same edge. Thus the selection forms
  a $k$-clique in $G$.
\end{proof}

\else
  \input{sec3_short.tex}
\fi
\iflong
  \section{FPT algorithm for delta-matroid triangle cover} \label{sec:triangle-cover}

A triangle packing is a collection of vertex-disjoint triangles $\mathcal{T} = \{ T_1, \cdots, T_{p} \}$.
We denote $V(\mathcal{T}) = \bigcup_{T \in \mathcal{T}} T$.
We say that $\mathcal{T}$ covers $S$ if $S \subseteq V(\mathcal{T})$.
We will solve the following problem. Given a linear delta-matroid $D$
over ground set $V$, a collection $\cT_0 \subseteq \binom{V}{3}$ of
triples (triangles) from $V$, and an integer $k$,
find a feasible set $F \subseteq V$ with $|F|=k$ such that $F$ can be
covered by a packing of triangles from $\cT_0$.
We call this problem \textsc{Delta-matroid Triangle Cover}.
As shown in Theorem~\ref{thm:ddd-hard}, the problem would be W[1]-hard if we insisted that $F$ must
be precisely partitioned into triangles. But here, we allow triangles
to have anywhere between 1 and 3 members in $F$.

The algorithm uses a mixture of algebraic and combinatorial arguments.  
We start with the algebraic side, where we introduce a new delta-matroid operation,  
termed \emph{$\ell$-projection}, that plays a crucial role in our approach.

\subsection{Delta-matroid $\ell$-projection}

We introduce an operation called \emph{$\ell$-projection} for delta-matroids,  
which generalizes the classical concept of matroid truncation.  
This new operation is of independent interest and may also be useful for kernelization  
(see, e.g., \cite{Wahlstrom24SODA}).
For a matroid $M = (V, \cI)$ and $k \in \N$, the \emph{$k$-truncation} of $M$ is a matroid $(V, \cI_k)$, where $\cI_k = \{ I \in \cI \mid |I| \le k \}$.
This operation is fundamental in the design of many (parameterized) algorithms for matroid problems.
See Marx~\cite{Marx09-matroid} and Lokshtanov et al.~\cite{LokshtanovMPS18TALG} for randomized and deterministic algorithms for truncation, respectively.

However, the situation is not as straightforward for delta-matroids.
A natural idea would be to define the $k$-truncation of a delta-matroid $D = (V, \cF)$ as $(V, \cF_k)$, where $\cF_k = \{ F \in \cF \mid |F| \le k\}$.
However, this does not always yield a delta-matroid, e.g., let $D = (V, \cF)$ be the (matching) delta-matroid, where $V = \{ a, b, c, d \}$ and $\cF = \{ \emptyset, \{ a, b \}, \{ c, d \}, \{ a, b, c, d \} \}$.
For $k = 2$, we obtain $\cF_2 = \{ \emptyset, \{ a, b \}, \{ c, d \} \}$, which fails the exchange axiom:
For $F = \{ a, b \}$ and $F' = \{ c, d \}$, there is no element $x \in F \Delta F'$ such that ${F \Delta \{ c, x \} \in \cF_2}$.

To overcome this issue, we introduce a more nuanced operation called \emph{$\ell$-projection}, which generalizes matroid truncation.
Let $D = (V, \mathcal{F})$ be a delta-matroid and let $X \subseteq V$.  
The \emph{$\ell$-projection} of $D$ to $V \setminus X$ is the delta-matroid $D |_{\ell} X = (V \setminus X, \cF |_{\ell} X)$, where
\begin{align*}
  \cF |_{\ell} X = \{ F \setminus X \mid F \in \cF, |F \cap X| = \ell \}.
\end{align*}
It can be verified that $D|_\ell X$ is indeed a delta-matroid. 
In fact, if we let $D_X$ be the basis delta-matroid of the uniform matroid $U_{X,\ell}$ of rank $\ell$ on $X$ (with all elements of $V \setminus X$ treated as loops in $D_X$), then
\[
  D |_\ell X = (D \Delta D_X) \setminus X.
\]
Because both the delta-sum and deletion operations preserve the delta-matroid property, the resulting structure $D |_\ell X$ is indeed a delta-matroid.
Since $D_X$ is representable over any sufficiently large field (see \cite{OxleyBook2}), a linear representation of the $\ell$-projection of a linear delta-matroid can be computed in randomized polynomial time (see \cite{KW24}).  
To be precise, we have the following:

\begin{lemma}
  \label{lemma:reduce}
  Let $D = (V, \mathcal{F})$ be a linear delta-matroid with $D = \bD(A) / T$ for $A \in \mathbb{F}^{(V \cup T) \times (V \cup T)}$. 
  For $X \subseteq V$, an $\varepsilon$-approximate linear representation of the $\ell$-projection of $D$ to $V \setminus X$ over a field extension $\mathbb{F}'$ of $\mathbb{F}$ with at least $\ell / \varepsilon$ elements can be computed in polynomial time.
\end{lemma}

An alternative, perhaps more transparent, perspective on maintaining a representation of $\ell$-projection is to apply the Ishikawa-Wakayama formula (a generalization of the Cauchy-Binet formula; see Lemma~\ref{lemma:cauchy-binet-ss}) to $\Pf BAB^T$,
where $B$ is a representation of a suitable transversal matroid.  
This approach is analogous to the randomized construction of $k$-truncation for linear matroids~\cite{Marx09-matroid} via the Cauchy-Binet formula.  
Nonetheless, from a technical standpoint, the construction described above fully meets our needs.

We illustrate that linear matroid truncation is a special case of the $\ell$-projection.
Let $M = (V, \cI)$ be a linear matroid represented by a matrix $A$.
As mentioned in Section~\ref{ssec:rep}, we can represent the independent sets of $M$ via the delta-matroid $\bD(A_D)|T$, where
\begin{align*}
  A_D = \kbordermatrix{
    & T & V \\
    T & O & A \\
    V & -A^T & O \\
  }
\end{align*}
Furthermore, by Lemma~\ref{lemma:det-pf}, a set $F \subseteq V \cup T$ is feasible in $\bD(A_D)$ if and only if $|F \cap V| = |F \cap T|$ and the submatrix $A[F \cap T, F \cap V]$ is non-singular.
Thus, applying $k$-projection to the ground set $V$ on $\bD(A_D)$ yields a linear representation of a delta-matroid on $V$, where a set $F \subseteq V$ is feasible if and only if there exists a set $T' \subseteq T$ of size $k$ such that $A[T', F]$ is non-singular.
This is exactly the $k$-truncation of $M$.

For notational convenience, we will sometimes use \emph{$\ell$-contraction} by $X$ to mean $\ell$-projection to $V \setminus X$.

\subsection{Algorihtm}

Suppose that there $\cT=\{T_1,\ldots,T_p\}$ is a triangle packing covering a feasible set $F$ of size $k$.
Using the standard color coding argument, we may assume that $V$ is partitioned into $V_1, \dots, V_p$ such that $T_i \subseteq V_i$.
The color coding step adds a multiplicative factor $k^{O(k)}$ to the running time.
Let $\cT_i$ be the collection of triangles $T \in \cT_0$ with $T \subseteq V_i$.
Our algorithm will try to find a packing of at most $|V_i \cap F| \le 3$ triangles from $\cT_i$.

Let us introduce the notion of \emph{flowery} vertices:
We say that a vertex $v$ is \emph{flowery} if it is at the center of at least
$7$ otherwise disjoint triangles in $\cT_i$.
Note that this can be efficiently
tested simply by computing the maximum
matching number in the graph induced by triangles covering $v$.
Let $H$ be the set of flowery vertices.
Other vertices are \emph{non-flowery}.
Let us also define a broader notion of \emph{safe} vertices.
Intuitively, a safe vertex can always be included into a feasible set, as discussed in \Cref{lemma:completable-packing}.
The central idea in our algorithm revolves around utilizing the $\ell$-projection to get rid of safe vertices.
A vertex $v$ is said to be \emph{safe} if one of the following holds:
\begin{itemize}
  \item 
  it is flowery, or
  \item
  there is a triangle $T = \{ u, v, w \}$ such that every triangle in $\mathcal{T}_i$ covering $u$ or $w$ covers $v$.
\end{itemize}
Note that the second condition holds e.g., when it is covered by an \emph{isolated triangle}, a triangle $T$ such that for every $v \in T$, $T$ is the unique triangle in $\cT_i$ covering $v$.

\begin{lemma} \label{lemma:completable-packing}
  Let $V' \subseteq V$ be the set of safe vertices. 
  Let $F$ be a set of size $k$ such that $|F \cap V_i| \le 3$ for every $i \in [p]$.
  Suppose that there is a triangle packing $\cT$ covering $F \setminus V'$, where every $T \in \cT$ belongs to $V_i$ (i.e., $T \subseteq V_i$) for some $i \in [p]$.
  Then, there is a triangle packing covering all of $F$. 
\end{lemma}
\begin{proof}
  Let $v \in (F \cap V') \setminus V(\cT)$.
  First, suppose that $v$ is non-flowery. 
  By definition, there is a triangle $T = \{ u, v, w \}$ such that every triangle covering $u$ or $w$ covers $v$.
  Note that neither $u$ nor $w$ is covered by $\cT$ since otherwise $v$ would be covered as well.
  Thus, $T$ is disjoint from every other triangle in $\cT$ and can safely be
  added to $\cT$.

  Next, suppose that $v \in (F \cap V') \setminus V(\cT)$ is flowery.
  Since $|F \cap V_i| \le 3$, we may assume that the packing $\cT$ restricted to $V_i$ contains at most two triangles, covering at most six vertices.
  Thus there is at least
  one triangle containing $v$ that is disjoint
  from $V(\cT)$ and can be added to $\cT$. The result follows by
  repetition. 
\end{proof}

Using color coding, we guess
a further partition of every $V_i$ as $V_i=V_{i,1} \cup V_{i,2} \cup V_{i,3}$ such that $T_i$ intersects all three parts.
Thus, we will henceforth assume that every triangle in $\cT_i$ intersects $V_{i,j}$ for all $j \in [3]$.
For each $i \in [p]$ and $j \in [3]$, let $v_{i,j}$ be the unique vertex of $T_i \cap V_{i,j}$.
We define the following for every $i \in [p]$ and $j \in [3]$:
\begin{itemize}
\item let $f_{i,j} \in \{ 0, 1 \}$ be a Boolean indicating whether $v_{i,j}$ is in the feasible set $F$ and
\item let $h_{i,j} \in \{ 0, 1 \}$ be a Boolean indicating whether $v_{i,j}$ is flowery.
\end{itemize}
We guess these Booleans; note that there are $2^{O(k)}$ possibilities.

For each non-flowery vertex, we can also guess one edge of its triangle in the triangle packing. 

\begin{lemma} \label{lemma:small-vc}
  If a vertex $v \in V_i$ is non-flowery, then there is a set $S_v \subseteq V_i$
  such that every triangle containing $v$ must intersect $S_v$ and $|S_v| \le 12$.
\end{lemma}
\begin{proof}
  Consider the graph induced by
  triangles in $\cT_i$ intersecting~$v$.
  Since $v$ is non-flowery, this graph has a maximum matching $M$ of size at most 6.
  Then, $S_v = V(M)$ intersects every triangle containing $v$.
\end{proof}

For each non-flowery vertex $v$, we choose a vertex $s(v) \in S_v$ uniformly at random.
We are interested in the case that  $s(v) \in T$  holds for every non-flowery vertex $v$ covered by a triangle $T_i \in \cT$.
By \Cref{lemma:small-vc}, this happens with probability at least $1/2^{O(k)}$.
Let us say that the \emph{$s$-graph} is a directed graph that contains an arc $(v_i, v_i')$ if and only if $s(v_i) = v_i'$.
By definition, flowery vertices have out-degree~0 and non-flowery vertices have out-degree 1 in the  $s$-graph.
Each triangle $T_i \in \cT$  
is of one of the following types (see Figure~\ref{fig:triangle-cover}).
\begin{enumerate}
\item \label{type:hhh} All three vertices of $T_i$ are flowery.
\item \label{type:hhl} $T_i$ contains two flowery vertices, say $v_{i,1}$
  and $v_{i,2}$, and $v_{i,3}$ points at one of them,
  say $s(v_{i,3})=v_{i,2}$.
\item \label{type:hla} $T_i$ contains one flowery vertex, say $v_{i,1}$,
  and $s(v_{i,2})=s(v_{i,3})=v_{i,1}$.
\item \label{type:hlp} $T_i$ contains one flowery vertex, say $v_{i,1}$,
  and the $s$-graph forms a path, e.g.,
  $s(v_{i,3})=v_{i,2}$ and $s(v_{i,2})=v_{i,1}$.
\item \label{type:hlc} $T_i$ contains one flowery vertex, say $v_{i,1}$,
  and the $s$-graph forms a 2-cycle of
  $s(v_{i,3})=v_{i,2}$ and $s(v_{i,2})=v_{i,3}$.
\item \label{type:l3c} $T_i$ contains only non-flowery vertices,
  and the $s$-graph forms a 3-cycle.
\item \label{type:l2c} $T_i$ contains only non-flowery vertices,
  and the $s$-graph forms a directed 2-cycle
  with one incoming edge,
  e.g., $s(v_{i,1})=v_{i,2}$, $s(v_{i,2})=v_{i,1}$
  and $s(v_{i,3})=v_{i,2}$.
\end{enumerate}
For every $i \in [p]$, we guess the type as well as the
ordering of $j=1,2,3$.
We will assume that all triangles in $\cT_i$ are compatible with our guesses (by deleting incompatible ones from $\cT_i$). 

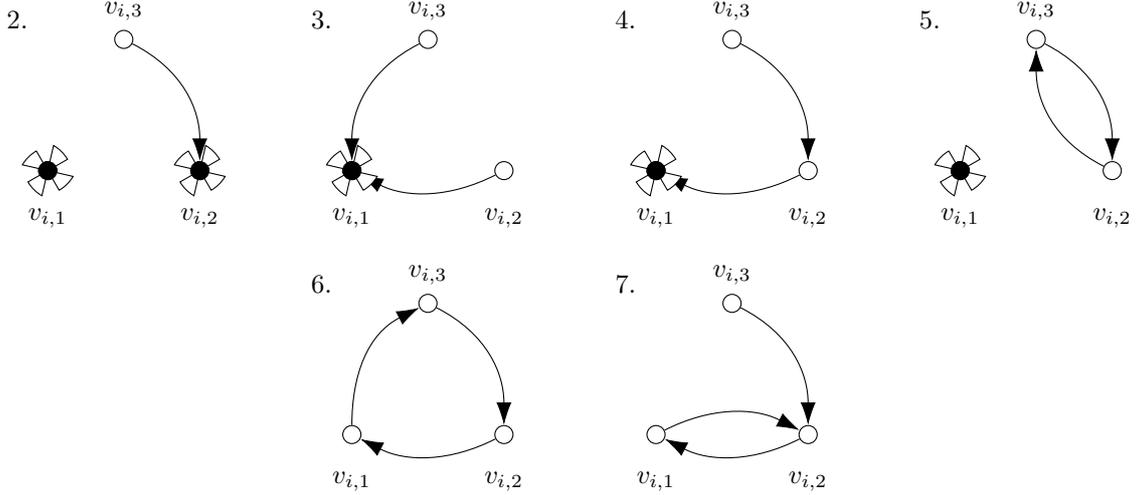
\begin{figure}
  \centering
  \begin{tikzpicture}
    \begin{scope}[yscale=0.87]
      \node at (-1.4, 2.3) {2.};
      \node[vertex,label={[below=5mm]:$v_{i,1}$}] (v1) at (-1, 0) {};
      \node[vertex,label={[below=5mm]:$v_{i,2}$}] (v2) at (1, 0) {};
      \node[vertex,label={above:$v_{i,3}$}] (v3) at (0, 2) {};

      \draw[->] (v3) to[out=330, in=90] (v2);
    \end{scope}
    \flowervertex{-1}{0}
    \flowervertex{1}{0}

    \begin{scope}[shift={(4,0)},yscale=0.87]
      \node at (-1.4, 2.3) {3.};
      \node[vertex,label={[below=5mm]:$v_{i,1}$}] (v1) at (-1, 0) {};
      \node[vertex,label={[below=5mm]:$v_{i,2}$}] (v2) at (1, 0) {};
      \node[vertex,label={above:$v_{i,3}$}] (v3) at (0, 2) {};
      \draw[->] (v3) to[out=210, in=90] (v1);
      \draw[->] (v2) to[out=210, in=330] (v1);
    \end{scope}
    \flowervertex{3}{0}

    \begin{scope}[shift={(8,0)},yscale=0.87]
      \node at (-1.4, 2.3) {4.};
      \node[vertex,label={[below=5mm]:$v_{i,1}$}] (v1) at (-1, 0) {};
      \node[vertex,label={[below=5mm]:$v_{i,2}$}] (v2) at (1, 0) {};
      \node[vertex,label={above:$v_{i,3}$}] (v3) at (0, 2) {};

      \draw[->] (v3) to[out=330, in=90] (v2);
      \draw[->] (v2) to[out=210, in=330] (v1);
    \end{scope}
    \flowervertex{7}{0}

    \begin{scope}[shift={(12,0)},yscale=0.87]
      \node at (-1.4, 2.3) {5.};
      \node[vertex,label={[below=5mm]:$v_{i,1}$}] (v1) at (-1, 0) {};
      \node[vertex,label={[below=5mm]:$v_{i,2}$}] (v2) at (1, 0) {};
      \node[vertex,label={above:$v_{i,3}$}] (v3) at (0, 2) {};

      \draw[->] (v3) to[out=330, in=90] (v2);
      \draw[->] (v2) to[out=150, in=270] (v3);
    \end{scope}
    \flowervertex{11}{0}

    \begin{scope}[shift={(4,-3.5)},yscale=0.87]
      \node at (-1.4, 2.3) {6.};
      \node[vertex,label={[below=5mm]:$v_{i,1}$}] (v1) at (-1, 0) {};
      \node[vertex,label={[below=5mm]:$v_{i,2}$}] (v2) at (1, 0) {};
      \node[vertex,label={above:$v_{i,3}$}] (v3) at (0, 2) {};

      \draw[->] (v3) to[out=330, in=90] (v2);
      \draw[->] (v2) to[out=210, in=330] (v1);
      \draw[->] (v1) to[out=90, in=210] (v3);
    \end{scope}

    \begin{scope}[shift={(8,-3.5)},yscale=0.87]
      \node at (-1.4, 2.3) {7.};
      \node[vertex,label={[below=5mm]:$v_{i,1}$}] (v1) at (-1, 0) {};
      \node[vertex,label={[below=5mm]:$v_{i,2}$}] (v2) at (1, 0) {};
      \node[vertex,label={above:$v_{i,3}$}] (v3) at (0, 2) {};

      \draw[->] (v3) to[out=330, in=90] (v2);
      \draw[->] (v2) to[out=210, in=330] (v1);
      \draw[->] (v1) to[out=30, in=150] (v2);
    \end{scope}
  \end{tikzpicture}
  \caption{Triangles of types 2 to 7. The arrows indicate the $s$-graph.}
  \label{fig:triangle-cover}
\end{figure}

Now we are ready to give the reduction to \textsc{Colorful Delta-matroid Matching}.
Recall that the input is a linear delta-matroid and an edge-colored graph over its ground set.
Whenever we identify a set $S$ of safe vertices,  we apply $|F \cap S|$-contraction by $S$.
This way, we are able to transform the task of finding a triangle packing into the simpler problem of finding a colorful matching.
The pairs will be colored in $C \subseteq [p]$.
For every $i \in [p]$, we will indicate whether $i \in C$, and if $i \in C$, then we also define a collection of pairs $P_i$ in~$V_{i}$, which are colored in $i$.

\paragraph*{Type 1.} If $T_i$ is of type \ref{type:hhh}, then we apply $f_{i,j}$-contraction by $H \cap V_{i,j}$ for each $j \in [3]$. (Recall that $H$ is the set of flowery vertices.)

\paragraph*{Types 2--5.}
Next, suppose that $T_i$ has one of the types \ref{type:hhl}--\ref{type:hlc}, where $v_{i, 1}$ is flowery.
First, we apply $f_{i,1}$-contraction by $H \cap V_{i,1}$.

We have two cases depending on $f_{i,2} = f_{i,3} = 1$ or not.
If $f_{i,2} = f_{i,3} = 1$, then we define a collection $P_i$ of pairs $(u_{i,2}, u_{i,3})$ from $V_{i,2} \times V_{i,3}$ colored in $i$ as follows.
\begin{itemize}
  \item If $T_i$ has type \ref{type:hhl}, then we add all pairs $(u_{i,2}, u_{i,3})$ such that $u_{i,2} = s(u_{i,3})$ and $N(u_{i,2}) \cap N(u_{i,3})$ contains a flowery vertex.
  \item If $T_i$ has type \ref{type:hla}, then we add all pairs $(u_{i,2}, u_{i,3})$ such that $u_{i,1} = s(u_{i,2}) = s(u_{i,3}) \in V_{i,1}$ is flowery and $\{ u_{i,1}, u_{i,2}, u_{i,3} \} \in \cT_i$.
  \item If $T_i$ has type \ref{type:hlp}, then we add all pairs $(u_{i,2}, u_{i,3})$ such that $u_{i,2} = s(u_{i,3})$ and $s(u_{i,2}) \in N(u_{i,3})$.
  \item If $T_i$ has type \ref{type:hlc}, then we add all pairs $(u_{i,2}, u_{i,3})$ such that $u_{i,2} = s(u_{i,3})$, $u_{i,3} = s(u_{i,2})$, and $N(u_{i,2}) \cap N(u_{i,3})$ contains a flowery vertex.
\end{itemize}
Otherwise (i.e., $(f_{i,2}, f_{i,3}) \ne (1, 1)$), $i \notin C$.
We apply the $f_{i,j}$-contraction by $V_{i,j}$ for $j = 2, 3$.

\paragraph*{Type 6.}
  Suppose that $T_i$ is of type \ref{type:l3c}.
  We may assume that $T_i$ is a strongly connected component of the $s$-graph.
  So, we delete the triangles from~$\cT_i$ that are not strongly connected components.
  Every triangle $T \in \cT_i$ contained in $V_i$ is then isolated and thus safe.
  We apply $f_{i,j}$-contraction by $V(\cT_i) \cap V_{i,j}$ for each $j \in [3]$.

\paragraph{Type 7.}
Finally, suppose that $T_i$ is of type \ref{type:l2c}.
We delete every triangle $T = \{ u_{i,1}, u_{i,2}, u_{i,3} \} \in \cT_i$ from $\cT_i$ unless the $s$-graph has bidirectional arcs $(u_{i,1}, u_{i,2}), (u_{i,2}, u_{i,1})$ and an arc $(u_{i,3}, u_{i,2})$.
We argue that $u_{i,j}\in V_{i,1} \cup V_{i,2}$ covered by some triangle in $\cT_i$ is safe.
We consider the graph induced by $E(\cT_i)$.
Every connected component consists of one vertex from $V_{i,1}$, another vertex from $V_{i,2}$, and some number of vertices from~$V_{i,3}$.
Thus, any triangle covering $u_{i,j} \in V_{i,1} \cup V_{i,2}$ certifies that it is safe.
We apply $f_{i,j}$-contraction by $V(\cT_i) \cap V_{i,j}$ for each $j \in [3]$.

\medskip

In summary, we apply $f_{i,j}$-contraction by $V_{i,j}$ for every $i \in [p]$ and $j \in [3]$, unless $T_i$ is of type \ref{type:hhl}--\ref{type:hlc} and $f_{i,j} = f_{i,j} = 1$.
We prove the correctness of our algorithm:


\dmtc*
\begin{proof}
  We show that our algorithm solves \textsc{Delta-matroid Triangle Cover} in $O^*(k^{\Oh(k)})$ time.
  Suppose that there is a triangle packing $\cT = \{ T_1, \dots, T_p \}$ covering a feasible set $F$ of size $k$.
  By the color-coding argument, we can find a partition of $V$ into $V_{i,j}$ for $i \in [p]$ and $j \in [3]$ with probability at least $1 / k^{O(k)}$.
  We guess whether each vertex $v_{i,j} \in V(\cT) \cap V_{i,j}$ is feasible or not and flowery or not with respect to the partition $V_{i,j}$ in $2^{\Oh(k)}$ time.
  For every non-flower vertex, we choose a vertex from $S_v$ uniformly at random.
  The probability that $s(v_{i,j}) \in T_{i,j}$ for every non-flowery vertex $v_{i,j} \in V(\cT) \cap V_{i,j}$ is at least $1/2^{\Oh(k)}$.
  As we have described, we then construct an instance $(D', (P_i)_{i \in C})$ of \textsc{Colorful Delta-Matroid Matching}, which can be solved in $O^*(2^{O(k)})$ time.
  Let $M = \{ v_{i,2} v_{i,3} \mid i \in C \}$, where $v_{i,2}$ and $v_{i,3}$ are the unique element in $T_i \cap V_{i,2}$ and $T_i \cap V_{i,3}$, respectively.
  By assumption, $V(M) \subseteq F$, and for every $i \in [p]$ and $j \in [3]$, $|F \cap V_{i,j}| = f_{i,j}$.
  Since $D'$ is the result of $f_{i,j}$-contraction by $V_{i,j}$ for $i \notin C$ and $j \in [3]$, 
  the set $V(M)$ is a feasible set in $D'$.
  

  Conversely, suppose that the \textsc{Colorful Delta-matroid Matching} instance has a solution $M$, i.e., $V(M)$ is feasible in $D'$.
  Let $F \supseteq V(M)$ be a feasible set in $D$ certifying the feasibility of $V(M)$ in $D'$.
  For every $i \in [p]$ and $j \in [3]$ with $f_{i,j} = 1$, let $v_{i,j}'$ be the unique element in $F \cap V_{i,j}$.
  We show that there is a triangle packing $\cT'$ with $F \subseteq V(\cT')$.
  It will consist of at most three triangles from $\cT_i$ for each $i \in [p]$.
  Let us examine each type:
  \paragraph*{Types 1 and 6.}
    If $T_i$ is of type \ref{type:hhh} or \ref{type:l3c}, then $v_{i,j}'$ is safe for every $j \in [3]$.
    Thus, by \Cref{lemma:completable-packing}, there is a triangle packing covering $F \cap V_i$.
  \paragraph*{Types 2--5.}
    Suppose that $T_i$ has one of the types \ref{type:hhl}--\ref{type:hlc}.
    If $f_{i,2} = f_{i,3} = 1$, then there is a pair $(v_{i,2}', v_{i,3}') \in M$.
    By construction, there is a triangle $T$ covering $v_{i,2}'$ and $v_{i,3}'$.
    Since $v_{i,1}'$ is flowery, by \Cref{lemma:completable-packing}, if $v_{i,1}' \notin T$, there is a triangle $T'$  disjoint from $T$ with $v_{i,1}' \in T'$. 
    If $(f_{i,2}, f_{i,3}) \in \{ (0, 1), (1, 0) \}$, then there is a triangle $T$ covering $v_{i,2}'$ or $v_{i,3}'$.
    Again, by \Cref{lemma:completable-packing}, $v_{i,1}'$ is covered by $T$ or another triangle disjoint from $T$.
  \paragraph*{Type 7.}
    Suppose that $T_i$ has type \ref{type:l2c}.
    By construction, there is a triangle covering $v_{i,3}'$.
    Since $v_{i,j}'$ is safe for $j \in [2]$, \Cref{lemma:completable-packing} yields a packing of at most three triangles covering $F \cap V_i$.
\end{proof}

\section{Cluster Subgraph above Matching} \label{sec:cs-above-matching}

Recall that \textsc{Cluster Subgraph} asks whether the input graph $G$ contains a cluster subgraph with at least $\ell$ edges.
In this section, we develop an FPT algorithm for \textsc{Cluster Subgraph} parameterized by $k = \ell - \MM(G)$, where $\MM(G)$ is the maximum matching size.
Our algorithm is a reduction to (a variant of) \textsc{Delta-matroid Triangle Cover}.

\csam*

We start with the simpler case that the input graph $G$ is a $K_4$-free graph with a perfect matching.
Note that $k = \ell - n/2$.
Let $D$ be the dual of the matching delta-matroid.
We can formulate the \textsc{Cluster Deletion} problem in terms of $D$ as follows.

\begin{lemma}
  \label{lemma:ce-feasible-cover}
  There is a cluster graph with at least $\ell$ edges if and only if there is a triangle packing $\mathcal{T}$ covering a feasible set $F \subseteq V(\mathcal{T})$ of size $2k$ in $D$.
\end{lemma}

To prove \Cref{lemma:ce-feasible-cover}, we start with an observation.
We fix a perfect matching $M$ in $G$.
An alternating path $P$ is a path that visits an edge in $M$ and not in $M$ in an alternating order.
Further, we say that $P$ is an $M$-alternating path if the first and last edges are in $M$.
For a vertex set $F$, an alternating $F$-path is a path where the endpoints are both in $F$ and the inner vertices are disjoint from $F$.
Two paths are disjoint when they do not have any vertex in common.

\begin{lemma}
  \label{lemma:alt-paths}
  For a vertex set $F \subseteq V$, $F$ is feasible in the dual matching delta-matroid $D$ if and only if there is a collection of $|F|/2$ vertex-disjoint $M$-alternating $F$-paths.
\end{lemma}
\begin{proof}
  If $F$ is feasible in $D$, then $G - F$ has a perfect matching $M'$.
  The symmetric difference $M \Delta M'$ yields a collection of $|F|/2$ disjoint $M$-alternating $F$-paths.
  Conversely, if $\cP$ is a collection of $|F|/2$ disjoint $M$-alternating $F$-paths, then $M \Delta E(\cP)$ is a perfect matching in $G - F$, so $F$ is feasible in $D$.
\end{proof}

Intuitively speaking, every $M$-alternating path with endpoints in triangles ``earns'' one edge for the solution.

\begin{proof}[Proof of \Cref{lemma:ce-feasible-cover}]
  First, assume that there is a feasible set $F$ of size $2k$ covered by a triangle packing~$\mathcal{T}$.
  By \Cref{lemma:alt-paths}, we have a collection $\mathcal{P}$ of $k$ vertex-disjoint $M$-alternating $F$-paths.
  Suppose that among $\mathcal{P}$, there are $k'$ paths of length one.
  We construct an edge set $E'$ of size at least $\ell$ such that $(V(G), E')$ is a cluster graph.
  We start with $E' = M$.
  Replace all edges incident to $V(\mathcal{T})$ (denoted by $\partial(V(\mathcal{T}))$) with $E(\mathcal{T})$, i.e., $E' = (E' \setminus \partial(\mathcal{V(\mathcal{T})})) \cup E(\mathcal{T})$.
  Note that $|E' \cap \partial(\mathcal{V(\mathcal{T})})| = 3|\mathcal{T}| - k'$ because we count one edge in $E' \cap \partial(\mathcal{V(\mathcal{T})})$ for every vertex in $V(\mathcal{T})$ and the edges corresponding to length-one paths in $\mathcal{P}$ are counted twice.
  Since $3|\mathcal{T}| - k'$ edges are replaced with $3|\mathcal{T}|$ edges, we have $|E'| = |M| + k'$.
  For every $M$-alternating $F$-path $P \in \mathcal{P}$ of length greater than one, we take $E' = E' \Delta (E(P) \setminus \delta(V(T)))$, increasing the size $E'$ by one this way.
  Since the paths in $\mathcal{P}$ are disjoint, we indeed increase the size of $E'$ by $k - k'$, yielding an edge set of size $|M| + k$.
  It is easy to verify that $E'$ indeed induces a cluster subgraph.

  Next, assume that there is a cluster subgraph $H = (V(G), E')$ with at least $\ell$ edges.
  We assume that $E'$ is chosen in such a way that it minimizes $|E' \Delta M|$ under the constraint that $(V(G), E')$ is a cluster graph and $|E'| \ge \ell$.
  By the hand-shaking lemma, we have $|E'| = \frac{1}{2}\sum_{v \in V(H)} \deg_H(v)$.
  Since $H$ has maximum degree two, we have $\deg_H(u) + \deg_H(v) \le 4$ for every edge $uv \in M$.
  We may assume that $\deg_H(u) + \deg_H(v) \in \{ 2, 3, 4 \}$ for every $uv \in M$:
  If $\deg_H(u) + \deg_H(v) \le 1$, then we obtain another cluster graph $(V(G), E'')$ as follows:
  If $\deg_H(u) = \deg_H(v) = 0$, then we add the edge $uv$ to $H$.
  If $\deg_H(u) + \deg_H(v) = 1$, then we replace the edge incident to $uv$ with $uv$.
  In both cases, we have $|E'' \Delta M| < |E \Delta M|$, contradicting our choice of $E'$.

  We want to find $M$-alternating paths whose endpoints belong to triangles in~$H$.
  First, observe that an edge $uv \in M$ with $\deg_H(u) = \deg_H(v) = 2$ is an $M$-alternating path of length one.
  Next, consider an edge $uv \in M$ such that $\deg_H(u) = 2$ and $\deg_H(v) {= 1}$. We can find an $M$-alternating path starting from $u$:
  Since $\deg_H(v) = 1$, there exists an edge $vw \in E'$.
  We have $\deg_H(w) = 1$, since otherwise $w$ is part of a triangle in $H$ containing both $v$ and~$w$, which contradicts $\deg_H(v) = 1$.
  Suppose that $w$ is matched to $x$ in~$M$.
  If $\deg_H(x) = 2$, then we find an $M$-alternating path.
  Otherwise, we find an $M$-alternating path by repeating this argument until we reach a vertex of degree two in~$H$.
  Thus, we have an $M$-alternating path for every edge $uv \in M$ with $\deg_H(u) + \deg_H(v) = 3$.
  (Note that every path is counted twice.)
  By the hand-shaking lemma,  
  \begin{align*}
  |E'| = \frac{1}{2} \sum_{v \in V(G)} \deg_H(v) = \frac{1}{2} \sum_{uv \in M} \deg_H(u) + \deg_H(v) = |M| + \frac{1}{2} \sum_{uv \in M} (\deg_H(u) + \deg_H(v) - 2).
  \end{align*}
  Since $|E'| \ge \ell = |M| + k$, we have $\frac{1}{2}|M_3| + |M_4| \ge k$, where $M_3$ and $M_4$ are the set of edges $uv \in M$ with $\deg_H(u) + \deg_H(v) = 3$ and $4$, respectively.
  With the double counting in mind, we find $k$ vertex-disjoint $M$-alternating paths.
  Let $F$ be the set of their endpoints.
  It is of size $2k$, and feasible in $D$ by \Cref{lemma:alt-paths}.
  Moreover, the triangles in $H$ cover $F$ as every vertex in $F$ has degree 2 in $H$.
\end{proof}

It follows from \Cref{theorem:dmtc-fpt} that \textsc{Cluster Subgraph} can be solved in $O^*(k^{\Oh(k)})$ time when the input graph is $K_4$-free and has a perfect matching.

\paragraph*{Lifting assumptions.}

We discuss how to solve \textsc{Cluster Subgraph} generally.
To that end, we will reduce to the color-coded variant of \textsc{Delta-matroid Triangle Cover}, which is to find a triangle packing containing at least one triangle from $\cT_i$, where $(V_1, \dots, V_p)$ is a partition of $V$ and $\cT_i$ is a collection of triangles in $V_i$.
Note that our algorithm (\Cref{theorem:dmtc-fpt}) solves the color-coded variant in $O^*(k^{\Oh(k)})$ time.


\begin{lemma} \label{lemma:lifting}
  There is a parameterized Turing reduction from \textsc{Cluster Subgraph above Matching} to the color-coded variant of \textsc{Delta-matroid Triangle Cover} that generates $k^{O(k)}$ instances.
\end{lemma}
\begin{proof}
  First, we find a greedy packing of $K_4$'s.
  Start with $K = \emptyset$.
  As long as there is a $K_4$ in the graph, we add it to $K$, deleting it from the graph. 
  If we find $k/2$ disjoint $K_4$'s, we conclude that there is a cluster subgraph with $\ell$ edges:
  There remains a matching of size at least $\MM(G) - 4 \cdot k/2 = \MM(G) - 2k$, and thus taking the disjoint union of the $K_4$'s and the matching yields a cluster subgraph with at least $6 \cdot k/2 + (\MM(G) - 2k) = \MM(G) + k$ edges.
  Thus, we may assume that there is a set $K$ of size at most $2k$ such that $G - K$ is $K_4$-free (and $K$ can be found in polynomial time).
  Let $H = (V(G), E')$ be a hypothetical solution and $\cC = \{ C_1, \cdots, C_{\kappa} \}$ be the connected components of~$H$.
  Suppose that $H[K]$ has $\kappa'$ components $C_1', \cdots, C_{\kappa'}'$. We assume that $C_i' \subseteq C_i$ for every $i \in [\kappa']$.
  For every $i \in [\kappa]$, we guess $\gamma_i = |C_i \setminus C_i'|$.
  Note that $\gamma_i \le 3$, because $G - K$ is $K_4$-free.
  We search for a cluster graph with at least $\ell_1 = \ell - (\sum_{i \in [\kappa']} \binom{|C_i'|}{2} + \gamma_{i} \cdot |C_i'|)$ edges in $G - K$ such that it has a connected component (pairwise distinct for each $i$) of size $\gamma_i$ contained in $\bigcap_{v \in C_i'} N(v)$.
  We construct an instance of the color-coded variant of \textsc{Delta-matroid Triangle Cover} as follows.
  First, we partition $V(G) \setminus K$ into $k'$ sets $V_1, \dots, V_{k'}$ using color coding, where $k'$ is $\kappa'$ plus the number of triangles in $H - \bigcup_{i \in [\kappa']} C_i$.
  By \Cref{lemma:ce-feasible-cover}, we may assume that $k' \in \Oh(k)$.
  For $i \in [\kappa']$, $V_i$ should contain $C_i \setminus C_i'$.
  So we will assume that $V_i \subseteq \bigcap_{v \in C_i'} N(v)$.
  For each $i \in [\kappa' + 1, k']$, let $\cT_i$ contain all triangles in $G[V_i]$.
  For $i \in [\kappa']$, we do as follows.
  If $\gamma_i \le 1$, then we apply $\gamma_i$-contraction by $V_i$.
  If $\gamma_i = 3$, then the collection $\cT_i$ contains all triangles $T \subseteq V_i$.
  If $\gamma_i = 2$, we may assume that there is no triangle in $V_i$:
  Let $C_i = C_i' \cup \{ x, y \}$.
  If there is a triangle $T = \{ u, v, w \}$ in $V_i$, we obtain another solution by deleting edges incident to $x, y$, and $T$ and adding edges so that $C_i' \cup T$ becomes a clique.
  We delete at most $2 |C_i'| + 3$ edges since we assumed that $H$ contains no triangle in $V_i$.
  On the other hand, we add $3 |C_i'| + 3 \ge 2|C_i'| + 3$ edges, and thus the assumption that $V_i$ has no triangle if $\gamma_i = 2$ is justified.
  To reduce to the triangle cover problem, for every $i$ with $\gamma_i = 2$, we introduce a new vertex $v$ and make it adjacent to $V_i$ so that the solution contains a triangle with $v$ in it.
  To account for two edges incident to $v$, we increase the solution size by two, i.e., we will look for a cluster subgraph with at least $\ell_2 = \ell_1 + 2 |\Gamma_2|$, where $\Gamma_2  = \{ i \in [\kappa'] \mid \gamma_i = 2 \}$.
  Note that $G - K$ remains $K_4$-free after $v$ is added.

  We fix a maximum matching $M$ of $G - K$.
  Let $U = V(G - K) \setminus V(M)$.
  For every unmatched vertex $v \in U$, we add a degree-one neighbor $u'$ only adjacent to $u$ and let $U' = \{ u' \mid u \in U \}$.
  We call the resulting graph $G'$, and let $M' = M \cup \{ uu' \mid u \in U \}$.
  If an unmatched vertex $u$ is isolated in $H$, then we can increase the solution size by one by adding the edge~$uu'$.
  Let $\alpha$ and $\beta$ be the number of vertices in $U$ that are isolated and not isolated in $H$, respectively. 
  Note that we cannot guess which unmatched vertices are isolated in $H$, so we only guess two numbers $\alpha$ and $\beta$.
  We then search for a cluster graph with at least $\ell_3 = \ell_2 + \alpha$ edges in $G'$ such that $\beta$ vertices of $U'$ are isolated.
  Let $D'$ be the dual matching delta-matroid of $G'$ and $D$ be the result of applying a $\beta$-contraction by $U'$ on $D$.
  Note that the ground set of $D$ is $V(G - K)$.
  Our algorithm searches for a feasible set in $D$ of size $2(\ell_3 - (\MM(G) + \alpha + \beta)) - \beta \in \Oh(k)$.
  By \Cref{lemma:alt-paths}, a feasible set in $D'$ corresponds to the endpoints of disjoint $M'$-alternating paths in $G'$.
  Thus, for a feasible set $F$ in $D$, there is a collection $\mathcal{P}$ of $\frac{1}{2}(|F| + \beta)$ disjoint $M'$-alternating paths in $G'$.
  Note that $\beta$ vertices of $U'$ are isolated in the symmetric difference between $E(\mathcal{P})$ and $M'$.
  Moreover, by \Cref{lemma:ce-feasible-cover}, if a feasible set $F$ in $D$ is of size $2(\ell_3 - (\MM(G) + \alpha + \beta)) - \beta$ can be covered by a triangle packing, then we have a cluster subgraph with at least $\ell_3$ edges.
\end{proof}

It follows that \textsc{Cluster Subgraph} can be solved in time $O^*(k^{O(k)})$ time, proving \Cref{theorem:cs-am}.

\section{Strong Triadic Closure above Matching}
\label{sec:tcam}

In this section, 
we study \textsc{Strong Triadic Closure} problem parameterized by $k = \ell - \mu$. Let us first recall the problem definition: 
given a graph $G$ and $\ell \in \N$, the goal is to find a strong set of size $\ell$, where an edge set $E' \subseteq E$ is strong if $uw \in E(G)$ whenever $uv \in E'$ and $vw \in E'$.

We start our investigation similarly to \textsc{Cluster Subgraph} in Subsection~\ref{ssec:tcamfpt} by showing that the problem is FPT on $K_4$-free graphs. However, as opposed to \textsc{Cluster Subgraph}, we show in Subsection~\ref{ssec:STC-w1-hard} that the assumption of $K_4$-freeness cannot be lifted and \textsc{Strong Triadic Closure} is W[1]-hard parameterized by $k$ on general graphs. We then show that we can still overcome this limitation if maximum degree in the graph is bounded in Subsection~\ref{ssec:STC_bounded_degree} or if we are looking for an approximate solution in FPT time in Subsection~\ref{ssec:STC_FPT_Approx}. We complement our constant FPT approximation algorithm with a strong inapproximability lower bound for polynomial time algorithms.


\subsection{FPT algorithm for $K_4$-free graphs} \label{ssec:tcamfpt}

As in Section~\ref{sec:cs-above-matching}, we will initially work under the simplifying assumption that the input graph $G$ is a $K_4$-free graph with a perfect matching.
We deal with the general case afterwards.
While \textsc{Cluster Subgraph} is a special case of \textsc{Delta-matroid Triangle Cover} under this assumption, there seems to be no simple reduction from \textsc{Strong Triadic Closure}.
So we will develop a problem-specific algorithm, significantly extending our algorithm for \textsc{Delta-matroid Triangle Cover}.

\tcamfpt*

A \emph{strong cycle} on $n$ vertices for $n \ge 3$ is a graph $C$ such that $V(C) = \{ v_0, \dots, v_{n - 1} \}$ and there is an edge between two vertices $v_i$ and $v_j$ if 
$i - j = \pm 1, 2 \bmod n$.
Note that the edge set $\{ v_{i} v_{(i + 1) \bmod n} \mid i = 0, \dots, n - 1 \}$ constitutes a strong set, which we denote by $E_s(C)$.
Note also that a triangle is also a strong cycle.
The following lemma connects the size of a feasible set (in the dual matching delta-matroid $D$) covered by a strong cycle packing to the solution size (cf. Lemma~\ref{lemma:ce-feasible-cover})

\begin{lemma} \label{lemma:tc-feasible-cover}
  There is a strong set of size $\ell$ if and only if there is a strong cycle packing $\mathcal{C}$ such that there is a feasible set $F \subseteq V(\cC)$ of size $2k$.
\end{lemma}
\begin{proof}
  First, assume that there is a feasible set $F$ of size $2k$ covered by a strong cycle packing $\mathcal{C}$.
  Again by Lemma~\ref{lemma:alt-paths}, we have a collection $\mathcal{P}$ of $k$ vertex-disjoint $M$-alternating $F$-paths.
  We essentially take the symmetric difference $E(\mathcal{P}) \Delta M$ and then remove from this set all edges incident to $V(\mathcal{C})$ and add a strong edge set $E_s(C)$ for every $C \in \cC$. Note that we end up with exactly the same construction as in Lemma~\ref{lemma:ce-feasible-cover}, with only difference being that $\mathcal{C}$ can contain also longer cycles and not only triangles.
  Then, we can argue in the same as in the proof of Lemma~\ref{lemma:ce-feasible-cover} that we obtain a strong edge set of size $|M| + k$.

  Next, assume that there is a strong edge set $E'$ with at least $\ell$ edges.
  Since there is no $K_4$, we may assume that the graph $H = (V(G), E')$ induced by a strong edge set has no vertex of degree greater than two, and thus, every connected component is either a path or a cycle, as observed by Golovach et al.~\cite{GolovachHKLP20}.
  In fact, we may assume that a path of length at least two can be replaced by triangles:
  If $uv, vw \in E'$, then the edge $uw$ exists by the triadic closure and hence $uw$ can be added to $E'$.
  Suppose that the path is on $p \ge 3$ vertices $(v_1, \dots, v_p)$ (note that it has $p - 1$ edges)
  If $p = 3q$ or $p = 3q + 1$, then there is a packing of $q$ triangles, namely, $\{ 3p' - 2, 3p' - 1, 3p' \}$ for each $p' \in [q]$, which contains $3q \ge p - 1$ edges.
  If $p = 3q + 2$, then there is a packing of $q$ triangles $\{ 3p' - 2, 3p' 1, 3p' \}$ for each $p' \in [q]$ and one edge $v_{p - 1} v_p$, which has $3q + 1 = p - 1$ edges.
  Thus, we can assume that every connected component of $H$ of size at least three is a cycle.
  The rest is analogous to the proof of Lemma~\ref{lemma:ce-feasible-cover}; simply replace ``triangles'' with ``cycles'' in the proof.
\end{proof}

We say that a strong cycle packing $\cC$ is \emph{triangle-maximal} if it maximizes the number of triangles.
We will search for a triangle-maximal strong cycle packing.
This assumption proves to be valuable since a strong cycle can often be replaced by a triangle packing in various cases.
Consequently, we can assume that the strong set of our interest is small as shown in the Lemma~\ref{lemma:tri-maximal}, which is crucial for color coding.
\begin{lemma} \label{lemma:tri-maximal}
  Suppose that $\cC$ is a triangle-maximal strong cycle packing covering a feasible set $F$ of size~$2k$.
  For $p \in \N$, the following holds:
  \begin{itemize}
    \item $\cC$ contains no cycle of length $3p$ for $p > 1$.
    \item If $\cC$ contains a cycle $C$ of length $3p + 1$, then $F \cap V(C) = V(C)$.
    \item If $\cC$ contains a cycle $C$ of length $3p + 2$, then $|F \cap V(C)| \ge \frac{1}{2} (3p +2)$.
  \end{itemize}
  Moreover, we have $|V(\cC)| \le 6k$.
\end{lemma}
\begin{proof}
  Suppose that the solution has a long strong cycle $C = (v_1, \dots, v_l)$ of length $l$.
  \begin{itemize}
    \item 
    If $l = 3p$, we can replace $C$ with $p$ disjoint triangles, $\{ v_{3p' - 2}, v_{3p' - 1}, v_{3p'} \}$ for $p' \in [p]$.
    \item
    Suppose that $l = 3p + 1$ and that $V(C) \setminus F \ne \emptyset$.
    Assume w.l.o.g.\ that $v_{l} \notin F$.
    Then we can replace $C$ with $p$ disjoint triangles $\{ v_{3p' - 2}, v_{3p' - 1}, v_{3p'} \}$ for $p' \in [p]$.
    \item
    Suppose that $l = 3p + 2$.
    If $|V(C) \cap F| < |C| / 2$, then there are two consecutive vertices in $C$ not contained in $F$.
    W.l.o.g., assume that they are $v_{l - 1}$ and $v_{l}$.
    Then, we can cover $V(C) \setminus \{ v_{l - 1}, v_l \}$ with $p$ triangles $\{ v_{3p' - 2}, v_{3p' - 1}, v_{3p'} \}$ for $p' \in [q]$, which can replace $C$.
  \end{itemize}
  For every $C \in \cC$ we have $|C \cap F| \ge |C| / 3$ (where the equality holds for triangles), and thus $|V(\cC)| \le 3 |F| = 6k$.
\end{proof}

Let a triangle-maximal strong cycle packing be $\cC = \{ C^1, \cdots, C^{\kappa} \}$, where $C^i = (v^i_{1}, \dots, v^i_{l_i})$, covering a feasible set $F$ of size~$2k$.
Reusing the first color-coding steps from Section~\ref{sec:triangle-cover}, we guess the following:
\begin{itemize}
\item A partition $V=V^{1} \cup \ldots \cup V^\kappa$ such that for every $i \in [\kappa]$, 
  $V(C^i) \subseteq V^{i}$.
\item An integer $l^i$ such that $|V(C^i)| = l^i$ and a further partition of every $V^{i}$ into $l^i$ sets $V^{i}_1 \cup \cdots \cup V^{i}_{l_i}$ 
  such that $C^i$ intersects all parts; for $i \in [\kappa]$ and $j \in [l_i]$,
  let $v^i_{j}$ be the unique vertex of $V(C^i) \cap V^{i}_{j}$.
\end{itemize}
By Lemma~\ref{lemma:tri-maximal}, this step adds a multiplicative factor $k^{O(k)}$ to the running time.
For $i \in [\kappa]$ with $l_i = 3$, we can handle it in the same way as \textsc{Delta-matroid Triangle Cover}.
See Section~\ref{sec:triangle-cover} for details. 
For every $i \in [\kappa]$ with $l_i \ne 3$, we may assume that $l_i  \bmod 3 \in \{ 1, 2 \}$ by Lemma~\ref{lemma:tri-maximal}.
As we will focus attention on one cycle, we will drop the superscript $\cdot^i$, unless it is unclear.
We will also use the cyclic notation for subscripts, e.g., $v_{0} = v_{l_i}$.

Each $V_j$ is referred to as a \emph{block}.
We will assume that every vertex $v \in V_{j}$ has neighbors in the two preceding and two following blocks $V_{j-2} \cup V_{j-1} \cup V_{j+1} \cup V_{j+2}$.
Call a vertex $v \in V_{j}$ is \emph{flowery} if it has a matching of size at least 13 in its neighborhood and \emph{non-flowery} otherwise.
For non-flowery vertices, we can guess a few of its neighbors:

\begin{lemma} \label{lemma:vc}
  Let $C = (v_1, \dots, v_{l})$ be a strong cycle.
  For a vertex cover $S$ of $N(v_j)$, the following hold: (i) $v_{j-1} \in S$ or $v_{j+1} \in S$, (ii) $v_{j-2} \in S$ or $v_{j-1} \in S$, and (iii) $v_{j+1} \in S$ or $v_{j+2} \in S$.
\end{lemma}
\begin{proof}
  From the definition of strong cycles, it follows that $v_{j-2}, v_{j-1}, v_{j+1}, v_{j+2} \in N(v_j)$ and that $v_{j -1} v_{j+1}, v_{j-2} v_{j-1}$, $v_{j+1} v_{j+2} \in E(C)$.
\end{proof}

\subsubsection{Strong cycle of length $3p + 1$} \label{sssec:1mod3}

Suppose that $\cC$ contains a cycle $C = (v_{1}, \dots, v_{3p+1})$.
By Lemma~\ref{lemma:tri-maximal}, we may assume that $V(C) \subseteq F$.
We may also assume that every vertex is non-flowery:

\begin{lemma} \label{lemma:1mod3:non-flowery}
  Every vertex in $C$ is non-flowery.
\end{lemma}
\begin{proof}
  For contradiction, assume w.l.o.g.\ that $v_{3p + 1}$ is flowery.
  Consider a triangle packing $\{ \{ v_{3q-2}, \allowbreak v_{3q-1}, v_{3q} \} \mid q \in [p] \}$.
  Since $v_{3p+1}$ is flowery, there exists a triangle that covers $v_{3p+1}$ and is disjoint from the packing.
  These triangles cover $V(C)$, contradicting the triangle maximality.
\end{proof}

As every vertex is non-flowery, we can guess a neighbor for each vertex either in their preceding or following block.
In a simple case where this guessing results in a collection of $l$-cycles, the situation aligns with a case discussed in Section~\ref{sec:triangle-cover}, namely case 6.
Here, our algorithms applies 1-contraction applied to each block.
We argue that even if the outcome does not neatly fit into the pattern of $l$-cycles, a structural simplification that enables the application of 1-contraction on each block can be achieved.
This is done through a nuanced application of guessing arguments, where we strategically work out the connections between vertices across blocks:

\begin{lemma} \label{lemma:1mod3:ultimate}
  There is a polynomial-time algorithm that finds a set $S$ of vertices such that (i) $C \subseteq S$ with probability at least $1/2^{O(l)}$ and (ii) for any $X \subseteq S$ with $|X \cap V_j| = 1$ for each block $V_j$, there is a strong set that covers $X$.
\end{lemma}
\begin{proof}
  For every $v \in V^i$, $i \in [l]$, let $S_v$ be a vertex cover on $N(v)$.
  By Lemma~\ref{lemma:1mod3:non-flowery}, we may assume that $S_v$ has size at most 26.
  For every vertex $v_j \in V(C)$, by Lemma~\ref{lemma:vc}, $v_{j-1} \in S_{v_j}$ or $v_{j+1} \in S_{v_j}$ holds.
  Fix $d_j \in \{ 1, -1 \}$ such that $v_{j+d_j} \in S_{v_j}$ for every $v_j \in V(C)$.
  For every $j \in [l]$, we randomly choose $\delta_j = \pm 1$ uniformly at random.
  Then, with probability $1/2^{O(l)}$, $\delta_j = d_j$ holds for all $j \in [l]$.
  For every $j \in [l]$ and $v \in V_j$, we randomly choose a vertex $s(v) \in S_v \cap V_{j+\delta_j}$.
  The probability that $s(v_j) = v_{j+d_j}$ for each $j \in [l]$ is at least $1/2^{O(l)}$.
  
  Consider a directed graph $H_C$ over the vertex set $C$, where there is an arc from $v_{j}$ to $v_{j+d_j}$.
  First, we consider a simpler case that $H_C$ forms a directed cycle.
  We introduce an \emph{$s$-graph}, over the vertex set $V^i$, where there is an arc $(v, s(v))$ for every vertex $v \in V^i$.
  From the $s$-graph, we construct a set $S$ as follows.
  For every connected component in the corresponding undirected graph that forms a cycle, which intersects with each subset $V_j$ precisely once, we include all its vertices in $S$, provided the component's edges constitute a strong cycle.
  Note that this is completely analogous to case 6 in Section~\ref{sec:triangle-cover}.
  Then, the item (i) is satisfied because if $s(v_j) = v_{j+d_j}$ for each $j \in [l]$, then all vertices of $C$ are incorporated into $S$.
  For the item (ii), note that the entire set $S$ can be covered by strong cycles.

  Suppose that $H_C$ is not a directed cycle.
  While it seems difficult to obtain a collection of cycles, we obtain a similar structure, a collection of \emph{cycles with parallel vertices}.
  Here, a cycle with parallel vertices is a graph obtained from a standard cycle $(x_1, \dots, x_n)$ by adding vertices parallel to $x_i$, i.e., adjacent to $x_{i-1}$ and $x_{i+1}$.
  As $H_C$ is not a directed cycle, there exists $a \in [l]$ such that the edge $v_{a-1} v_{a}$ is missing in the underlying undirected graph.
  Assume w.l.o.g.\ that $a = 1$.
  It follows that $d_a = d_1 = +1$ and $d_{a-1} = d_{l} = -1$.
  Let $a'$ be the largest integer such that $d_j = +1$ for all $j \in [a, a']$.
  Such an integer $a'$ exists because $d_l = -1$.
  By the definition of $a'$, $d_{b'} = -1$, where $b' = a' + 1$.
  Also, let $b$ be the largest integer such that $d_j = -1$ for all $j \in [b', b]$.
  We will simplify the structure of the $s$-graph over $V_a \cup \cdots \cup V_b$ as follows.
  First, we find all bidirectional arcs $(u_{a'}, u_{b'})$ within $V_{a'} \cup V_{b'}$ in the $s$-graph, i.e., $s(u_{a'})= u_{b'}$ and $s(u_{b'}) = u_{a'}$.
  From the $s$-graph, delete all vertices in $V_{a'} \cup V_{b'}$ not incident with any bidirectional arc.
  Now consider the subgraph of the $s$-graph induced by $V_a \cup \cdots \cup V_{a'}$.
  As each arc is leading from $V_j$ to $V_{j+1}$ for $j \in [a, a']$, each weakly connected component is a directed tree rooted at some vertex in $V_{a'}$.
  We delete a vertex $v \in V_{a} \cup \cdots \cup V_{a'}$ from the $s$-graph, if it is not part of any path starting from $V_{a}$ and ending at $V_{a'}$.
  
  We further simplify the $s$-graph as follows.
  Henceforth, we will treat the $s$-graph as an undirected graph, disregarding the orientations.
  For every vertex $v$ in the $s$-graph that is part of $V_j$, $j \in [a+2, a']$, we randomly choose a vertex $s'(v) \in S_v \cap (V_{j-2} \cup V_{j-1})$.
  Based on the chosen $s'(v)$, we apply the following rules:
  \begin{itemize}
    \item If $s'(v) \in V_{j-1}$, then we preserve only the subtree rooted at $s'(v)$, eliminating all other subtrees connected to $v$. \item If $s'(v) \in V_{j-2}$, then let $v'\in V_{j-1}$ be the intermediate vertex linking $s'(v)$ to $v$.
    In this case, we keep the subtree rooted at $v'$, deleting all other subtrees connected to $v$.
  \end{itemize}
  We apply the rules similarly on $V_{b'} \cup \cdots \cup V_{b}$ as well.
  As a result, every connected component, when restricted to $V_{a+1} \cup \cdots \cup V_{b-1}$, transforms into a path.
  Note however that vertices in $V_{a+1}$ and $V_{b-1}$ may have multiple neighbors in $V_a$ and $V_{b}$ linking to them, thereby forming a double broom.
  Recall that a double broom is a tree obtained from a path by adding any number of degree-1 neighbors to its leaves.

  Next, we describe how to deal with missing edges in the underlying undirected graph of $H_C$.
  Suppose that neither $v_{a-1} v_a$ nor $v_{a} v_{a-1}$ is present in $H_C$.
  By the above discussion, there are weakly connected components which form double brooms.
  Suppose that they intersect blocks $V_{c} \cup \cdots \cup V_{a-1}$ and  $V_{a} \cup \cdots \cup V_{b}$.
  Denote those intersecting $V_{c} \cup \cdots \cup V_{a-1}$ and  $V_{a} \cup \cdots \cup V_{b}$ by $\cB^{-}$ and $\cB^+$, respectively.
  We show how to connect these double brooms.
  For each $B^+ \in \cB^+$, we randomly choose a double broom in $\cB^-$ denoted by $r(B^+)$ as follows.
  Let $v$ be the vertex in $V_{a+1}$ on $B^+$.
  We randomly choose a vertex $s'(v) \in S_v \cap (V_{a-1} \cap V_{a})$, and 
  apply one of the following rules:
  \begin{enumerate}
    \item
    If $s'(v) \in V_{a-1}$, then let $r(B^+)$ be the double broom in $\cB^-$ containing $s'(v)$.
    \item
    If $s'(v) \in V_{a}$, then we further choose $s''(s'(v)) \in S_{s'(v)} \cap (V_{a-2} \cup V_{a-1})$ uniformly at random.
    Let $r(B^+)$ be the double broom in $\cB^-$ containing $s''(s'(v))$.
  \end{enumerate}
  For each $B^- \in \cB^-$, we randomly choose $r(B^-) \in \cB^+$ similarly.
  We keep only the pairs $(B^+, B^-) \in \cB^+ \times \cB^-$ such that $r(B^+) = B^-$ and $r(B^-)= B^+$.
  Let $v^+$ be the vertex on $B^+$ in $V_{a+1}$ and $v^-$ be the vertex on $B^-$ in $V_{a-2}$.
  We connect double brooms as follows (see Figure~\ref{fig:1mod3:connection} for an illustration):
  \begin{enumerate}
    \item[(a)] Suppose that rule 1 is applied when choosing $r(B^+)$ and $r(B^-)$.
    Then, we connect $(v^-, s'(v^+), s'(v^-),\allowbreak v^+)$ as a path.
    \item[(b)] Suppose that rules 1 and 2 are applied when choosing $r(B^+)$ and $r(B^-)$, respectively (or vice versa).
    If $s''(s'(v^-)) \in V_a$, then we connect $(v^-, s'(v^-), s''(s'(v^-)), v^+)$ as a path.
    Otherwise, we have $s''(s'(v^-)) \in V_{a+1}$, and we connect as a path with parallel vertices: $(v^-, s'(v^-), (s')^{-1}(v^+) \cap V_{a}, v^+)$, where $(s')^{-1}(v^+)$ denotes the vertices $v$ with $s'(v) = v^+$.
    \item[(c)] Suppose that rule 2 is applied when choosing $r(B^+)$ and $r(B^-)$.
    Then, we connect $(v^-, s'(v^-), s'(v^+), \allowbreak v^+)$ as a path.
  \end{enumerate}
  We apply this argument to every $a \in [l]$ such that $v_{a-1} v_a$ and $v_a v_{a-1}$ are absent from $H_C$.
  This results in a collection of cycles with parallel vertices.
  We apply sanity check as follows.
  First, we keep only those that circle around the blocks exactly once.
  Second, we check whether every vertex is adjacent to the two preceding and two succeeding vertices.
  If the vertex fails adjacency test, we have two cases. 
  If the vertex is a parallel vertex, then delete it from the cycle.
  Otherwise, delete the entire cycle.
  Finally, let $S$ be the vertices covered the remaining cycles with parallel vertices.

  We claim both conditions of the lemma hold.
  For the item (i), note that the probability that for every vertex $v \in C$, $s(v)$, $s'(v)$, and $s''(v)$ belong to $C$ is at least $1/2^{O(l)}$.
  It is straightforward to verify that under this condition the cycle $C$, possibly with parallel vertices, remains until the step where it is included into $S$.
  The item (ii) holds because $S$ is precisely a set of vertices covered by a collection of strong cycles with parallel vertices.
\end{proof}

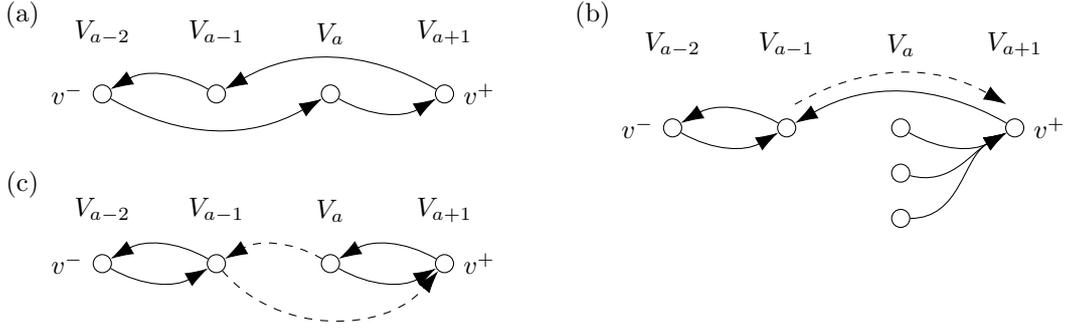
\begin{figure}
  \centering
  \begin{tikzpicture}[scale=1.5]
    \node at (-.7, .5) {(a)};
    \begin{scope}[shift={(0, -0.2)}]
    \node[vertex,label={180:$v^-$},label={[above=4mm]:$V_{a-2}$}] (v1) at (0, 0) {};
    \node[vertex,label={[above=4mm]:$V_{a-1}$}] (v2) at (1, 0) {};
    \node[vertex,label={[above=4mm]:$V_{a}$}] (v3) at (2, 0) {};
    \node[vertex,label={0:$v^+$},label={[above=4mm]:$V_{a+1}$}] (v4) at (3, 0) {};
    \draw[->] (v2) to[out=150, in=30] (v1);
    \draw[->] (v3) to[out=330, in=210] (v4);
    \draw[->] (v4) to[out=150, in=30] (v2);
    \draw[->] (v1) to[out=330, in=210] (v3);
    \end{scope}

    \begin{scope}[shift={(5, -.5)}]
    \node at (-.7, 1) {(b)};
    \node[vertex,label={180:$v^-$},label={[above=7mm]:$V_{a-2}$}] (v1) at (0, 0) {};
    \node[vertex,label={[above=7mm]:$V_{a-1}$}] (v2) at (1, 0) {};
    \node[vertex,label={[above=7mm]:$V_a$}] (v3) at (2, 0) {};
    \node[vertex,label={0:$v^+$},label={[above=7mm]:$V_{a+1}$}] (v4) at (3, 0) {};
    \node[vertex] (v31) at (2, -.4) {};
    \node[vertex] (v32) at (2, -.8) {};
    \draw[->] (v2) to[out=150, in=30] (v1);
    \draw[->] (v3) to[out=330, in=210] (v4);
    \draw[->] (v4) to[out=150, in=30] (v2);
    \draw[->] (v1) to[out=330, in=210] (v2);
    \draw[->, dashed, transform canvas={yshift=2.5mm}] (v2) to[out=30, in=150] (v4);
    \draw[->] (v31) to[out=345, in=210] (v4);
    \draw[->] (v32) to[out=0, in=210] (v4);
    \end{scope}

    \begin{scope}[shift={(0, -1.7)}]
    \node at (-.7, .7) {(c)};
    \node[vertex,label={180:$v^-$},label={[above=3mm]:$V_{a-2}$}] (v1) at (0, 0) {};
    \node[vertex,label={[above=3mm]:$V_{a-1}$}] (v2) at (1, 0) {};
    \node[vertex,label={[above=3mm]:$V_a$}] (v3) at (2, 0) {};
    \node[vertex,label={0:$v^+$},label={[above=3mm]:$V_{a+1}$}] (v4) at (3, 0) {};
    \draw[->] (v1) to[out=330, in=210] (v2);
    \draw[->] (v2) to[out=150, in=30] (v1);
    \draw[->] (v3) to[out=330, in=210] (v4);
    \draw[->] (v4) to[out=150, in=30] (v3);
    \draw[->, dashed] (v3) to[out=150, in=30] (v2);
    \draw[->, dashed] (v2) to[out=310, in=230] (v4);
    \end{scope}
  \end{tikzpicture}
  \caption{The connection between two double brooms. An arrow indicates $s'(\cdot)$, and a dashed arrow $s''(\cdot)$.}
  \label{fig:1mod3:connection}
\end{figure}

\subsubsection{Strong cycle of length $3p + 2$} \label{sssec:2mod3}

Suppose that $\cC$ contains a cycle $C = (v_{1}, \dots, v_{3p+2})$.
This situation is different from when the cycle's length is $3p + 1$, as we cannot  assume that all vertices in $C$ are included in the set $F$. Consequently, we need to guess for each vertex $v_i$ whether it is \emph{intersecting}, i.e., whether it is part of $F$. 
We call a vertex $v_i$ \emph{safe} if it is flowery or not intersecting.
We will start off with another observation with respect to the triangle-maximality.

\begin{lemma} \label{lemma:2mod3:unique}
  The common neighborhood $N(v_{j}) \cap N(v_{j+1}) \cap V_{j+2}$ is a singleton $\{ v_{j+2} \}$ for every $j \in [l_i]$.
\end{lemma}
\begin{proof}
  Assume for contradiction that, w.l.o.g., $v_{3p+1}$ and $v_{3p+2}$ has a common neighbor $v_{1}' \in V_{1}$ in addition to $v_{1}$.
  Then there is a triangle packing $\{ v_{3q-2}, v_{3q-1}, v_{3q} \mid q \in [p]\} \cup \{ v_{3p+1}, v_{3p+2}, v_{1}' \}$ that covers the strong cycle.
\end{proof}

In view of Lemma~\ref{lemma:2mod3:unique}, we obtain a polynomial-time algorithm that generates a collection of strong cycles in $V^i$, denoted by $\mathcal{S}$, such that $C \in \mathcal{S}$.
To simplify our terminology, we will extend the use of the terms ``flowery'', ``intersecting'', and ``safe'': a block $V_{j}$ will be described as flowery (intersecting, or safe) if its corresponding vertex $v_{j}$ is flowery (intersecting, or safe).

\begin{lemma} \label{lemma:2mod3:find-cycles}
  There is a polynomial-time algorithm that finds a collection of strong cycles
  $\mathcal{S}$ such that the following hold: (i) $C^i \in \mathcal{S}$ with probability at least $1/2^{O(l)}$ and (ii) strong cycles in $\mathcal{S}$ are vertex-disjoint on non-flowery blocks, i.e., for every vertex $v$ in a non-flowery block, there is at most one strong cycle in $\mathcal{S}$ covering $v$.
\end{lemma}
\begin{proof}
  In view of Lemma~\ref{lemma:2mod3:unique}, we obtain a polynomial-time algorithm that generates a collection of edge-disjoint strong cycles in $V_i$, denoted by $\cS'$, such that $C \in \mathcal{S}'$.
  For each edge $u_1 u_2$ with $u_1 \in V_{1}$ and $u_2 \in V_{2}$, this algorithm will find at most one strong cycle in $\cS$ as follows.
  If $u_1$ and $u_2$ share more than one common neighbor in $V_{3}$, then Lemma~\ref{lemma:2mod3:unique} allows us to conclude that $u_1 u_2$ is not part of a strong cycle in $\cC$.
  On the other hand, if $u_1$ and $u_2$ share a unique common neighbor $u_3$ in $V_{3}$, then we check whether there is a unique neighbor between $u_2$ and $u_3$ in $V_{4}$.
  This process continues iteratively until it circles back to $V_{1}$.
  Through this process, a strong cycle $(u_1, \dots, u_{3p+2})$ that is part of $\cC$ is identified.
  By construction, the strong cycles found this way are edge-disjoint.

  To ensure vertex-disjointness on non-flowery blocks, we randomly pick $\cS \subseteq \cS'$ as follows.
  For every non-flowery vertex $v \in V_j$, we randomly choose a neighbor $s(v)$ from $S_v \cap (V_{j-1} \cup V_{j+1})$, where $S_v$ is a vertex over on $N(v)$.
  Strong cycles $C'$ are deleted from $\cS'$ if there is a vertex $v \in V(C')$ such that $s(v) \notin V(C')$.
  Let $\cS$ denote the remaining strong cycles.
  We claim that strong cycles in $\mathcal{S}$ are vertex-disjoint on non-flowery blocks.
  For a vertex $v$ in a non-flowery block $V_j$, and it is possible that multiple strong cycles in $\cS'$ cover $v$.
  However, by the edge-disjointness, these cycles pass through different vertices in $V_{j-1}$ (and $V_{j+1}$).
  Since we keep at most one cycle in $\cS$, it follows that at most one strong cycle covering $v$ remains, thereby establishing vertex-disjointness.
  Note that $C^i \in \cS$ holds when $s(v_j) \in V(C)$ holds for all non-flowery vertices $v_j \in V(C)$.
  This occurrence has probability at least $1/2^{O(l)}$ by Lemma~\ref{lemma:vc}.
\end{proof}

\begin{figure}
  \centering
  \begin{tikzpicture}
    \begin{scope}[yscale=0.8]
    \node[vertex,label={[label distance=.3cm]below:$v_{1}$}] (f1) at (0, 0) {};
    \node[bvertex,label={above:$v_{2,1}$}] (v1a) at (1, 1) {};
    \node[vertex,label={below:$v_{2,2}$}] (v1b) at (1, -1) {};
    \node[bvertex,label={[label distance=.3cm]below:$v_{3}$}] (f2) at (2, 0) {};
    \node[vertex,label={above:$v_{4,1}$}] (v2a) at (3, 1) {};
    \node[bvertex,label={below:$v_{4,2}$}] (v2b) at (3, -1) {};
    \node[bvertex,label={above:$v_{5,1}$}] (v3a) at (4, 1) {};
    \node[vertex,label={below:$v_{5,2}$}] (v3b) at (4, -1) {};
    \node[bvertex,label={[label distance=.3cm]below:$v_{6}$}] (f3) at (5, 0) {};
    \node[bvertex,label={above:$v_{7,1}$}] (v4a) at (6, 1) {};
    \node[vertex,label={below:$v_{7,2}$}] (v4b) at (6, -1) {};
    \node[bvertex,label={above:$v_{8,1}$}] (v5a) at (7, 1) {};
    \node[vertex,label={below:$v_{8,2}$}] (v5b) at (7, -1) {};

    \draw (v1a) -- (f1) -- (v1b);
    \draw (v1a) -- (f2) -- (v1b);
    \draw (v2a) -- (f2) -- (v2b);
    \draw (v2a) -- (v3a) -- (f3) -- (v3b) -- (v2b);
    \draw (v5a) -- (v4a) -- (f3) -- (v4b) -- (v5b);
    \draw (v5a) -- (7.5, 0.5) {};
    \draw (v5b) -- (7.5, -0.5) {};
    \draw[dotted] (v5a) -- (7.8, 0.2) {};
    \draw[dotted] (v5b) -- (7.8, -0.2) {};
    \draw (-0.5, 0.5) -- (f1) -- (-0.5, -0.5);
    \draw[dotted] (-0.8, 0.8) -- (f1) -- (-0.8, -0.8);
    \end{scope}

    \flowervertex{0}{0}
    \flowervertex{2}{0}
    \flowervertex{5}{0}
  \end{tikzpicture}
  \caption{}
  \label{fig:2mod3:flowery}
\end{figure}

Lemma~\ref{lemma:2mod3:find-cycles} is sufficiently powerful to obtain a statement similar to Lemma~\ref{lemma:1mod3:ultimate}, when there is no flowery vertex in $C$:
\begin{lemma} \label{lemma:2mod3:no-flowery}
  If $C$ contains no flowery vertex,
  then there is a polynomial-time algorithm that finds a set $S$ of vertices such that (i) $V(C) \subseteq S$ with probability at least $1/2^{O(l)}$ and (ii) for any $X \subseteq S$ such that $|X \cap V_j| = 1$ for each intersecting block $V_j$ and $|X \cap V_j| = 0$ for each non-intersecting block, there is a strong set that covers $X$.
\end{lemma}
\begin{proof}
  Define $S$ as the set of vertices covered by strong cycles in $\cS$ of Lemma~\ref{lemma:2mod3:find-cycles}.
  Since $C$ belongs to $\cS$ with probability at least $1/2^{O(l)}$, the condition (i) holds.
  Furthermore, it is possible to cover a set $X \subseteq S$, where $|X \cap V_j| \le 1$ for each flowery block $V_j$, with a strong set.
\end{proof}

Let us assume that $C$ includes at least one flowery vertex.
Note that the statement of Lemma~\ref{lemma:2mod3:no-flowery} may not apply.
Specifically, the strategy of simply selecting $S$ to be the set of vertices covered by strong cycles in $\cS$ is not viable because the condition (ii) may not be met in the presence of flowery vertices, as illustrated in Figure~\ref{fig:2mod3:flowery}.
Here, there are two strong cycles $(v_1, v_{2,i}, v_3, v_{4,i}, v_{5,i}, v_6, v_{7,i}, v_{8,i})$ for $i = 1, 2$, with vertices in $X$ marked in black.
We claim that there is no strong set covering $X$.
As $v_{4,2} \in X$, either $\{ v_3, v_{4,2}, v_{5,2} \}$ or $\{ v_{4,2}, v_{5,2}, v_6 \}$ needs to be part of the strong set.
Similarly, as $v_{5,1}$, either $\{ v_3, v_{5,1}, v_{5,2} \}$ or $\{ v_{5,1}, v_{5,2}, v_6 \}$ needs to be part of the strong set.
It follows that both $v_3$ and $v_6$ have been covered.
Thus, $\{ v_1, v_{7, 1}, v_{8,1} \}$ must be included into the strong set as well, leaving $v_{2,1}$ uncovered.
This example underscores the limitation of Lemma~\ref{lemma:2mod3:find-cycles} to provide a statement similar to Lemma~\ref{lemma:1mod3:ultimate}.
To work around this challenge, we leverage the idea of pairing some vertices via reduction to \textsc{Colorful Delta-matroid Matching}.
In order to apply this idea, let us give a few more observations with respect to triangle maximality.

\begin{lemma} \label{lemma:2mod3:consecutive}
  No two consecutive vertices $v_j$ and $v_{j+1}$ are safe.
\end{lemma}
\begin{proof}
  Assume for contradiction that two consecutive vertices $v_{3p+1}$ and $v_{3p+2}$ are both safe.
  Since $C$ is a strong cycle, there are triangles $\{ v_{3q-2}, v_{3q-1}, v_{3q}\}$ for each $q \in [p]$.
  If $v_{3p+1}$ (or $v_{3p+2}$) is flowery, then it implies the existence of an additional triangle containing $v_{3p+1}$ (or $v_{3p+2}$) that are disjoint from all other triangles (these are two separate triangles if both $v_{3p+1}$ and $v_{3p+2}$ are flowery).
  We can choose such triangles by the definition of flowery vertices.
  The packing of these triangles yields a solution with a greater number of triangles, a contradiction.
\end{proof}

By Lemma~\ref{lemma:2mod3:consecutive}, we may assume w.l.o.g.\ that $v_{1}$ is not safe and that $v_{3p+2}$ is safe.
Let $f_1, \dots, f_r$ be the safe vertices in $C_i$.
These vertices partition $C_i$ into paths $P_1, \dots, P_r$ of non-safe vertices, where the endpoints of $P_i$ are adjacent to $f_{i-1}$ and $f_i$.

\begin{lemma} \label{lemma:2mod3:not-divisible}
  For each $r' \in [r]$, $|V(P_r)|$ is not divisible by 3.
\end{lemma}
\begin{proof}
  If a path, say $P_r$, has length divisible by three, then we can cover $F \cap V(C_i)$ with a triangle packing: 
  \begin{itemize}
    \item 
    $V(P_r)$ can be covered by $|V(P_r)|/3$ disjoint triangles.
    \item
    If $f_{r-1}$ (and $f_r$) is flowery, then it can be covered by disjoint triangles.
    \item
    The remaining vertices form a path whose number of vertices is divisible by three, which can be covered by a triangle packing.
  \end{itemize}
  This contradicts the triangle maximality.
\end{proof}

Lemma~\ref{lemma:2mod3:not-divisible} can be further strengthened as follows.

\begin{lemma} \label{lemma:2mod3:path-lengths}
  There is exactly one path $P_{r'}$ such that $|V(P_{r'})| \bmod 3 = 1$ and for all other paths $P_{r''}$, $|V(P_{r''})| \bmod 3 = 2$.
\end{lemma}
\begin{proof}
  Suppose that there exist two paths $P_{r'}$ and $P_{r''}$ for $r' < r'' \in [r]$ such that both paths have order $1 \bmod 3$.
  By Lemma~\ref{lemma:2mod3:not-divisible}, we may assume w.l.o.g.\ that for every $r''' \in [r', r'' - 1]$, the path $P_{r'''}$ has order $2 \bmod 3$.
  Observe that $f_{r'-1}$ and $f_{r''}$ give a partition into two paths of order divisible by 3:
  \begin{itemize}
    \item The path containing $P_{r'}$ and $P_{r''}$ includes (i) two paths with lengths of $1 \bmod 3$, (ii) $r'' - r' - 1$ paths with lengths of $2 \bmod 3$, and (iii) $r'' - r'$ safe vertices, resulting in a total order divisible by 3, specifically $2 + 2(r'' - r' - 1) + (r'' - r') \equiv 0 \bmod 3$.
    \item The other path also has order $0 \bmod 3$ as well because $C$ has $3p + 2$ vertices.
  \end{itemize}
  This implies that both paths can be covered by disjoint triangles, which contradicts the triangle maximality.
\end{proof}

Given Lemma~\ref{lemma:2mod3:path-lengths}, we may assume w.l.o.g.\ that $|V(P_1)| \bmod 3 = 1$, and that $|V(P_{r'})| \bmod 3 = 2$ for $2 \le r' \le r$.
Under this assumption, we have the following:

\begin{lemma} \label{lemma:2mod3:flowery-index}
  If $v_{j}$ is flowery, then $j \bmod 3 = 2$.
\end{lemma}
\begin{proof}
  Assuming that $v_j = f_r$, we have
  \begin{align*}
    j \equiv (|V(P_1)| + 1) + \sum_{r' = 2}^r (|V(P_{r'})| + 1) \equiv 2
  \end{align*}
  because $|V(P_{r'})| + 1$ is divisible by 3 for each $r' \in [2, r]$.
\end{proof}


\newcommand{\at}{\mathsf{A}}
\newcommand{\tria}{\mathsf{T}}
\newcommand{\scycle}{\mathsf{S}}

Suppose that $P_1$ is on $3 q_1 + 1$ vertices and that $P_{r'}$ is on $3q_{r'} + 2$ vertices for each $r' \in [2, r]$.
To facilitate the manipulation of indices, we introduce a notation $\at(r', t)$, defined as the integer $j$ such that $v_j$ is the $t$-th vertex of $P_{r'}$.
Additionally, for a strong cycle $S \in \cS$, let $\tria(S, j)$ denote the triangle in $S$ that intersects $V_{j}, V_{j + 1}$, and $V_{j + 2}$.

Now we are ready to give a statement similar to Lemma~\ref{lemma:1mod3:ultimate} with specific pairing conditions.
Specifically, for a pair of vertices in a strong cycle $S \in \cS$, we require that both must be part of $F$, or both must be outside $F$.
At its core, we construct the pairing as follows: for any safe vertex, the two vertices immediately following it, as well as the two vertices preceding it are paired (see Figure~\ref{fig:2mod3:pairing} for an example).
This aims to ensure the coverage of these two vertices, possibly in conjunction with the safe vertex.
Yet, the situation is still complex.
Notably, in certain scenarios (case 4 in the proof), the packing involves a triangle with a safe vertex alongside its preceding vertex and following vertex.
As these two vertices are not paired, this triangle is not necessarily available, especially when these two vertices are part of different strong cycles that intersect at the safe vertex.

\begin{figure}
  \centering
  \begin{tikzpicture}
    \begin{scope}[xscale=0.9, yscale=0.72]
    \foreach \i in {0,1,2,...,16} {
      \pgfmathtruncatemacro{\nxt}{\i + 1}
      \node[vertex,label={[label distance=2.5mm]above:$V_{\nxt}$}] (v\i) at (\i, 0) {};
    }
    \foreach \i in {0,1,2,...,15} {
      \pgfmathtruncatemacro{\nxt}{\i + 1}
      \draw (v\i) -- (v\nxt);
    }

    \draw (-0.5, 0) -- (v0);
    \draw[dotted] (-0.8, 0) -- (v0);
    \draw (16.5, 0) -- (v10);
    \draw[dotted] (16.8, 0) -- (v10);

    \draw [decorate,decoration={brace,amplitude=1.5mm,mirror,raise=4mm}] (-0.1,0) -- (1.1,0);
    \draw [decorate,decoration={brace,amplitude=1.5mm,mirror,raise=4mm}] (1.9,0) -- (3.1,0);
    \draw [decorate,decoration={brace,amplitude=1.5mm,mirror,raise=4mm}] (4.9,0) -- (6.1,0);
    \draw [decorate,decoration={brace,amplitude=1.5mm,mirror,raise=4mm}] (7.9,0) -- (9.1,0);
    \draw [decorate,decoration={brace,amplitude=1.5mm,mirror,raise=4mm}] (10.9,0) -- (12.1,0);
    \draw [decorate,decoration={brace,amplitude=1.5mm,mirror,raise=4mm}] (13.9,0) -- (15.1,0);

    \end{scope}

    \flowervertex{7 * 0.9}{0}
    \flowervertex{16 * 0.9}{0}
  \end{tikzpicture}
  \caption{An illustration of the pairing constraints in Lemma~\ref{lemma:2mod3:ultimate} on a cycle of length 17.}
  \label{fig:2mod3:pairing}
\end{figure}
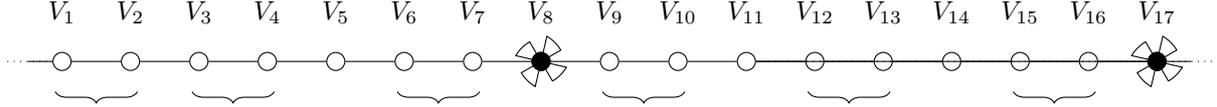

\begin{lemma} \label{lemma:2mod3:ultimate}
  Let $X \subseteq V^i$ be a set such that $X \cap V_j = \{ x_j \}$ for each intersecting block $V_j$ and $X \cap V_j = \emptyset$ for each non-intersecting block $V_j$.
  For $r' \in [r]$ and $t \in \N$,
  define $\scycle(r', t)$ as the strong cycle in $\cS$ covering $x_{\at(r', t)} \in X$.
  Suppose that the following hold:
  \begin{itemize}
    \item If  $|V(P_1)| \ne 1$, then $\scycle(1, 1) = \scycle(1, 2)$.
    \item For each $q \in [q_1]$, $\scycle(1, 3q) = \scycle(1, 3q + 1)$.
    \item For each $r' \in [2, r]$ and $q \in [q_{r'} + 1]$, $\scycle(r', 3q - 2) = \scycle(r', 3q - 1)$.
  \end{itemize}
  Under these conditions on $X$, it is possible to cover $X$ with a packing of strong cycles within $V^i$.
\end{lemma}
\begin{proof}
  Consider $x_j \in X$ such that $V_{j}$ is flowery.
  Our goal is to construct a strong cycle packing where each block $V_{j}$ is covered by at most three of its vertices.
  Given that $x_j$ is part of at least 13 triangles intersecting at $x_j$, it follows that we can select a disjoint triangle to cover $x_j$.

  We show that all $x_j$'s within non-safe blocks $V_j$ can be covered with a strong cycle packing as well.
  To that end, we initiate our analysis with respect to $Z = \scycle(1, 1)$, which is the strong cycle in $\cS$ covering~$x_1$.
  Our technical insight is the identification of the following four cases:
  \begin{enumerate}
    \item $\scycle(1, 3q_1 + 1) = Z$ and that $\scycle(r', 1) = \scycle(r', 3q_{r'} + 2) =  Z$ for each $r' \in [2, r]$. 
    \item $\scycle(1, 3q_1 + 1) \ne Z$.
    \item There exists $r' \in [2, r]$ such that $\scycle(r', 1) \ne \scycle(r', 3q_{r'} + 2)$.
    \item The remaining cases.
  \end{enumerate}
  For each case, we will establish a strategy for covering $X$.
  See Figure~\ref{fig:2mod3:ultimate} for an example of each case.

\begin{figure}
  \centering
  \begin{tikzpicture}
  \begin{scope}
    \begin{scope}[xscale=0.7, yscale=0.56]
      \node at (-.8, 2.5) {Case 1.};
      
      \draw[gray, rounded corners=10pt, line width=5mm, opacity=0.5, line cap=round] (-0.8, 0.2) -- (0, 1) -- (6, 1) -- (7, 0) -- (8, 1) -- (9, 1) -- (10, 0) -- (11, 1) -- (15, 1) -- (18, 1) -- (19, 0) -- (19.8, 0.8);
      \draw[gray, rounded corners=10pt, line width=5mm, opacity=0.5, line cap=round] (2, -1) -- (4, -1);
      \draw[gray, rounded corners=10pt, line width=5mm, opacity=0.5, line cap=round] (13, -1) -- (15, -1);
      \draw[gray, rounded corners=10pt, line width=5mm, opacity=0.5, line cap=round] (16, -1) -- (18, -1);

      \foreach \i in {0,1,2,...,6} {
        \node[vertex] (va\i) at (\i, 1) {};
        \node[vertex] (vb\i) at (\i, -1) {};
      }
      \foreach \i in {0,1,2,...,5} {
        \pgfmathtruncatemacro{\nxt}{\i + 1}
        \draw (va\i) -- (va\nxt);
        \draw (vb\i) -- (vb\nxt);
      }

      \foreach \i in {0,1} {
        \node[vertex] (wa\i) at (\i + 8, 1) {};
        \node[vertex] (wb\i) at (\i + 8, -1) {};
      }
      \draw (wa0) -- (wa1);
      \draw (wb0) -- (wb1);

      \foreach \i in {0,1,2,...,7} {
        \node[vertex] (ua\i) at (\i + 11, 1) {};
        \node[vertex] (ub\i) at (\i + 11, -1) {};
      }
      \foreach \i in {0,1,2,...,6} {
        \pgfmathtruncatemacro{\nxt}{\i + 1}
        \draw (ua\i) -- (ua\nxt);
        \draw (ub\i) -- (ub\nxt);
      }

      \node[vertex] (f1) at (7, 0) {};
      \node[vertex] (f2) at (10, 0) {};
      \node[vertex] (f3) at (19, 0) {};

      \draw (va6) -- (f1) -- (vb6);
      \draw (wa0) -- (f1) -- (wb0);
      \draw (ua0) -- (f2) -- (ub0);
      \draw (wa1) -- (f2) -- (wb1);
      \draw (ua7) -- (f3) -- (ub7);

      \draw[dotted] (19.8, 0.8) -- (f3) -- (19.8, -0.8) {};
      \draw (19.5, 0.5) -- (f3) -- (19.5, -0.5) {};
      \draw (-0.5, 0.5) -- (va0);
      \draw[dotted] (-0.8, 0.2) -- (va0);
      \draw (-0.5, -0.5) -- (vb0);
      \draw[dotted] (-0.8, -0.2) -- (vb0);

      \foreach \i in {0,1,4,5,6} {
        \node[bvertex] at (\i, 1) {};
      }
      \foreach \i in {2, 3} {
        \node[bvertex] at (\i, -1) {};
      }
      \foreach \i in {0, 1} {
        \node[bvertex] at (\i + 8, 1) {};
      }
      \foreach \i in {0, 1, 2, 6, 7} {
        \node[bvertex] at (\i + 11, 1) {};
      }
      \foreach \i in {3, 4, 5} {
        \node[bvertex] at (\i + 11, -1) {};
      }
    \end{scope}

    \flowervertex{7 * 0.7}{0}
    \flowervertex{10 * 0.7}{0}
    \flowervertex{19 * 0.7}{0}
  \end{scope}

  \begin{scope}[shift={(0, -3.5)}]
    \begin{scope}[xscale=0.7, yscale=0.56]
      \node at (-.8, 2.5) {Case 2.};
        
      \draw[gray, rounded corners=5pt, line width=5mm, opacity=0.5, line cap=round] (0, 1) -- (2, 1);
      \draw[gray, rounded corners=5pt, line width=5mm, opacity=0.5, line cap=round] (3, 1) -- (5, 1);
      \draw[gray, rounded corners=5pt, line width=5mm, opacity=0.5, line cap=round] (2, -1) -- (4, -1);
      \draw[gray, rounded corners=5pt, line width=5mm, opacity=0.5, line cap=round] (5, -1) -- (6, -1) -- (7, 0);
      \draw[gray, rounded corners=5pt, line width=5mm, opacity=0.5, line cap=round] (8, -1) -- (9, -1) -- (10, 0);
      \draw[gray, rounded corners=5pt, line width=5mm, opacity=0.5, line cap=round] (11, 1) -- (13, 1);
      \draw[gray, rounded corners=5pt, line width=5mm, opacity=0.5, line cap=round] (11, -1) -- (13, -1);
      \draw[gray, rounded corners=5pt, line width=5mm, opacity=0.5, line cap=round] (14, -1) -- (16, -1);
      \draw[gray, rounded corners=5pt, line width=5mm, opacity=0.5, line cap=round] (17, 1) -- (18, 1) -- (19, 0);

      \foreach \i in {0,1,2,...,6} {
        \node[vertex] (va\i) at (\i, 1) {};
        \node[vertex] (vb\i) at (\i, -1) {};
      }
      \foreach \i in {0,1,2,...,5} {
        \pgfmathtruncatemacro{\nxt}{\i + 1}
        \draw (va\i) -- (va\nxt);
        \draw (vb\i) -- (vb\nxt);
      }

      \foreach \i in {0,1} {
        \node[vertex] (wa\i) at (\i + 8, 1) {};
        \node[vertex] (wb\i) at (\i + 8, -1) {};
      }
      \draw (wa0) -- (wa1);
      \draw (wb0) -- (wb1);

      \foreach \i in {0,1,2,...,7} {
        \node[vertex] (ua\i) at (\i + 11, 1) {};
        \node[vertex] (ub\i) at (\i + 11, -1) {};
      }
      \foreach \i in {0,1,2,...,6} {
        \pgfmathtruncatemacro{\nxt}{\i + 1}
        \draw (ua\i) -- (ua\nxt);
        \draw (ub\i) -- (ub\nxt);
      }

      \node[vertex] (f1) at (7, 0) {};
      \node[vertex] (f2) at (10, 0) {};
      \node[vertex] (f3) at (19, 0) {};

      \draw (va6) -- (f1) -- (vb6);
      \draw (wa0) -- (f1) -- (wb0);
      \draw (ua0) -- (f2) -- (ub0);
      \draw (wa1) -- (f2) -- (wb1);
      \draw (ua7) -- (f3) -- (ub7);

      \draw[dotted] (19.8, 0.8) -- (f3) -- (19.8, -0.8) {};
      \draw (19.5, 0.5) -- (f3) -- (19.5, -0.5) {};
      \draw (-0.5, 0.5) -- (va0);
      \draw[dotted] (-0.8, 0.2) -- (va0);
      \draw (-0.5, -0.5) -- (vb0);
      \draw[dotted] (-0.8, -0.2) -- (vb0);

      \foreach \i in {0,1,2,3} {
        \node[bvertex] at (\i, 1) {};
      }
      \foreach \i in {4, 5, 6} {
        \node[bvertex] at (\i, -1) {};
      }
      \foreach \i in {0, 1} {
        \node[bvertex] at (\i + 8, -1) {};
      }
      \foreach \i in {0, 1, 6, 7} {
        \node[bvertex] at (\i + 11, 1) {};
      }
      \foreach \i in {2, 3, 4, 5} {
        \node[bvertex] at (\i + 11, -1) {};
      }
    \end{scope}

    \flowervertex{7 * 0.7}{0}
    \flowervertex{10 * 0.7}{0}
    \flowervertex{19 * 0.7}{0}
  \end{scope}

  \begin{scope}[shift={(0, -7)}]
    \begin{scope}[xscale=0.7, yscale=0.56]
      \node at (-.8, 2.5) {Case 3.};
        
      \draw[gray, rounded corners=10pt, line width=5mm, opacity=0.5, line cap=round] (-0.8, 0.2) -- (0, 1) -- (1, 1);
      \draw[gray, rounded corners=10pt, line width=5mm, opacity=0.5, line cap=round] (19, 0) -- (19.8, 0.8);
      \draw[gray, rounded corners=10pt, line width=5mm, opacity=0.5, line cap=round] (2, 1) -- (4, 1);
      \draw[gray, rounded corners=10pt, line width=5mm, opacity=0.5, line cap=round] (5, 1) -- (6, 1) -- (7, 0);
      \draw[gray, rounded corners=10pt, line width=5mm, opacity=0.5, line cap=round] (8, -1) -- (9, -1) -- (10, 0);
      \draw[gray, rounded corners=10pt, line width=5mm, opacity=0.5, line cap=round] (11, 1) -- (13, 1);
      \draw[gray, rounded corners=10pt, line width=5mm, opacity=0.5, line cap=round] (11, -1) -- (13, -1);
      \draw[gray, rounded corners=10pt, line width=5mm, opacity=0.5, line cap=round] (14, 1) -- (16, 1);
      \draw[gray, rounded corners=10pt, line width=5mm, opacity=0.5, line cap=round] (16, -1) -- (18, -1);

      \foreach \i in {0,1,2,...,6} {
        \node[vertex] (va\i) at (\i, 1) {};
        \node[vertex] (vb\i) at (\i, -1) {};
      }
      \foreach \i in {0,1,2,...,5} {
        \pgfmathtruncatemacro{\nxt}{\i + 1}
        \draw (va\i) -- (va\nxt);
        \draw (vb\i) -- (vb\nxt);
      }

      \foreach \i in {0,1} {
        \node[vertex] (wa\i) at (\i + 8, 1) {};
        \node[vertex] (wb\i) at (\i + 8, -1) {};
      }
      \draw (wa0) -- (wa1);
      \draw (wb0) -- (wb1);

      \foreach \i in {0,1,2,...,7} {
        \node[vertex] (ua\i) at (\i + 11, 1) {};
        \node[vertex] (ub\i) at (\i + 11, -1) {};
      }
      \foreach \i in {0,1,2,...,6} {
        \pgfmathtruncatemacro{\nxt}{\i + 1}
        \draw (ua\i) -- (ua\nxt);
        \draw (ub\i) -- (ub\nxt);
      }

      \node[vertex] (f1) at (7, 0) {};
      \node[vertex] (f2) at (10, 0) {};
      \node[vertex] (f3) at (19, 0) {};

      \draw (va6) -- (f1) -- (vb6);
      \draw (wa0) -- (f1) -- (wb0);
      \draw (ua0) -- (f2) -- (ub0);
      \draw (wa1) -- (f2) -- (wb1);
      \draw (ua7) -- (f3) -- (ub7);

      \draw[dotted] (19.8, 0.8) -- (f3) -- (19.8, -0.8) {};
      \draw (19.5, 0.5) -- (f3) -- (19.5, -0.5) {};
      \draw (-0.5, 0.5) -- (va0);
      \draw[dotted] (-0.8, 0.2) -- (va0);
      \draw (-0.5, -0.5) -- (vb0);
      \draw[dotted] (-0.8, -0.2) -- (vb0);

      \foreach \i in {0,1,2,3,4,5,6} {
        \node[bvertex] at (\i, 1) {};
      }
      \foreach \i in {0, 1} {
        \node[bvertex] at (\i + 8, -1) {};
      }
      \foreach \i in {0, 1, 3, 4} {
        \node[bvertex] at (\i + 11, 1) {};
      }
      \foreach \i in {2, 5, 6, 7} {
        \node[bvertex] at (\i + 11, -1) {};
      }
    \end{scope}

    \flowervertex{7 * 0.7}{0}
    \flowervertex{10 * 0.7}{0}
    \flowervertex{19 * 0.7}{0}
  \end{scope}

  \begin{scope}[shift={(0, -10.5)}]
    \begin{scope}[xscale=0.7, yscale=0.56]
      \node at (-.8, 2.5) {Case 4.};
        
      \draw[gray, rounded corners=10pt, line width=5mm, opacity=0.5, line cap=round] (0, 1) -- (2, 1);
      \draw[gray, rounded corners=10pt, line width=5mm, opacity=0.5, line cap=round] (3, 1) -- (5, 1);
      \draw[gray, rounded corners=10pt, line width=5mm, opacity=0.5, line cap=round] (6, 1) -- (7, 0) -- (8, 1);
      \draw[gray, rounded corners=10pt, line width=5mm, opacity=0.5, line cap=round] (9, 1) -- (10, 0) -- (11, 1);
      \draw[gray, rounded corners=10pt, line width=5mm, opacity=0.5, line cap=round] (11, -1) -- (13, -1);
      \draw[gray, rounded corners=10pt, line width=5mm, opacity=0.5, line cap=round] (14, 1) -- (16, 1);
      \draw[gray, rounded corners=10pt, line width=5mm, opacity=0.5, line cap=round] (14, -1) -- (16, -1);
      \draw[gray, rounded corners=10pt, line width=5mm, opacity=0.5, line cap=round] (17, -1) -- (18, -1) -- (19, 0);

      \foreach \i in {0,1,2,...,6} {
        \node[vertex] (va\i) at (\i, 1) {};
        \node[vertex] (vb\i) at (\i, -1) {};
      }
      \foreach \i in {0,1,2,...,5} {
        \pgfmathtruncatemacro{\nxt}{\i + 1}
        \draw (va\i) -- (va\nxt);
        \draw (vb\i) -- (vb\nxt);
      }

      \foreach \i in {0,1} {
        \node[vertex] (wa\i) at (\i + 8, 1) {};
        \node[vertex] (wb\i) at (\i + 8, -1) {};
      }
      \draw (wa0) -- (wa1);
      \draw (wb0) -- (wb1);

      \foreach \i in {0,1,2,...,7} {
        \node[vertex] (ua\i) at (\i + 11, 1) {};
        \node[vertex] (ub\i) at (\i + 11, -1) {};
      }
      \foreach \i in {0,1,2,...,6} {
        \pgfmathtruncatemacro{\nxt}{\i + 1}
        \draw (ua\i) -- (ua\nxt);
        \draw (ub\i) -- (ub\nxt);
      }

      \node[vertex] (f1) at (7, 0) {};
      \node[vertex] (f2) at (10, 0) {};
      \node[vertex] (f3) at (19, 0) {};

      \draw (va6) -- (f1) -- (vb6);
      \draw (wa0) -- (f1) -- (wb0);
      \draw (ua0) -- (f2) -- (ub0);
      \draw (wa1) -- (f2) -- (wb1);
      \draw (ua7) -- (f3) -- (ub7);

      \draw[dotted] (19.8, 0.8) -- (f3) -- (19.8, -0.8) {};
      \draw (19.5, 0.5) -- (f3) -- (19.5, -0.5) {};
      \draw (-0.5, 0.5) -- (va0);
      \draw[dotted] (-0.8, 0.2) -- (va0);
      \draw (-0.5, -0.5) -- (vb0);
      \draw[dotted] (-0.8, -0.2) -- (vb0);

      \foreach \i in {0,1,2,3,4,5,6} {
        \node[bvertex] at (\i, 1) {};
      }
      \foreach \i in {0, 1} {
        \node[bvertex] at (\i + 8, 1) {};
      }
      \foreach \i in {3, 4} {
        \node[bvertex] at (\i + 11, 1) {};
      }
      \foreach \i in {0, 1, 2, 5, 6, 7} {
        \node[bvertex] at (\i + 11, -1) {};
      }
    \end{scope}

    \flowervertex{7 * 0.7}{0}
    \flowervertex{10 * 0.7}{0}
    \flowervertex{19 * 0.7}{0}
  \end{scope}
  \end{tikzpicture}
  \vspace{5ex}
  \caption{Examples of strong cycle packings in the proof of Lemma~\ref{lemma:2mod3:ultimate}. There are two strong cycles that intersect two flowery vertices. The vertices in $X$, which satisfy the condition specified in Lemma~\ref{lemma:1mod3:ultimate}, are marked in black. For each of the four cases, a strong cycle packing that covers $X$ is indicated in gray.}
  \label{fig:2mod3:ultimate}
\end{figure}
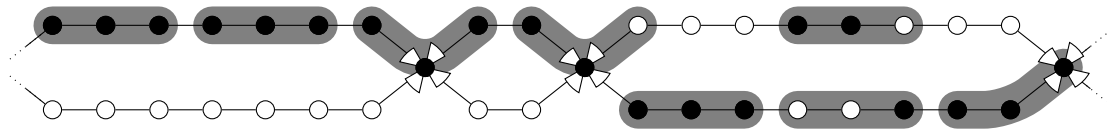

  \paragraph*{Case 1.}
  Suppose that $\scycle(1, 3q_1 + 1) = Z$ and that $\scycle(r', 1) = \scycle(r', 3q_{r'} + 2) =  Z$ for each $r' \in [2, r]$. 
  We construct a strong cycle packing $\cP$ that covers $X$ as follows.
  Initially, let $\cP = \{ Z \}$.
  \begin{itemize}
    \item 
    The part of $X$ corresponding to $P_1$ can be covered as follows.
    First, note that the part corresponding to the first and last two vertices of $P_1$ is covered by $Z$.
    For the remaining part of $P_1$, add to $\cP$ the triangles $\tria(\scycle(1, 3q), 3q) = \tria(\scycle(1, 3q+1), 3q)$ and $\tria(\scycle(1, 3q+2), 3q)$ for $q \in [q_1 - 1]$.
    Since strong cycles in $\cS$ are vertex-disjoint at non-flowery blocks as specified by Lemma~\ref{lemma:2mod3:find-cycles}, any two of these triangles are either identical or entirely vertex-disjoint.
    \item
    Next, we consider the part of $X$ corresponding to $P_{r'}$ for $r' \in [2, r]$.
    By the condition of the lemma, $\scycle(r', 2) = \scycle(r', 3q_{r'} + 1) =  Z$ for each $r' \in [2, r]$.
    Thus, the first two vertices $x_{\at(r',1)}, x_{\at(r', 2)}$ as well as the last two vertices $x_{\at(r', 3q_{r'} + 1)}, x_{\at(r', 3q_{r'} + 2)}$ are covered by $Z$.
    As above, add to $\cP$ the triangles $\tria(\scycle(1, 3q), \at(r', 3q))$ and  $\tria(\scycle(1, 3q+1), \at(r', 3q)) = \tria(\scycle(1, 3q+2), \at(r', 3q))$ for $q \in [q_{r'} - 1]$.
    These cover all but $x_{\at(r', 3q_{r'})}$.
    If it has not been covered, i.e., $\scycle(r', 3q_{r'}) \ne Z$, then add the triangle $\tria(\scycle(r', 3q_{r'}), r', 3q)$ to $\cP$, which is disjoint from $Z$.
  \end{itemize}

  \paragraph{Case 2.}
  Suppose that $\scycle(1, 1) \ne \scycle(1, 3q_{1} + 1)$.
  We consider the smallest integer $q$ such that $\scycle(1, 3q) \ne Z$
  Note that such an integer always exists, as $\scycle(1, 3q_{1} + 1) = \scycle(1, 3q_{1})$ by the condition of the lemma.
  Now we construct a triangle packing $\cP$ that covers $X$.
  We first consider the part of $X$ corresponding to $P_1$:
  \begin{itemize}
    \item For each $q' \in [q]$, include a triangle $\tria(\scycle(1, 3q'-2), 1, 3q'-2)$.
    Note that these triangles cover $x_{1}, x_{2}$, and $x_{3q'}$ and $x_{3q'+1}$ for each $q' \in [q-1]$.
    \item For each $q' \in [q, q_{1}]$, include a triangle $\tria(\scycle(1, 3q'), 3q')$.
    Although both $\tria(\scycle(3q-2), \at(1, 3q-2))$ and $\tria(\scycle(1, 3q), \at(1, 3q))$ intersect the block $V_{3q}$, they are disjoint.
    This is because $\scycle(1, 3q-2) = \scycle(1, 3q-3) = Z$, where the first equality follows from the condition of the lemma and the second from the definition of $q$, and $\scycle(1, 3q) \ne Z$.
    Therefore, $\scycle(1,3q-2)$ and $\scycle(1,3q)$ are distinct.
    These triangles cover $x_{3q'}$ and $x_{3q'+1}$ for all $q' \in [q, q_{1}]$.
    \item 
    If $x_{\at(r',3q'+2)}$ for $q' \in [q_1 - 1]$ is not covered yet, then add a triangle $\tria(\scycle(1, 3q'+2), 3q')$.
  \end{itemize}
  The part of $X$ corresponding to $P_{r'}$ for $r' \in [2, r]$ can be covered with the following triangles:
  \begin{itemize}
    \item For each $q' \in [q_{r'}]$, one or two triangles $\tria(\scycle(r', 3q' - 2), \at(r', 3q' - 2)) = \tria(\scycle(r', 3q' - 1), \at(r', 3q' - 2))$ and $\tria(\scycle(r', 3q'), \at(r', 3q' - 2))$.
    \item A triangle $\tria(\scycle(r', 3q_{r'} + 2), \at(r', 3q_{r'}))$ that includes a vertex from a flowery block.
  \end{itemize}

  \paragraph{Case 3.}
  Suppose that $\scycle(r', 1) \ne \scycle(r', 3q_{r'} + 2)$ for $2 \le r' \le r$.
  We consider the smallest integer $q$ such that $\scycle(r', 3q + 1) \ne \scycle(r', 1)$.
  Note that such an integer always exists, as $\scycle(r', 3q_{r'} + 2) = \scycle(r', 3q_{r'} + 1)$ by the condition of the lemma.
  The part of $X$ corresponding to $P_{r'}$ can be covered by the following triangles:
  \begin{itemize}
    \item For each $q' \in [q]$, include a triangle $\tria(\scycle(r', 1), \at(r', 3q' - 2))$, which covers $x_{\at(r', 3q' - 2)}$ and $x_{\at(r', 3q' - 1)}$.
    \item For each $q' \in [q, q_{r'}]$, include a triangle $\tria(\scycle(r', 3q' + 1), \at(r', 3q')) = \tria(\scycle(r',3q' + 2, 3q'))$, which covers $x_{\at(r', 3q' + 1)}$ and $x_{\at(r', 3q' + 2)}$.
    Though $\tria(\scycle(r', 1), \scycle(r', 3q - 2))$ and $\tria(\scycle(r', 3q + 1), \at(r, 3q))$ both intersect the block $V_{3q}$, they are disjoint by the definition of $q$.
    \item If $x_{\at(r', 3q')}$ for $q' \in [q_{r'}]$ is not covered, then add a triangle at $\tria(\scycle(r', 3q'), \at(r', 3q'-2))$.
  \end{itemize}
  The part of $X$ corresponding to $P_1$ can be covered with two adjacent flowery vertices:
  \begin{itemize}
    \item $\tria(\scycle(1, 1), 3p+2)$ and $\tria(\scycle(1, 3q_1), 3q_1)$ that include flowery vertices.
    \item For each $q' \in [q_1]$, $\tria(\scycle(3q', 1), 3q')$ and $\tria(\scycle(3q'+2, 1), 3q')$.
  \end{itemize}
  For other paths $P_{r''}$, we include the following triangles together with one flowery vertex:
  \begin{itemize}
    \item If $r'' \in [2, r' - 1]$, then $\tria(\scycle(r'', 3q'-2), \at(r'', 3q'-2))$ and $\tria(\scycle(r'', 3q'), \at(r'', 3q'-2))$ for each $q' \in [q_{r''}]$, as well as $\tria(\scycle(r'', 3q_{r'}+1), \at(r'', 3q_{r'}+1))$.
    \item If $r'' \in [r' + 1, r]$, then $\tria(\scycle(r'', 3q'), \at(r'', 3q'+1))$ and $\tria(\scycle(r'', 3q'+1), \at(r'', 3q'))$ for each $q' \in [q_{r''}]$, as well as $\tria(\scycle(r'', 1), \at(r'', 1) - 1)$.
  \end{itemize}

  \paragraph{Case 4.}
  Suppose that $\scycle(r', 1) = \scycle(r', 3q_{r'} + 2)$ for each $r' \in [2, r]$.
  Define $r' \in [r]$ as the smallest integer such that $\scycle(r', 1) \ne Z$, and let $\at(r', 1) = 3q$.
  For each $q' \in [q]$, we include a triangle $\tria(Z, 3q'-2)$.
  When the triangle intersects a flowery block, i.e., $V_{3q'-1}$ is a flowery block between $P_{r''-1}$ and $P_{r''}$ by Lemma~\ref{lemma:2mod3:flowery-index}, and the triangle covers $x_{3q'-2}$ and $x_{3q'}$. 
  Thus, these triangles cover the elements of $X$ that corresponds to the first and last vertices of each path $P_{r''}$ for $r'' \in [2, r']$.
  The inner vertices can be covered with the following triangles:
  \begin{itemize}
    \item 
    For $r'' = 1$, triangles $\tria(\scycle(1, 3q'), 3q'-2)$, $\tria(\scycle(1, 3q'+1), 3q'-2)$, and $\tria(\scycle(1, 3q'+2), 3q'-2)$ for $q' \in [q_{1}-1]$.
    \item
    For $r'' \in [2, r']$, triangles $\tria(\scycle(r'', 3q'-1), 3q'-1)$, $\tria(\scycle(r'', 3q'), 3q'-1)$, and $\tria(\scycle(r'', 3q'+1), 3q'-1)$ for $q' \in [q_{r''}]$.
  \end{itemize}
  For the path $P_{r'}$, the triangles $\tria(\scycle(r', 3q'-2), \at(r', 3q'-2)) = \tria(\scycle(r', 3q'-1), \at(r', 3q'-2))$ and $\tria(\scycle(r', 3q'), \at(r', 3q'-2))$ for each $q' \in [q_{r'}]$ cover the corresponding elements in $X$.
  While the block $V_{\at(r', 1)}$ is intersected by $\tria(\scycle(r', 1), \at(r', 3q'-2))$ as well as $\tria(\scycle(r'-1, 3q_{r'-1} + 2), 3q_{r'-1}+2)$, they are disjoint as $\scycle(r', 1) \ne Z$.
  The remaining part of $X$ can be covered using an argument in the second part of case 2.
\end{proof}

\subsubsection{Algorithm overview} \label{ssec:tcam:overview}

In the introductory part of Section~\ref{ssec:tcamfpt}, we showed that the objective is to find
a triangle-maximal strong cycle packing $\cC = \{ C^1, \cdots, C^{\kappa} \}$, where $C^i = (v^i_{1}, \dots, v^i_{l_i})$, covering a set $F$ of size~$2k$ that is feasible in the dual matching delta-matroid $D$.
Via a color-coding argument, we guess a partitioning $V= V^1 \cup \cdots \cup V^{\kappa}$ where each subset $V^i$ is further divided into $V_{1}^i \cup \cdots V_{l_i}^i$.
Then $v_j^i \in V_{j}^i$ holds for all $i \in [\kappa]$ and $j \in [l_i]$ with probability $1/k^{O(k)}$.
Our algorithm gives a reduction to the \textsc{Colorful Delta-matroid Matching} problem as follows.

\begin{itemize}
  \item For $i \in [\kappa]$ with $l_i \bmod 3 = 0$, we may assume that $l_i = 3$ by Lemma~\ref{lemma:tri-maximal}, so this case aligns with \textsc{Delta-matroid Triangle Cover}.
  \item Consider $i \in [\kappa]$ with $l_i \bmod 3 = 1$.
  By Lemma~\ref{lemma:tri-maximal} and Lemma~\ref{lemma:1mod3:non-flowery}, we may assume that $v_j^i$ is part of $F$, and that $v_j^i$ is non-flowery.
  After deleting all flowery vertices from $V^i$, we find a set $S \subseteq V^i$ according to Lemma~\ref{lemma:1mod3:ultimate}.
  We adjust the delta-matroid $D$ via $1$-contraction by $S \cap V_j^i$ for each $j \in [l_i]$.
  \item
  Consider $i \in [\kappa]$ with $l_i \bmod 3 = 2$.
  First, we guess whether each $v^i_j$ is flowery or not and part of $F$ or not.
  Flowery (non-flowery) vertices are then from a flowery (non-flowery) block $V_j$.
  In the absence of flowery vertices in $C^i$, we find a set $S \subseteq V^i$ following Lemma~\ref{lemma:2mod3:no-flowery}.
  For each $j \in [l_i]$, apply 1-contraction by $S \cap V_j^i$ if $v^i_j$ is part of $F$, and delete $S \cap V_j^i$ from $D$ otherwise.
  If $C^i$ contains a flowery vertex, we may assume the configuration aligned with Lemma~\ref{lemma:2mod3:path-lengths}, where $|V(P_1)| \mod 3 = 1$ and $|V(P_{r'})| \mod 3 = 2$.
  As per Lemma~\ref{lemma:2mod3:find-cycles}, let $\cS$ be a collection of strong cycles $\cS$ which are disjoint at non-flowery blocks.
  We start off by applying 1-contractions on flowery blocks $V^i_j$, where $v_j$ is part of $F$.
  We then encode pairing constraints as given in Lemma~\ref{lemma:2mod3:ultimate} by coloring the pairs.
  Specifically, for $P_r$, which has $3q_{r'} + 2$ (or $3q_1 +1$ for $r' = 1$) vertices, we introduce $q_{r'} + 1$ colors $c^i_{r',1}, \dots, c^i_{r',q_{r'}+1}$.
  In the context of $P_1$, for each strong cycle in $S \in \cS$, and we establish a colored pairing between the elements of $V_1 \cap V(S)$ and $V_2 \cap V(S)$ using color $c^i_{1,1}$.
  Similarly, for each $q \in [q_1]$, a pair is created between the elements of $V_{3q} \cap V(S)$ and $V_{3q+1} \cap V(S)$, marked in the $c^i_{1,q+1}$.
  For each path $P_{r'}$ with $r' \in [2, r]$, and for each $q \in [q_r+1]$, we have a pairing between the elements in $V_{3q-2} \cap V(S)$ and $V_{3q-1} \cap V(S)$ with the color $c^i_{r',q}$ for each strong cycle in $S \in \cS$.
  Finally, for each block $V_j$ within the path $P_{r'}$ for $r' \in [r]$ that is not involved in any pairing, we apply 1-contraction by $V(\cS) \cap V_j$. 
\end{itemize}
This concludes the description of our algorithm.
Now we proceed to validate the correctness.

Assume that there is  a strong packing covering feasible set $F$ of size $2k$ covered.
We verify that the algorithm affirm the existence of such a packing with probability $1/2^{O(k)}$.
For a strong cycle $C^i$ of length $1 \bmod 3$, the set $S$ of Lemma~\ref{lemma:1mod3:ultimate} covers $C^i$ with probability at least $1/2^{O(l_i)}$.
Similarly,  in cases where $C^i$ has length $2 \bmod 3$ and lacks flowery vertices, the outcome aligns with the previous scenario by Lemma~\ref{lemma:2mod3:no-flowery}.
For a strong cycle $C^i$ of length $2 \bmod 3$, the collection of strong cycles $\cS$ contains $C^i$ with probability $1/2^{O(l_i)}$ by Lemma~\ref{lemma:2mod3:find-cycles}.
If $C^i$ is indeed part of $\cS$, then all the requisite colored pairings are introduced.
Aggregating these probabilities across all cycles, our algorithm terminates in the affirmative with probability $1/2^{O(k)}$.

Conversely, assume that the algorithm concludes in the affirmative, i.e., there is a feasible set $F$ that is in compliance with the colored pairing constraints, as well as 1-contraction conditions.
First, consider a strong cycle $C^i$ of length $1 \bmod 3$, where 1-contraction by the intersection $S \cap V_j$ is applied for each block $V_j$.
This ensures that $F \cap V^i$ comprises exactly one vertex from $S \cap V^i_j$.
By Lemma~\ref{lemma:1mod3:ultimate} (ii), the intersection $S \cap V^i_j$ indeed can be covered by a strong cycle packing within $V^i$.
The situation is similar for a strong cycle $C^i$ of length $2 \bmod 3$ without flowery vertices, where the coverage  of $F \cap V^i$ with a strong cycle packing is guaranteed Lemma~\ref{lemma:2mod3:no-flowery} (ii).
Finally, consider a strong cycle $C^i$ of length $2 \bmod 3$ with flowery vertices.
As our algorithm introduces pairing constraints in alignment with the condition of Lemma~\ref{lemma:2mod3:ultimate}, there exists a strong cycle packing that covers $F \cap V^i$.

Consequently, we have established the algorithm's correctness in the presence of a perfect matching.
Finally, we address the case where $G$ lacks a perfect matching, repeating the approach used in the proof of Lemma~\ref{lemma:lifting}.
Assume that $E'$ is a strong edge set of size $\ell$.
Initially, fix a maximum matching $M$ of $G$, and define $U = V(G) \setminus V(M)$ as the set of unmatched vertices.
For each $u \in U$, we introduce a degree-one neighbor $u'$ adjacent to $u$, and let $U' = \{ u' \mid u \in U \}$.
Call the resulting graph $G'$, which has a perfect matching $M' = M \cup \{ uu' \mid u \in U \}$.
Within $G'$, there exists a strong set $E''$ of size $\ell + \alpha$, where $\alpha \le |U|$ represents the number of vertices in $U$ that is not incident with any edge in $E'$.
Our objective then shifts to finding a strong edge set with at least $\ell' = \ell + \alpha$ edges in $G'$ where $\beta = |U| - \alpha$ vertices of $U'$ are isolated.
Let $D'$ be the dual matching delta-matroid of $G'$.
By Lemma~\ref{lemma:tc-feasible-cover}, the existence of a strong set of size $\ell'$ is equivalent to the existence of a strong cycle packing covering a feasible set $F$ of size  $2(\ell' - |M'|) = 2(k - \beta)$.
For the additional requirement that $\beta$ vertices of $U'$ becomes isolated, note that the isolation of $u' \in U'$ is achieved when it is part of $F$:
By Lemma~\ref{lemma:alt-paths}, a set $F$ is feasible in $D'$ if and only if there are $|F|/2$ $M'$-alternating $F$-paths.
It follows from the proof of Lemma~\ref{lemma:tc-feasible-cover} that the edge $uu'$ is excluded from a strong set when it is part of $|F|/2$ $M$-alternating $F$-paths, i.e., when $u' \in F$. 
Consequently, we apply a $\beta$-contraction by $U'$ on $D'$, calling the resulting delta-matroid $D$.
We then search for a strong cycle packing covering a set of size $2k - 3\beta$ feasible in $D$.
This concludes the proof of Theorem~\ref{theorem:tc-am-fpt}.

\subsection{W[1]-hardness proof}\label{ssec:STC-w1-hard}
While lifting the FPT algorithm for \textsc{Cluster Subgraph} from $K_4$-free graphs to arbitrary graphs was relatively straightforward, this is not the case for \textsc{Strong Triadic Closure}. Indeed, we will show that the interaction in a strong set between $K_4$-free deletion set and the rest of the graph is much more intricate and it is not possible to be handled just by additional color-coding. Even more surprisingly, we show that \textsc{Strong Triadic Closure} is W[1]-hard parameterized by above matching.
\begin{theorem}
  \textsc{Strong Triadic Closure} is W[1]-hard parameterized by above matching.
\end{theorem}

\begin{proof}
We reduce from \textsc{Grid Tiling}, which is formally defined as follows. Given integers $k,n\in \mathbb{N}$ and a collection $\mathcal{S}$ of $k^2$ nonempty sets $S_{i,j}\subseteq [n]\times [n]$ $ (i,j\in [k])$. The task is to find for each $i,j\in [k]$ a pair $s_{i,j}\in S_{i,j}$ such that 
\begin{itemize}
\item[(C1)] If $s_{i,j} = (a,b)$ and $s_{i+1,j} = (a',b')$, then $a=a'$. (For all $i\in [k-1], j\in [k]$.)
\item[(C2)] If $s_{i,j} = (a,b)$ and $s_{i,j+1} = (a',b')$, then $b=b'$. (For all $i\in [k], j\in [k-1]$.)
\end{itemize} 
\textsc{Grid Tiling} is well known to be W[1]-hard parameterized by $k$~\cite{CyganFKLMPPS15PCbook}. We will construct a graph $H$, where $V(H)$ is partitioned into five parts $P \cup Q \cup M \cup R\cup C$, where $|P|=\sum_{S\in \mathcal{S}}|S|$, $|Q|= 2k^2$, $|M| = |P|+|Q|$ and $|R|= k^2-k$ and $|C|= k^2-k$.  

Let us start with very informal intuition behind the proof. 
The set of vertices $P$ can be thought of as partitioned into $k^2$ sets $P_{i,j}$ $(1\le i,j\le k)$ such that for every $i,j\in [k]$ and every $s\in S_{i,j}$, $P_{i,j}$ contains a vertex $p_s^{i,j}$ representing selection of $s$ from $S_{i,j}$. In addition $H$ contains edges between vertices in $P$ corresponding to Conditions (C1) and (C2) from the definition of Grid Tiling. 
To select $p_s^{i,j}$, $Q$ contains an edge $q^{i,j}_1q^{i,j}_2$ that can form a triangle with a single $p_s^{i,j}$ in a strong set. The set $M$ then contains for every $v\in P\cup Q$ a unique pendant $m(v)$ of $v$ in $H$. $H' = H[P\cup Q \cup M]$ will be $K_4$-free and half of the vertices (that are in $M$) there have degree one. So any maximum strong set in $H'$ has exactly the size of the perfect matching in $H'$, which is $|P|+|Q|$, regardless whether we used $p_s^{i,j}q^{i,j}_1q^{i,j}_2$ triangles or just match vertices in $P\cup Q$ to their matching partners in $M$. However, only one vertex in each $P_{i,j}$ can form a triangle with $q^{i,j}_1q^{i,j}_2$ in the strong set. All the remaining vertices $p\in P$ have to be matched with their matching partner $m(p)$.

Finally $R\cup C$ will be a clique on $2k^2-2k$ vertices, so the excess above the perfect matching in a maximum size strong set will be obtained from this large clique and edges between $R\cup C$ and $P\cup Q$. The vertices in $R$ will serve to check the condition (C1). That is a vertex $r\in R$ will be only adjacent to all vertices in $P$, but only to two edges from $Q$ that represent $S_{i,j}$ and $S_{i+1,j}$. This will force that in a strong set, $r$ can have only two edges to $P\cup Q$, one to $p_1\in P_{i,j}$ and one to $p_2\in P_{i+1,j}$ and $p_1,p_2$ have to also be adjacent. Analogously, each $c\in C$ will serve to check the condition (C2) for some $i,j\in [k]$.

\paragraph*{Construction.}
Let us now proceed with more precise formal construction. 

We start with set $P$.
For every $1\le i,j\le k$, we let $P_{i,j} = \{ p_s^{i,j} \mid s \in S_{i,j} \}$ and we let 
$P = \bigcup_{1\le i,j \le k }P_{i,j}$
that is for every $i,j\in [k]$ and every $s\in S_{i,j}$ we introduce vertex $p_{s}^{i,j}$. In addition $H$ contains an edge between two vertices $p_1,p_2\in P$ if and only if $p_1$ and $p_2$ represent a selection matching one of the conditions (C1) or (C2). Or in other words if and only if
\begin{itemize}
\item either $\{p_1, p_2\} = \{p_s^{i,j}, p_{s'}^{i+1,j}\}$ with $i\in [k-1]$, $j\in [k]$, $s = (a,b)$, $s' = (a',b')$, and $a=a'$; or
\item $\{p_1, p_2\} = \{p_s^{i,j}, p_{s'}^{i,j+1}\}$ with $i\in [k]$, $j\in [k-1]$, $s = (a,b)$, $s' = (a',b')$, and $b=b'$.
\end{itemize} 

Let $Q$ be defined as $Q = \{ q_{1}^{i,j}, q_2^{i,j} \mid i, j \in [k] \}$.
For each $i, j \in [k]$, we introduce the edge $q_1^{i,j}, q_2^{i,j}$.
Furthermore, for each $i, j \in [k]$, and every $p\in P_{i,j}$ we add edges $pq_1^{i,j}$ and $pq_2^{i,j}$, so that $\{ q_{i,j,1}, q_{i,j,2}, p \}$ forms a triangle.

For every vertex $v\in P\cup Q$, $M$ contains the vertex $m(v)$ and $H$ contains the edge between $v$ and $m(v)$.

Let $R = \{ r^{i,j} \mid i \in [k-1], j\in [k] \}$ and $C = \{ c^{i,j} \mid i \in [k], j\in [k-1] \}$. 
We add edges so that $R\cup C$ forms a clique.
We also add an edge between each $p \in P$ and each $v \in R\cup C$. 
For each $i, j \in [k]$, if $i< k$ we add edges from $r^{i,j}$ to $q_1^{i,j}$, $q_2^{i,j}$, $q_1^{i+1,j}$, $q_2^{i+1,j}$, and if $j<k$ we also add edges from $c^{i,j}$ to $q_1^{i,j}$, $q_2^{i,j}$, $q_1^{i,j+1}$, $q_2^{i,j+1}$.

Finally, we let $\ell = |P|+ |Q| + 2|R\cup C| + \binom{|R\cup C|}{2}$.

\paragraph{Correctness.} First note that $H$ has a perfect matching (matching $v\in P\cup Q$ to $m(v)$ and matching $r^{i,j}\in R$ with $c^{j,i}\in C$), which has size $\mu = |P| + |Q| + \frac{|R\cup C|}{2}$ and $\ell - \mu = 3(k^2-k) + \binom{2k^2-2k}{2}$, so the above matching parameter is indeed bounded by a function of $k$. Hence it only remains to show that the two instances are equivalent. 

$(\Rightarrow)$ Let $S = \{s_{i,j}\mid i,j\in [k], s_{i,j}\in S_{i,j} \}$ be a set of pairs satisfying (C1) and (C2). We construct a strong edge set $S'$ of size $\ell$ as follows: 
\begin{itemize}
\item For every $i,j\in[k]$, include triangle $\{p_{s_{i,j}}^{i,j}, q_1^{i,j}, q_2^{i,j}\}$.
\item For each $p\in P$ not involved in these triangles, include the edge $pm(p)$.
\item Include all $\binom{2k^2-2k}{2}$ edges in the clique on $R\cup C$.
\item For each $i\in [k-1]$, $j\in [k]$, add edges $r^{i,j}p_{s_{i,j}}^{i,j}$ and $r^{i,j}p_{s_{i+1,j}}^{i+1,j}$.
\item For each $i\in [k]$, $j\in [k-1]$, add edges $c^{i,j}p_{s_{i,j}}^{i,j}$ and $c^{i,j}p_{s_{i,j+1}}^{i,j+1}$.
\end{itemize}
It is easy to verify that precisely $\ell = |P|+ |Q| + 2|R\cup C| + \binom{|R\cup C|}{2}$ many edges has been selected to $S'$. We only need to check that $S'$ forms a strong set or in other words that it satisfies the triadic closure property. First observe that the edges $pm(p)$ included in $S'$ are isolated, that is, neither $p$ nor $m(p)$ is incident with any other edges of $S'$. Consequently we identify the following types of $P_3$'s in $S$: 
\begin{itemize}
\item 
$(x,p,y)$ for $x,y\in R\cup C$ and $p\in P$. These are closed as $R\cup C$ induces a clique.
\item $(r^{i,j}, p_{s_{i,j}}^{i,j}, q^{i,j}_{\alpha})$ and $(r^{i,j}, p_{s_{i+1,j}}^{i+1,j}, q^{i+1,j}_{\alpha})$ for $\alpha\in \{1,2\}$ are closed because our construction adds edges from $r^{i,j}$ to  $q_{\alpha}^{i,j}$, $q_{\alpha}^{i+1,j}$.
\item 
$(c^{i,j}, p_{s_{i,j}}^{i,j}, q^{i,j}_{\alpha})$ and $(c^{i,j}, p_{s_{i,j+1}}^{i,j+1}, q^{i,j+1}_{\alpha})$ for $\alpha\in \{1,2\}$ are closed because our construction adds edges from $c^{i,j}$ to  $q_{\alpha}^{i,j}$, $q_{\alpha}^{i,j+1}$.
\item  
$(p_{s_{i,j}}^{i,j},r^{i,j}, p_{s_{i+1,j}}^{i+1,j})$ are closed because since $S$ is a solution for \textsc{Grid Tiling}, the pair $s_{i,j}$ and $s_{i+1,j}$ satisfies condition (C1), and so our construction adds the edge $p_{s_{i,j}}^{i,j}p_{s_{i+1,j}}^{i+1,j}$.
\item 
$(p_{s_{i,j}}^{i,j},c^{i,j}, p_{s_{i,j+1}}^{i,j+1})$ are closed because the pair $s_{i,j}$ and $s_{i,j+1}$ satisfies condition (C2), and so our construction adds the edge $p_{s_{i,j}}^{i,j}p_{s_{i,j+1}}^{i,j+1}$.
\item
  $(p, r, r')$ for $p \in P$ and $r, r' \in R$ is closed as there is an edge between each $p \in P$ and $r \in R$.
\end{itemize}
Therefore $S$ is indeed a strong set of size $\ell$.  

$(\Leftarrow)$
Let $S$ be a strong edge set of size $\ell$. We show that for every $i,j\in [k]$, we can select $s_{i,j}\in S_{i,j}$ such that the selection satisfies conditions (C1) and (C2).
 
Let us first observe that solution can contain at most $\binom{|R\cup C|}{2}$ edges from $R\cup C$ and $|P|+|Q|$ many edges from $P\cup Q\cup M$. The former is trivial, as that is the number of edges in $R\cup C$, the latter follows from the fact that $H[P\cup Q\cup M]$ is $K_4$-free and all vertices in $M$ have degree one in $H$. So every vertex $v$ in $P\cup Q$ has at most two of its incident edges from $H[P\cup Q\cup M]$ in $S$, and it can only be two if neither of them is $vm(v)$, so if $x$ many vertices in $v\in P\cup Q$ are matched with their matching partner $m(v)$ in $S$, then the maximum number of edges from $H[P\cup Q\cup M]$ in $S$ is $\frac{2(|P|+|Q|-x)}{2} + x = |P|+|Q|$. 

We will now show that for every vertex $v$ in $R\cup C$, there are at most two edges incident on $v$ and a vertex in $P$ and there is no edge in $S$ incident on $v$ and a vertex in $Q$.  
Note that since $\ell = |P|+ |Q| + \binom{|R\cup C|}{2} + 2|R\cup C| $, this immediately means that $S$ contains precisely  $|P|+ |Q|$ many edges from $P\cup Q\cup M$, $\binom{|R\cup C|}{2}$ edges from $R\cup C$, and $2|R\cup C|$ edges between  $R\cup C$ and $P$. 
First observe that $H[P]$ is bipartite and hence triangle-free. Indeed, for $p\in P_{i,j}$ and $p'\in P_{i',j'}$, we have that $pp'$ is an edge in $H$, then either $i'=i+1$ and $j'=j$ or $i'=i$ and $j'=j+1$. In addition, recall that neighborhood of every vertex in the subgraph of $H$ induced by $S$ has to be a clique in $H$. It follows that every vertex $v\in R\cup C$ has at most two neighbors $p_1, p_2\in P$ such that $vp_1, vp_2\in S$. Since $H[P\cup Q]$ is $K_4$-free, there are at most three edges in $S$ that are incident on $v$ and a vertex in $P\cup Q$.
Let $v\in R\cup C$ and $q\in Q$. Note that $q$ is adjacent to exactly two vertices in $R$ and exactly two vertices in $v$. Therefore, if $qv\in S$, then $S$ contains at most $3$ out of $|R\cup C|-1$ edges with one endpoint being $v$ and the other endpoint in $R\cup C$, or equivalently at least $|R\cup C|-4$ edges incident on $v$ in $R\cup C$ are missing from $S$. Now let $X$ be the set of vertices $v\in R\cup C$ for which there exists $q\in Q$ such that $vq\in S$. We can upper bound the number of edges in $S$ incident on a vertex in $R\cup C$ as follows: 
\begin{itemize}
\item $3|X|$ edges between $X$ and $P\cup Q$.
\item $2(|R\cup C| - |X|)$ edges between $(R\cup C)\setminus X$ and $P$.
\item $\binom{|R\cup C|}{2} - \frac{|X|\cdot (|R\cup C| - 4)}{2}$. 
\end{itemize}
Which sums to $2|R\cup C| + |X| + \binom{|R\cup C|}{2} - \frac{|X|\cdot (|R\cup C| - 4)}{2}$. Since we already upper bounded the number of edges in $S$ that are not incident on a vertex in $R\cup C$ by $|P|+|Q|$, we get from $|S|\ge \ell$ that $|X| - \frac{|X|\cdot (|R\cup C| - 4)}{2} = \frac{|X|\cdot (6-|R\cup C|)}{2}  \ge 0$ and since we can assume that $|R\cup C| = 2k^2- 2k > 6$, we get that $|X| = 0$. Therefore, for every vertex $v$ in $R\cup C$, there are at most two edges incident on $v$ and a vertex in $P$ and there is no edge in $S$ incident on $v$ and a vertex in $Q$.

From above discussion it follows that $S$ contains:
\begin{itemize}
\item $|P|+|Q|$ edges from $H[P\cup Q\cup M]$.
\item All the edges from the clique on $R\cup C$. 
\item For every $v\in R\cup C$ exactly two edges with one endpoint $v$ and the other endpoint in $P$.
\end{itemize} 
First observe that if for $p\in P$, the edge $pm(v)$ is not in $S$, then $S$ contains two edges with $p$ as the endpoint. Moreover for every $p\in P$, there is only one edges in the neighborhood of $p$ restricted to $P\cup Q\cup M$. In particular if $p\in P_{i,j}$, this edge is $q^{i,j}_1q^{i,j}_2$. Since for $p, p'\in P_{i,j}$ there is no edge between $p$ and $p'$, It follows that for every $i,j\in [k]$, there is at most one vertex $p^{i,j}_{s_{i,j}}\in P_{i,j}$ such that $S$ contains edges $p^{i,j}_{s_{i,j}}q_1^{i,j}, p^{i,j}_{s_{i,j}}q_2^{i,j}$, and $q_1^{i,j}q_2^{i,j}$ and for all $p \in  P_{i,j}\setminus \{p^{i,j}_{s_{i,j}}\}$ the strong set $S$ contains $pm(p)$. 

We now show that there is not only at most one, but precisely one such vertex $p^{i,j}_{s_{i,j}}$ and selecting the pair $s_{i,j}$ from $S_{i,j}$ for each $i,j\in [k]$ gives us a solution for the original \textsc{Grid Tiling} instance, which concludes the proof. 

Let us consider $r^{i,j}\in R$ for $i\in [k-1]$ and $j\in [k]$. As we already argued there are two vertices $p_1,p_2\in P$ such that $r^{i,j}p_1, r^{i,j}p_2\in S$. First notice that $m(p_1)$ nor $m(p_2)$ are adjacent to $r^{i,j}$. It follows $p_1m(p_1)$ and $p_2m(p_2)$ are not in $S$ and so for $\alpha\in \{1,2\}$, $p_\alpha$ is the unique vertex $p^{i',j'}_{s_{i',j'}}\in P_{i',j'}$, for some $i',j'\in [k]$, such that $S$ contains edges $p^{i',j'}_{s_{i',j'}}q_1^{i',j'}, p^{i',j'}_{s_{i',j'}}q_2^{i',j'}$(, and $q_1^{i',j'}q_2^{i',j'}$). However $r^{i,j}$ is adjacent only to $q_1^{i,j}, q_2^{i,j}, q_1^{i+1,j}, q_2^{i+1,j}$. Hence $\{p_1, p_2\} = \{p^{i,j}_{s_{i,j}}, p^{i+1,j}_{s_{i+1,j}}\}$. Moreover, $r^{i,j}p_1, r^{i,j}p_2\in S$ implies that $p_1p_2\in E[H]$ and so $s_{i,j}=(a,b)$ and $s_{i+1,j}=(a',b')$ satisfy the condition (C1) that $a=a'$ for all $i\in [k-1]$ and $j\in [k]$.

Analogously, if we consider $c^{i,j}\in C$ for all $i\in [k]$ and $j\in [k-1]$, we can conclude that $S$ contains the edges $c^{i,j}p^{i,j}_{s_{i,j}}$ and $c^{i,j}p^{i,j+1}_{s_{i,j+1}}$, so $p^{i,j}_{s_{i,j}}p^{i,j+1}_{s_{i,j+1}}\in E[H]$ and $s_{i,j}=(a,b)$ and $s_{i,j+1}=(a',b')$ satisfy the condition (C2) that $b=b'$ for all $i\in [k]$ and $j\in [k-1]$.
\end{proof}

\subsection{FPT algorithm for bounded-degree graphs}\label{ssec:STC_bounded_degree}

We show that \textsc{Strong Triadic Closure} can be solved in single-exponential time in $k = \ell - \MM(G)$ for graphs of bounded degree.
This generalizes the result of
Golovach et al.~\cite{GolovachHKLP20}, who proposed an $O^*(2^{O(k)})$-time algorithm  for \textsc{Strong Triadic Closure}, specifically for the graph of maximum degree at most four.
Compared to Golovach et al., our algorithm is simpler by adopting a delta-matroid perspective.

\tcamdelta*
\begin{proof}
  Starting with an approach in the proof of Lemma~\ref{lemma:lifting},
  our algorithm uses a greedy method to find a packing $\mathcal{K}$ of $K_4$'s.
  If $k/2$ disjoint $K_4$'s have been identified, then there is a strong set of size $\ell$: 
  Letting $K = V(\mathcal{K})$ be the set of vertices involved in $\mathcal{K}$,
  there is a matching of size $\MM(G) - |K| = \MM(G) - 4|\mathcal{K}|$ disjoint from $K$.
  Together with $K_4$'s in $\mathcal{K}$, we obtain a strong set of size $(\MM(G) - 4|\mathcal{K}|) + 6|\mathcal{K}| \ge \ell$, provided that $|\mathcal{K}| \ge k/2$. 
  So we will assume that $|\mathcal{K}| < k/2$.

  Further, we infer that a strong set $E'$ of size $\ell$ contains at most $3k$ edges within $V(\mathcal{K})$.
  This allows us to enumerate all possible edges inside $V(\mathcal{K})$ that are part of $E'$ in $\binom{\Delta |K|}{3k} = \Delta^{O(k)}$ time.
  Additionally, we can guess all edges with exactly one endpoint in $V(\mathcal{K})$ in $\Delta^{O(k)}$ time:
  For each vertex $v$, the endpoints of the edges incident with $v$ must form a clique to maintain the triadic closure property, and $G - K$ is $K_4$-free, it follows that each vertex in $K$ is incident with at most three edges in $E'$.
  So we can enumerate all these edges in $\Delta^{O(k)}$ time.
  Note that there are at most $3|K|$ vertices outside $K$ that is incident with an edge in $E'$ where exactly one endpoint is in $K$.
  For each of these vertices outside $K$, we guess at most two incident edges in $E'$.
  With all these edges guessed, we can verify whether all edges that intersect $K$ satisfy the triadic closure property.
  Specifically, there are four types of $P_3$'s $(u, v, w)$ involving a vertex in $K$:
  (i) $u, v, w \in K$, (ii) $u, v \in K$ and $w \notin K$, (iii) $u, w \in K$ and $v \notin K$, and (iv) $u \in K$ and $v, w \notin K$.
  All these relevant $P_3$'s can be enumerated, which allows us to verify the triadic closure property.

  Our next step is to identify the edges of $E'$ that lie outside $K$.
  As we have verified the triadic closure property for all $P_3$'s intersecting $K$, we need to address the \textsc{Strong Triadic Closure} problem on $G - K$ (with some edges prescribed to be part of the desired strong set).
  Suppose that we need to find a strong set of size $\ell'$, having guessed $\ell - \ell'$ edges incident with $K$ to be included in a strong set.
  By Lemma~\ref{lemma:ce-feasible-cover}, if $G - K$ has a perfect matching, the problem is to find a strong cycle packing covering a feasible set of size $2k'$, where $k' = \ell' - \MM(G - K) \le \ell - (\MM(G) - |K|) \le 3k$.
  This criterion holds regardless of the absence of perfect matching as well from the discussions in Section~\ref{sssec:2mod3}.

  To address this problem, we will adopt a probabilistic approach.
  Specifically, for each vertex in $V(G) \setminus K$, we randomly select two incident edges in $G - K$.
  For each vertex $v$ that is connected to $K$ via an edge included with $E'$, this selection includes the edges in $G - K$ predetermined to be part of $E'$.
  We keep those edges $uv$ that are chosen by both $u$ and $v$.
  The remaining edges induce a graph where every connected component is a path or a cycle.
  Edges not contributing to strong cycles are subsequently pruned.
  Our algorithm then tests the existence of a feasible set of size $2k$ within the vertices that are involved in remaining strong cycles.
  Assuming that there is a strong cycle packing that covers a feasible set of size $2k'$, the algorithms confirms this with probability $1/\Delta^{O(k)}$.
  This is because there is a strong cycle packing over at most $6k'$ vertices by Lemma~\ref{lemma:tri-maximal}.
  For each strong cycle of length $l$, the probability that it is retained is $1/\Delta^{2l}$, resulting in an overall success probability at least $1/\Delta^{O(k)}$.
  Conversely, if a feasible set can be covered by the vertices connected to remaining edges, a strong set of size $\ell'$ can be constructed by Lemma~\ref{lemma:tc-feasible-cover}.
\end{proof}

\subsection{FPT approximation}\label{ssec:STC_FPT_Approx}

In this subsection, we develop an FPT algorithm for \textsc{Strong Triadic Closure} parameterized by $k = \ell - \MM(G)$, with an approximation ratio $7$.
Specifically, we give an $O^*(k^{O(k)})$-time algorithm that either determines that there is no strong set of size $\MM(G) + k$ or finds a strong set of size $\MM(G) + k/7$.
Complementing this finding, we show the infeasibility of achieving a polynomial-time $n^{1-\varepsilon}$-factor approximation for \textsc{Strong Triadic Closure} under the assumption that P $\ne$ NP, even with parameter~$\ell$.

\tcamapprox*

\begin{proof}
  Let $c = 14$ and $t = \MM(G) + 2k/c$.
  Our algorithm will either identify a strong set of size $t$ or establish the absence of a strong set of size $\MM(G) + k$.
  First, we find a packing $\mathcal{K}$ of $K_4$'s through a greedy approach, leading to the following two case:
  \begin{itemize}
    \item 
    If $|\mathcal{K}| \ge k/c$, a strong set of size $t$ exists:
    $G - V(\mathcal{K})$ contains a matching of size $\MM(G) - 4|\mathcal{K}|$ and $K_4$'s in $\mathcal{K}$ contribute $6|\mathcal{K}|$ edges.
    \item
    Conversely, suppose that $|\mathcal{K}| < k/c$.
    We solve \textsc{Strong Triadic Closure} on $G - V(\mathcal{K})$ using the FPT algorithm for $K_4$-free graphs in Theorem~\ref{theorem:tc-am-fpt}.
    Initially, we test whether $G - V(\mathcal{K})$ has a strong set of size $t$.
    If such a set is absent, we proceed to optimally solve \textsc{Strong Triadic Closure} on $G - V(\mathcal{K})$.
    As $V(|\mathcal{K}|) < k/c$, and we can also test whether $G[V(\mathcal{K})]$ contains a strong set of size $6|\mathcal{K}|$ in $k^{O(k)}$ time.
    If this succeeds, we obtain a strong set of size $t$ as in the first case.
    Otherwise, we can find the largest strong set within $G[V(\mathcal{K})]$ in $O(k^{O(k)})$ time by enumerating all sets up to $6|\mathcal{K}|$ edges.
    Now let us consider a combine optimal solution from $G[V(\mathcal{K})]$ and $G - V(\mathcal{K})$.
    If this set has size smaller than $t$, then there is no solution of size $t + 12k/c$:
    For any strong set, every vertex in $V(\mathcal{K})$ has at most three edges that connects it to a vertex in $V(G) \setminus V(\mathcal{K})$ by the strong triadic closure property.
  \end{itemize}
  This gives a 7-factor approximation guarantee, as we either find a strong set of size $t = \MM(G) + k/7$, or conclude that there is no strong set of size $t + 12k/c = \MM(G) + k$.
\end{proof}

We conjecture that a more refined analysis, by adopting a smaller value of $c$, could improve the approximation factor.
We now turn to the hardness result.

\begin{lemma}
  For every $\varepsilon > 0$, approximating the value of $\ell$ for \textsc{Strong Triadic Subgraph} within a factor of $O(n^{1-\varepsilon})$ in polynomial time is impossible, unless P = NP.
\end{lemma}
\begin{proof}
  It is known that, for every $\varepsilon > 0$, no polynomial-time algorithm can approximate \textsc{Clique}
   within a factor better than $O(n^{1-\varepsilon})$, unless P = NP~\cite{Zuckerman07}.
   In particular, for
  every $\varepsilon > 0$, the existence of a polynomial-time algorithm that distinguishes between graphs
  with a clique of size $n^{1-\varepsilon}$ and those without 
  a clique of size $n^{\varepsilon}$, implies P = NP.
  Assume, for contradiction, the existence of a polynomial-time approximation for \textsc{Strong Triadic Closure} within a factor $O(n^{1-\varepsilon'})$ for every $\varepsilon' > 0$.
  
  Consider a $G=(V,E)$ with a clique $C$ with $|C| \geq n^{1-\varepsilon}$. 
  Then the set of edges $S=\binom{C}{2}$ is a strong set of size $\frac{1}{4} n^{2-2\varepsilon}$ for sufficiently large $n$.
  On the other hand, if $G$ has no clique of size $n^{\varepsilon}$, then there is no strong set of size $n^{1+\varepsilon}$.
  This is because for a strong set $S$, each vertex has degree $n^{\varepsilon}$ in $S$, as the neighborhood must be a clique to adhere to the triadic closure property.
  It follows that an approximation of \textsc{Strong Triadic Closure} within a factor of $\frac{1}{4}n^{1-3\varepsilon}$ would effectively differentiate these \textsc{Clique} instances, contradiction the assumption.
\end{proof}

\else
  \input{sec456_short.tex}
\fi

\section{Conclusion}
\label{sec:conclusions}

In this work, we have initiated the studies into the parameterized complexity of problems formulated over linear delta-matroids.
We have found a striking contrast to matroids:
the choice of parameterization, whether by rank or by cardinality, influences problem complexity. 
Specifically, when parameterized by rank, solutions applicable to matroids can be extended to delta-matroids,
with more involved algorithms but ultimately with the same $f(k)$-factor in the running time.
In particular, the recent method of determinantal sieving can be lifted to linear delta-matroids parameterized by rank.
It is at the moment unclear precisely how powerful this result is, but it does solve the basic
intersection and packing problems over linear delta-matroids.
In contrast, when parameterized by cardinality, we stumble on hardness in seemingly simple delta-matroids.
This underscores the significant impact of the absence of a truncation operation on their complexity.
Yet, we demonstrated that the \textsc{Delta-matroid Triangle Cover} remains FPT for delta-matroids of unbounded rank.
This was achieved by leveraging a novel algebraic technique, namely $\ell$-contraction, complemented by sunflower-like combinatorial approaches.
It is worth noting that the $\ell$-contraction operation is an extension of the matroid truncation.
Employing delta-matroid perspectives, we resolve two open questions in the literature, \textsc{Cluster Subgraph} and \textsc{Strong Triadic Closure} parameterized by above matching.
These problems have close connections to \textsc{Delta-matroid Triangle Cover}.
Specifically, when the graph is $K_4$-free and has a perfect matching, \textsc{Cluster Subgraph} parameterized by above matching is a special case of \textsc{Delta-matroid Triangle Cover}.
We were then able to give FPT algorithms for both problems on $K_4$-free graphs.
However, to our surprise, the complexity differs in the presence of $K_4$'s.
While we successfully extended the FPT algorithm to accommodate the general case for \textsc{Cluster Subgraph}, our attempts for an algorithmic solution for \textsc{Strong Triadic Closure} were met with difficulties, which led us to a W[1]-hardness proof.

Let us conclude this study with several open questions.
The first question concerns the potential for enhancing our algorithm for \textsc{Delta-matroid Triangle Cover}, which currently runs in $O^*(k^{O(k)})$ time, to $O^*(2^{O(k)})$ time.
The algebraic approach, and the simpler color-coding approach, in conjunction with dynamic programming, typically achieves single-exponential time complexity.
However, we have not been able to improve our algorithm to $O^*(2^{O(k)})$ time.
Another area worth investigating is the derandomization of our algorithm.
Our algorithm is ``heavily'' randomized, as we repeatedly apply $\ell$-projection, which relies on the Schwartz-Zippel lemma.
Given that the matroid truncation can be performed deterministically \cite{LokshtanovMPS18TALG}, it is reasonable to believe that the deterministic implementation of $\ell$-projection can be achieved.
On the other hand, a more fundamental obstacle is to derandomize the algorithm for \textsc{Colorful Delta-matroid Matching}. 

Lastly, let us discuss future research directions regarding kernelization.
Over matroids, The representative sets approach yields a polynomial kernel for, e.g., \textsc{3-Matroid Intersection} and \textsc{3-Matroid Parity}.
For delta-matroids with the rank parameter, the picture looks promising using the delta-matroid representative sets statement of Wahlström~\cite{Wahlstrom24SODA}.
Under a cardinality parameter, things look more difficult. In particular, it is completely unclear to us whether \textsc{DMM Intersection} has a polynomial kernel parameterized by cardinality. 
We remark that \textsc{Cluster Subgraph} and \textsc{Strong Triadic  Closure} do not have polynomial kernels even when parameterized by $\ell$ \cite{GolovachHKLP20,GruttemeierK20}, under standard complexity assumptions.
More broadly, is there a further application of representative sets statement of Wahlström~\cite{Wahlstrom24SODA}? 
Though there is no truncation operation for delta-matroids, it is plausible that the $\ell$-projection unlocks the potential of delta-matroid representative sets approach for kernelization.

\bibliographystyle{abbrv}
\bibliography{all}

\end{document}